\documentclass[12pt,final]{article}
\usepackage[T1]{fontenc}
\usepackage[utf8]{inputenc}
\usepackage[ngerman,american]{babel}
\usepackage{todonotes}
\usepackage{color}
\usepackage{xcolor}
\usepackage{soul}
\usepackage{graphicx}
\usepackage{subfig}
\usepackage{verbatim}
\usepackage{float}
\usepackage{amsopn}
\usepackage{amsmath}

{\begin{quote}\begin{tiny}\it}%
		{\end{tiny}\end{quote}}
\newcommand{\Cone}{C_1}
\newcommand{\rhozero}{\rho}
\newcommand{\ezero}{e_0}

\newcommand{\state}[2]{\langle {#1},{#2}{#1}\rangle}

\newcommand{\tinyyboxcontribution}{\mathcal{L}}
\newcommand{\Ca}{C_1}
\newcommand{\Cb}{C_2}

\newcommand{\Mloc}{\mathcal{M}}
\newcommand{\MMloc}{\mathcal{M}}
\newcommand{\quadpar}{\sigma}
\newcommand{\Rexponentialendpoint}{\eta}
\newcommand{\Rexponentialendpointvalue}{\mfr{1}{30}}
\newcommand{\locpsi}{\widetilde{\psi}}
\newcommand{\aprioripsi}{\Psi}
\newcommand{\LHYConstant}{\frac{128}{15\sqrt{\pi}}}

\newcommand{\Rexponential}{\mu}

\newcommand{\noell}{{}}
\newcommand{\positivepart}[1]{\off{{#1}}_+^2}

\newcommand{\simpleto}{\mapsto}
\newcommand{\epsilonTconst}[2]{{#2}}
\newcommand{\theepsilonTconst}{\delta}

\newcommand{\bigbox}{\abs{B}}
\newcommand{\intt}[2]{\int_{{#1}}^{{#2}}}
\newcommand{\inttt}[2]{\overset{{#2}}{\underset{{#1}}{\int}}}
\newcommand{\innerp}[2]{\langle {#1},{#2}\rangle}
\newcommand{\innerpp}[3]{\langle {#1}\vert {#2} \vert {#3}\rangle}
\newcommand{\innerppa}[3]{\langle {#1}, {#2}  {#3}\rangle}

\newcommand{\summ}[2]{\overset{{#2}}{\underset{{#1}}{\sum}}}

\newcommand{\limit}[1]{\underset{{#1}}{\lim}}

\newcommand{\dx}{\, \mathrm{d}x}
\newcommand{\dy}{\, \mathrm{d}y}

\newcommand{\dk}{\, \mathrm{d}k}
\newcommand{\dd}{\, \mathrm{d}}
\newcommand{\nn}{\nonumber}
\newcommand{\norm}[1]{\vert \vert {#1} \vert \vert }
\newcommand{\abs}[1]{\lvert {#1} \rvert }
\newcommand{\labs}[1]{\of{#1}}

\newcommand{\of}[1]{\left ( {#1}\right )}
\newcommand{\off}[1]{\left [ {#1}\right ]}
\newcommand{\offf}[1]{\left \{ {#1}\right \}}
\newcommand{\ad}{^{\ast}}
\newcommand{\inv}{^{-1}}

\newcommand{\SFACTORED}{(\sqrt{\rhozero a}ds\ell)^{-3}}
\newcommand{\SFINAL}{\mathcal{S}}
\newcommand{\Quaderror}{\mathcal{E}_{\mathrm{Quad}}}
\newcommand{\FF}{\mathcal{F}}

\newcommand{\MM}{\mathcal{M}}
\newcommand{\NN}{\mathcal{N}}

\usepackage{marginnote}
\usepackage{amsmath,amsfonts,amsthm,amssymb}
\usepackage[final,pdfusetitle, pdftex,
pdfauthor={Birger Brietzke and Jan Philip Solovej},
]{hyperref}
\usepackage{amstext}
\usepackage{bbold}
\allowdisplaybreaks 
\usepackage{setspace}  

\usepackage{showkeys}
\theoremstyle{plain}
\newtheorem{thm}{THEOREM}[section]
\newtheorem{lm}[thm]{LEMMA}
\newtheorem{cl}[thm]{COROLLARY}
\newtheorem{prop}[thm]{PROPOSITION}
\theoremstyle{definition}

\theoremstyle{remark}

\newcounter{condition}
\newenvironment{condition}{\par\refstepcounter{condition}\smallskip\noindent
	{\bf CONDITION \thecondition:}\begingroup \it}{\endgroup\par} 
\newcounter{assumption}
\newenvironment{assumption}{\par\refstepcounter{assumption}\smallskip\noindent
	{\bf ASSUMPTION \theassumption:}\begingroup \it}{\endgroup\par} 
\newcommand{\upchi}{\raise1pt\hbox{$\chi$}}
\newcommand{\R}{{\mathord{\mathbb R}}}
\newcommand{\C}{\mathbb{C}}
\newcommand{\Z}{{\mathord{\mathbb Z}}}

\newcommand{\mfr}[2]{{\textstyle\frac{#1}{#2}}}

\newcommand{\hcT}{\mathord{\widehat{\cal T}}}
\newcommand{\htheta}{\mathord{\widetilde\theta}}
\newcommand{\hchi}{\mathord{\widetilde\chi}}
\newcommand{\hQ}{\mathord{\widetilde Q}}

\newcommand{\hB}{\mathord{\widetilde B}}

\newcommand{\bn}{{\mathord{b}}^{\phantom{*}}}

\newcommand{\cA}{{\mathord{\cal A}}}
\newcommand{\cB}{{\mathord{\cal B}}}
\newcommand{\cT}{{\mathord{\cal T}}}
\newcommand{\cN}{{\mathord{\cal N}}}
\newcommand{\cF}{{\mathord{\cal F}}}
\newcommand{\cQ}{{\mathord{\mathcal Q}}}
\newcommand{\cU}{{\mathord{\mathcal U}}}
\newcommand{\la}{{\mathord{\lambda}}}
\newcommand{\one}{{\mathord{\mathbb 1}}}
\begin{document}
\title{THE SECOND-ORDER CORRECTION TO THE GROUND STATE ENERGY OF THE
  DILUTE BOSE GAS} %
\title{The Second Order Correction to the Ground State Energy of the Dilute Bose Gas}
\author{Birger Brietzke$^{1,\ast}$ \qquad Jan Philip Solovej$^{2,}$\thanks{This work was done in Copenhagen and was partially supported by the Villum Centre of Excellence for the Mathematics of Quantum Theory (QMATH) and the ERC Advanced grant 321029.}}
\date{\small{\begin{otherlanguage}{ngerman}{ 1. Institute of Applied Mathematics, University of Heidelberg, brietzke@math.uni-heidelberg.de\\
Im Neuenheimer Feld 205, 69120 Heidelberg, Germany}\end{otherlanguage}}\\
{2. Department of Mathematics, University of Copenhagen, solovej@math.ku.dk\\
Universitetsparken 5, 2100 Copenhagen, Denmark}}
                                \maketitle\vspace{-1 cm}
                                \thispagestyle{empty}
                                \begin{abstract}
                                	We establish the Lee--Huang--Yang formula for the ground state energy of a dilute
                                	Bose gas for a broad class of repulsive pair-interactions in 3D as a lower bound. Our result is valid in an appropriate parameter regime of soft potentials and confirms that the Bogolubov approximation captures the right second order correction to the ground state energy.
                                \end{abstract}
\vfill\eject

\setcounter{tocdepth}{1}
\tableofcontents
\section{Introduction}
\setcounter{page}{2}
Bogolubov's 1947 paper \cite{B} laid the foundation for our present
theories of the ground states of dilute Bose gases.  His approximate
theory was intended to explain the properties of liquid Helium, but it
is expected to be most accurate in the opposite regime of a dilute gas
of particles (e.g., atoms) interacting pairwise with a repulsive
potential $v(x_i-x_j) \geq 0$. 

The simplest question that can be asked is the correctness of the
prediction for the ground state energy. This, of course, can
only be exact in a certain limit -- the `weak coupling' limit. In the
case of the charged Bose gas that one of the authors studied \cite{LS1,LS2}, in
which particles interact via Coulomb forces, the appropriate limit is the
{\it high density} limit. In this setting Bogolubov's
prediction, first elucidated in \cite{Foldy}, is correct to leading
order in the inverse density.

In gases with short range forces, which are the object of study here, 
the weak coupling limit corresponds to {\it low density} instead. The reader is 
referred to \cite{LSSY} for background information and more details.

Our system consists of $N$ three-dimensional particles in a large box 
$\Lambda$ 
of volume $|\Lambda|=L^3$. As usual, we are interested in the thermodynamic
limit $N\to \infty, \ |\Lambda|\to \infty$ with the ratio $N/|\Lambda|\to \rho$. 
The Hamiltonian is
\begin{equation}\label{ham} 
H_N = -\nu \sum_{j=1}^N \Delta_j  \ + \sum_{1\leq i<j\leq N}v(x_i-x_j),
\end{equation}
with $\nu = \hbar^2/2m$, where $m$ is the mass of the particles. From Sect. \ref{sect.2} on we will set
$\nu=1$, but we leave it in place in this introduction in order to
emphasize that the scattering length of $v$ depends on $\nu$ as well
as on $v$. We assume that $v(x) \geq 0$ and that $v$ is spherically
symmetric, i.e., $v(x)=v(|x|)$.  We could assume that $v$ is
a function of sufficiently fast decay at infinity, but in order not to
overburden the paper, we assume that $v$ is of finite range, 
i.e., there is an $R_0$ such
that $v(x) =0 $ for $|x|>R_0$. The ground state energy is denoted by
$E_N$ and 
\begin{align}
e (\rho)= \lim_{L\to \infty,\ N/|\Lambda|\to \rho} E_N/|\Lambda|
\end{align}
denotes the ground state energy density.
 
The {\it scattering length} $a$ is defined by the solution of the
equation $-\nu \Delta f(x) + \frac12 v(x) f(x)=0$ that goes to 1 as
$|x|\to\infty$. Such an $f$ must satisfy $f(x)= 1-a/|x|$ for $| x|>R_0$.
If $v$ is simply a hard core repulsion of radius $r$ then $a=r$.

Bogolubov's formula for the first two terms for $e$ in a small $\rho$
asymptotic expansion is
\begin{equation}\label{eq:energy_general}
\boxed{ \ \ e(\rho) = 4\pi \nu \rho^2 a \left(1+ \frac{128}{15\sqrt{\pi}}\sqrt{\rho a^3}
+o(\sqrt{\rho a^3})  \right).   }
\end{equation}
To be more exact, the leading $4\pi \nu \rho^2 a$ term was proposed by
Lenz \cite{Lenz}. Bogolubov derived $\frac12 \rho^2 \int v(x)\dx$ for the
leading term by his method but, realizing that this could not be
correct, noticed that $\int v(x)\dx $ is the first Born approximation
to $8\pi \nu a$. The second term, while inherent in Bogolubov's work
(see \cite{boulder,LSSY}) is credited to Lee, Huang and Yang who
actually found it \cite{LHY} for the hard core gas. Again,
Bogolubov's method has $\frac12 \int v(x)\dx$ in place of $4\pi \nu a$
in this term as well. It is worth noting that the naive
perturbation result, $\frac12 \int v(x)\dx$, does not depend on $\nu$,
which is absurdly incorrect. Other derivations that do not use
Bogolubov's setup or the momentum space formulation exist \cite{L1},
but no rigorous derivation of this second term other than \cite{GS}, which we discuss below, exists so far.\\
\indent The first term $4\pi \nu \rho^2 a $ was attacked rigorously by Dyson
\cite{D2} for the hard core case; he proved a variational upper bound
of this precise form (up to $o(\rho^2)$), as well as a lower bound that,
unfortunately, was 14 times too small. He also formulated an
inequality that gives a lower bound for the expectation value of $v$
in terms of that of a longer range, softer potential. This inequality
has been used in most subsequent rigorous investigations. In
particular, it was essential in the paper \cite{LY} that finally
proved that $4\pi \nu \rho^2 a $ is the correct leading term in
three-dimensions for any $v\geq 0$ and finite range -- including the
hard core case.

Our focus here is on the second term. From Bogolubov's perspective it
is a correlation effect and in his derivation it presupposes
Bose-Einstein condensation (BEC) in the ground state. But the fact
that his method gives the correct second term in {\it one}-dimension
\cite{LL}, despite the fact that there is no BEC in one-dimension,
suggests that BEC may not be completely relevant for this problem.

The second term is not merely a perturbation of the first, for it
involves new physics. The mean particle spacing is $\rho^{-1/3}$, which
is much greater than the size of a particle (which we may take to be
$a$, not $R$).  The uncertainty principle\label{heuristic argument} tells us that the energy
per particle,
which is 
$4\pi\nu\rho a$, defines a length $\lambda = (\rho a)^{-1/2} \gg
\rho^{-1/3}\gg a$ below which a particle cannot be localized without
seriously altering its energy.  Thus, it is totally impossible to
think of individual particles in the gas; their wave-functions overlap
considerably.  We can also think of $\lambda$ as the wavelength of the
disturbance caused by dropping a particle of size $a\ll\rho^{-1/3}$
into the 'sea' of particles.  The energy of this very long (on the
scale of $\rho^{-1/3}$) wave, relative to the main term $\nu\rho^2 a$, can
be understood heuristically from the perturbation it causes in the
scattering solution resulting in a change of the
(two-particle) density, $ \rho \to \rho (1+O(a/\lambda))$. This
gives rise to an energy shift $\rho^2 a \to \rho^2 a (1+O(a/\lambda)) \sim
\rho^2 a (1+O(\sqrt{\rho a^3})) $.

Until recently it seemed impossible to go beyond the methods of
\cite{LY} to derive the second term in \eqref{eq:energy_general} rigorously. Giuliani and Seiringer
\cite{GS} have made an important step forward in this quest, however,
by considering a situation in which a soft potential $v$ gets fatter
and thinner as $\rho \to 0$ in such a way that $\int v$ is kept
constant. In this limit the second Born approximation to the
scattering length is of order $v^2$, and if $v$ is soft enough one can
hope for sufficient accuracy to achieve $4\pi\nu\rho^2 a$ in place of
Bogolubov's $\frac12 \rho^2 \int v(x)\dx$. With the leading term
sufficiently under control one can then hope to see the $\sqrt{\rho a^3}$
term.  In this approach, in which the length scale of the potential is
adjustable, the potential has the form
\begin{equation}\label{eq:vscale}
v_R(r)  = R^{-3} v_1(r/R)     \qquad {\rm with}\  R\to \infty\ {\rm as} \ \rho\to 0.
\end{equation}
Here $v_1$ is a fixed, bounded and sufficiently smooth function with
finite (dimensionless) support which we take to be the unit ball such that $v_R$ has \emph{range} $R$. In \cite{GS},
$v_1(r) = a_0 e^{-r}$.

The interesting question is how $R$ depends on $\rho$ as $\rho \to 0$.
In \cite{GS} they take it to be $R\sim \rho^{-1/3 - 7/46} $, which implies
that each particle 'sees' infinitely many others via the interaction.
That is why \cite{GS} has ``high density'' in the title, even
though the gas is dilute ($a \ll \rho^{-1/3}$), and the leading
term is still $4\pi \nu \rho^2 a$. Nevertheless, this is the first time
that the famous $128\sqrt{\rho a^3}/15\sqrt{\pi}$ term was seen
rigorously as both a lower and an upper bound.  Partly relying on the
ideas in \cite{LS1} they achieve a proof of \eqref{eq:energy_general}. For the upper bound a variational trial state, following \cite{GA}, was used. 
Recent progress on the ground state energy and the excitation spectrum for confined gases in a translation invariant setting was made in the series of papers \cite{BBCS, BBCSII, BBCSIV, BBCSIII}.\label{Schlein references}\\
\indent Upper bounds corresponding to \eqref{eq:energy_general} were also established in \cite{ESY}, \cite{NRSII} respectively \cite{YY}, and the latter was extended to higher dimensions in \cite{Aaen}.\\
\indent Our goal is to improve the situation concerning the lower bound a bit. While we still utilize the
scaling in (\ref{eq:vscale}), our $R$ will also be allowed to be {\it less than} $\rho^{-1/3}$,
which is closer to the physical situation; a particle rarely 'sees'
another one now. This will require improving the methodology of
\cite{LS1}.
In addition we shall allow for a large class of $v_1$. 

In the accompanying paper \cite{BBSFJPS} we give a second order lower bound on the ground state energy for the unscaled setting ($R=a$), which is consistent with \eqref{eq:energy} but does not capture the sharp constant.

We use the convention 
$$
\widehat f(p)=\int_{\R^3}e^{-ipx}f(x) \dx
$$
for the Fourier transform. Our assumptions on $v_1$ and $R$ are the following:
\begin{assumption}\label{V1 assumption} We require that $v_1$ is non-negative, spherically symmetric, continuous with compact support within the unit ball, in particular, $v_1\in L^{1}(\R^3)$, and satisfies $v_1(0)>0$.
\end{assumption}
\begin{assumption}\label{R assumption}
We require, with $\Rexponentialendpoint<\Rexponentialendpointvalue$,
\begin{align}
\lim_{\rho\to0}\frac{R}{a}(\rho a^3)^{\frac{1}{2}}=0\qquad \text{and}\qquad
\lim_{\rho\to0}R\rho^{1/3}(\rho a^3)^{-\Rexponentialendpoint}=\infty.\label{R conditions}
\end{align}
\end{assumption}
Our main result is:
\begin{thm}\label{thm:main} 
	Consider a Bose gas with Hamiltonian (\ref{ham}) where $v$ is replaced
	by $v_R$ given in (\ref{eq:vscale}), with $v_1$ as in Assumption~\ref{V1 assumption} and $R$ as in Assumption~\ref{R assumption}.
	Then, after taking the thermodynamic limit, the energy density, $e(\rho)$, is bounded below by
	\begin{equation}\label{eq:energy}
    e(\rho) \geq 4\pi \nu \rho^2 a \left(1+ \frac{128}{15\sqrt{\pi}}\sqrt{\rho a^3}
		+o(\sqrt{\rho a^3})  \right) .
	\end{equation}
\end{thm}
Our error terms will depend on the dimensionless quantity 
$
a^{-1}\int  v_1=a^{-1}\int  v_R
$.\\
{\it Remark on the Born approximation:} The `Born approximation' or 'Born series', is a formula for 
$a$ as a power series in $v/8\pi\nu$. 
\begin{equation}\label{eq:borseries}
a= (8\pi\nu)^{-1}  \int_{\R^3}  v(x)\dx -\sum_{k=2}^\infty (-8\pi\nu)^{-k} 
\int_{\R^3} ({\mathcal L}_v )^{k-1}(v)(x) \dx \, =:\sum_{k=1}^\infty
\nu^{-k}  a_k ,
\end{equation}
where ${\mathcal L}_v $ is the operator given by ${\mathcal L}_v (g)(x) = v(x) \int_{\R^3}
|x-y|^{-1} g(y)\dy $. If each term in the series is finite for a given $v$, then,
upon replacing $v$ by $v_R$ as in (\ref{eq:vscale}), the $k^{\rm th}$
term in the sum, $\nu^{-k}a_k$, will be proportional to $R^{1-k}$. Thus, if the series converges for
some $R$, it will converge for all larger $R$. Convergence will hold for
large enough $R$ if $v\in L^1\cap L^\infty$, but milder conditions suffice.\\
\indent With the restrictions on $v_1$ in Assumption~\ref{V1 assumption} we have that the Born series for $a$ converges and may therefore write
\begin{align}\label{a_1+a_2}
	a=a_1+a_2+O(R^{-2})=(8\pi)\inv\widehat{v_R}(0)-(4\pi)\inv(2\pi)^{-3}\int\frac{1}{4}\widehat{v_R}(k)^2\abs{k}^{-2}\dd k+O(R^{-2}),
\end{align}
where we have used that $\int\frac{\abs{\widehat{v}_{1}(k)}^2}{\abs{k}^2}\dd k
=2\pi^2\int \int \frac{v_{1}(x)\overline{v_{1}(y)}}{\abs{x-y}}\dd x \dd y
$. 
The higher order corrections to $a$ will give contributions to $e$ that are higher
order than the term we seek, namely $\rho^2 a \sqrt{\rho a^3}$.
For an estimate on the error term in \eqref{eq:energy} and some additional background we refer to \cite{BB}.\label{ref to thesis}
\section{Background Potential and Chemical Potential}\label{sect.2}

In order to utilize the technical advantages of second quantization,
it is convenient for us to work in Fock space ${\cF}$ (where
$N$ takes all values $\geq 0$). On Fock space we
introduce a Hamiltonian $H_{  \rhozero  }$  that depends on a (density) parameter $  \rhozero  $. 
Its action on the $N$-particle sector of Fock space is given by
\begin{align}
H_{  \rhozero  ,N} =  \sum_{j=1}^N \left(-\Delta_j  -  \rhozero   \int\limits v_R(x)\dx\right)  + 
\sum_{1\leq i<j\leq N}v_R(x_i-x_j) +\frac12  \rhozero  ^2\abs{\Lambda}\int\limits v_R(x)\dx.\label{background Hamiltonian}
\end{align}
The parameter $\nu =1$ from now on. The box for the particles is
$\Lambda =[-L/2,\, L/2] \in \R^3$ with Dirichlet boundary conditions. 
The introduction of the parameter $  \rhozero  $ is equivalent to the more common
grand canonical approach of adding a term $-\mu N$ with the    
chemical potential being $\mu=  \rhozero  \int v_R$. The last term in \eqref{background Hamiltonian} is simply a constant
depending on $  \rhozero  $ and we add it for convenience.
It is well known \cite{DW} that this grand canonical formulation is equivalent to the canonical description
(fixed $N$) that we started with, but we will not use this
fact.
In this paper we will focus on the background Hamiltonian $H_{  \rhozero  }$
and its
thermodynamic ground state energy density
$$
\ezero(  \rhozero  )=\lim_{|\Lambda|\to\infty}|\Lambda|^{-1}\inf_{\Psi\in\cF,\|\Psi\|=1}\langle\Psi|H_{  \rhozero  }|\Psi\rangle.
$$
The ground state energy of $H_{  \rhozero  }$ is of course the same as the
ground state energy of $H_{  \rhozero  ,N}$ minimized over $N$. 
The Fock space may seem irrelevant since $H_{  \rhozero  }$ conserves
particle number, but will be convenient as we shall be localizing to regions where the particle number is not a priori known.
The main goal of our analysis on Fock space is to provide the lower bound on $\ezero(  \rhozero  )$ in Theorem~\ref{thm:main1}, which we prove in Section~\ref{sec: proof of Theorem 2.1}. As we will show momentarily, this implies our main result, Theorem~\ref{thm:main}.
\begin{thm}[Ground state of background Hamiltonian]\label{thm:main1}~\\
Let $v_1$ satisfy Assumption~\ref{V1 assumption} and let $R$ satisfy Assumption~\ref{R assumption}.
The thermodynamic ground state energy density 
of $H_{  \rhozero  }$ then satisfies the
asymptotics, as $\rhozero \to0$,
\begin{align}
\ezero(  \rhozero  )\geq 4\pi\rhozero ^2
\left(a_2+\frac{128}{15\sqrt{\pi}}a(\sqrt{\rho a^3}+o(\sqrt{\rho a^3}))\right).\label{main1:lower bound}
\end{align}
Here $a$ is the scattering length of $v_R$ and $a_2$ is the second
term in the Born series (\ref{eq:borseries}) for $a$.
\end{thm}
We will now prove the main result Theorem~\ref{thm:main} from
Theorem~\ref{thm:main1}.\label{proof of main result}
\begin{proof}[Proof of Theorem~\ref{thm:main}]\label{proof: Main Theorem}
By choosing a trial state for $H_{\rhozero, N}$ corresponding to the ground state for $H_N$, we obtain in the thermodynamic limit that
\begin{align*}\label{eq:ee0new}
e(\rho)&\geq e_0(\rho)+\underset{N/\abs{\Lambda}\to \rho }{\underset{\abs{\Lambda}\to \infty}{\lim}}\rho \frac{N}{\abs{\Lambda}}\int v_R-\frac12\rho^2\int v_R={}e_0(\rho)+\frac12\rho^2\int v_R.
\end{align*}
If we recall that $\int v_R=8\pi a_1$, 
we find from
Theorem~\ref{thm:main1} that as $\rho\to0$
\begin{align*}
e(\rho)&\geq 4\pi\rho^2
\left(a_2+\frac{128}{15\sqrt{\pi}}a(\sqrt{\rho a^3}+o(\sqrt{\rho a^3}))\right)+4\pi
a_1\rho^2\\
&=4\pi\rho^2a
\left(1+\frac{128}{15\sqrt{\pi}}\sqrt{\rho a^3}+o(\sqrt{\rho a^3})\right),
\end{align*}
where we have used that $\frac{a_1+a_2}{a}=1+O(\frac{a^2}{R^2})=1+o(\sqrt{\rho a^3})$.
\end{proof}
{\it \indent Notation:} In our setup the ratio of the scattering length $a$ to
	$\int v_R$ is bounded above and below by constants. In all our
	error bounds there is therefore no point in distinguishing between $\int v_R$ and $a$. 
	We choose to write the estimates in terms of $a$.\\
The rest of the paper will be devoted to developing tools enabling us to prove Theorem \ref{thm:main1} in Section~\ref{sec: proof of Theorem 2.1}.\\
\indent In Section~\ref{sec:loc} we present a double localization procedure for the kinetic and the potential energy. The double localization should be understood as a method allowing us to estimate the energy on the larger length scale using results obtained on the smaller length scale.
In Section~\ref{Energy in a Single Box} we decompose our Hamiltonian into a number of terms by differentiating how it acts on excited particles and particles in the (box) condensate. Then we discuss how the formalism of second quantization can be used to estimate the quadratic part of the Hamiltonian. Effectively this means that we use a Bogoliubov diagonalization. In Section~\ref{sec: small box} we use a bootstrap argument to first obtain control over the number of particles in a small box, which then leads to a lower bound on the energy. 
Since the result from the small box is not strong enough to prove our main theorem directly, we transition to the large box in Section~\ref{sec: large box}. 
On the large box we continue to bootstrap to control the particle number and introduce Theorem~\ref{thm:Localizing large matrices}, Lemma~\ref{lm:d_1 and d_2} and Lemma~\ref{lm:Q3 bound} to obtain improved control over relevant error terms. Finally, in Section~\ref{sec: proof of Theorem 2.1} we put the pieces together and give a proof of Theorem~\ref{thm:main1}, which, as we already showed, implies our main result, Theorem~\ref{thm:main}.

\section{Localization}\label{sec:loc}
As usual in the rigorous theory of the ground state energy of the Bose
gas we find it necessary to localize the particles into boxes of a
certain definite size. This achieves two goals. One is the control of
the local fluctuations in particle number and the other is to create a
gap in the spectrum of the kinetic operator, which allows us to assert
that most particles are in the lowest state of the kinetic energy
operator, i.e., they are effectively Bose-condensed on the scale of the
box. Alas, this does not allow us to prove Bose-Einstein condensation in the
thermodynamic limit, but for the purpose of computing the ground state
energy local condensation suffices. 

Because there are several length scales, it will turn out to be necessary to localize twice into boxes of two different sizes. The physical length scales of the problem that we are interested in are the following
\begin{equation}\label{eq:scales}
a\ll R\ll \rho^{-1/3}\ll (\rho a)^{-1/2}
\end{equation}
and these have the following interpretations:
\begin{itemize}
\item $a$ is the scattering length of the two-body potential, $v_R$.
\item $R$ is the range of the potential in case it has compact support and in general it describes the length scale on which the potential vanishes. In our treatment $R$ will be required to be much larger than $a$.
\item $\rho^{-1/3}$ is the mean particle spacing.
\item $(\rho a)^{-1/2}$ is the distance determined by the uncertainty
  principle given that the energy per particle is approximately $\rho
  a$. In other words if one throws a particle into the gas it makes a
  splash of size $(\rho a)^{-1/2}$. In fact, $(\rho a)^{-1/2}$, sometimes called the healing length, is the typical distance between the particles in the virtual pairs in the Bogolubov Theory. Momenta of the order of $(\rho
  a)^{1/2}$ are responsible for the second term in \eqref{eq:energy}.
\end{itemize}
The theorem that we prove includes a bigger range than indicated by
(\ref{eq:scales}). If we write $R=\rho^{-1/3}Y^\Rexponential$ then, as
stated in Theorem~\ref{thm:main}, $\Rexponential$\label{Gamma condition}
can range in $(-\frac{1}{6},\eta)$ with $\Rexponentialendpoint=\Rexponentialendpointvalue$.

The box sizes we are concerned with for localization are $\ell$ and
$d\ell$, where $d\ll1$ in such a way that 
$\ell\gg(\rho a)^{-1/2}\gg d\ell\gg\rho^{-1/3}$. Below we will also 
introduce a small parameter $s>0$ and the length scales 
$s\ell$ and $ds\ell$. This will give 
the complete list of length scales\footnote{Note that we allow but not require $R$ to be smaller than $\rho^{-1/3}$.}
\begin{equation}\label{eq:scales2}
a\ll R, \rho^{-1/3}\ll ds\ell\ll d\ell\ll(\rho a)^{-1/2}\ll s\ell\ll\ell.\\
\end{equation}
Although in Theorem~\ref{thm:main1} we also allow $R$ to be much larger than
$\rho^{-1/3}$, the physically interesting case is, of course, 
$R\ll \rho^{-1/3}$. To be precise, we will in the rest of the paper assume that the
following conditions are satisfied.
\begin{condition}\label{cond:1} There is a sufficiently small constant
  $0<\delta<1$ (to be specified in the course of the paper) such that 
  $a,R,s,d,\ell,\rho>0$ satisfy
\begin{align*}
a/R&< \delta, &\rho a^3&<\delta, &\rho^{-\frac{1}{3}}(ds\ell)\inv&<\delta ,& s<\delta,\\
d\ell(\rho a)^{1/2}&< \delta,&R(ds\ell)\inv&<\delta, &(\rho a)^{-1/2}(s\ell)\inv&<\delta.
\end{align*}
In particular $d<s\delta^2$.
\end{condition}
More (and stronger) conditions will be added later. As explained, $\delta$ will be chosen in the course of the paper. It will 
depend on $v_1$ and on the integer $M$, which we introduce in 
\eqref{zeta new and eq:defchi} 
below. The integer $M$ will however be chosen at the
end and then $\delta$ really depends only on $v_1$.

For $u\in\R^3$ we introduce the notation $\Gamma_u=u+[-1/2,1/2]^3$ for
the unit cube centered at $u$.
There are three kinds of boxes to be considered. The first is
$B(u)=\ell\,\Gamma_u$, which is a cube of side length $\ell$ and center
$\ell u$. The second kind is the smaller cube
$\widetilde{B}(u')=(d\ell)\,\Gamma_{u'}$ of side length $d\ell$ and
center $d\ell u'$. Finally, we have the rectangles
$B(u,u')=B(u)\cap \widetilde{B}(u')$, which occur when the smaller box is only partially inside the larger box. 
The second kind is really just a special case of the third kind, so we
will not introduce a name for it at this point. Generically, we will let $B$ denote any of these boxes. We denote the side lengths of $B$ by $\la_1\leq\la_2\leq\la_3$.\\
We now introduce a localization function $0\leq \chi\in C^{M}_0(\R^3)$ 
where $M$ is an integer to be determined in this paper and finally chosen in Lemma~\ref{lm: Chosen parameters are allowed}. We let
\begin{align}\label{zeta new and eq:defchi}
\zeta(y)=\left\{\begin{array}{ll}\cos(\pi y),&\hbox{if } \abs{y}\leq 1/2\\
0,&\hbox{if } \abs{y}\geq 1/2
\end{array}\right . \qquad \textrm{and define} \qquad 
\chi(x)=C_M\big (\,\zeta(x_1)\,\zeta(x_2)\,\zeta(x_3)\,\big )^{M+1}.
\end{align}
We choose $C_M$ such that $\int\chi^2=1$ and note that $\max \chi =\chi(0)=C_M$.\\
\indent It is important not to choose $\chi$ to be infinitely differentiable since the proof of Lemma~\ref{lm:wBUB} exploits that $\zeta$ is concave on its support.
In the following constants may depend on $M$, but we shall mostly omit this fact.

For $u\in\R^3$ we write \label{localization function for the big box}$\chi_u(x)=\chi(x\ell^{-1}-u)$ for the
localization function corresponding to the box $B(u)$.

The localization function for the box
$B(u,u')$ is $\chi_u(x)\chi_{u'}(x/d)$.  We  introduce the
notation 
\begin{equation}\label{eq:chiB}
  \chi_B(x)=\left\{\begin{array}{ll}\chi_u(x),&\hbox{if } B=B(u)\\
     \chi_u(x)\chi_{u'}(x/d),&\hbox{if } B=B(u,u')
 \end{array}\right.\qquad 
\end{equation}
and note that
\begin{align}
	\int\chi_u^2(x)\dd u=1,\qquad\int\chi_u^2(x)\dx=\ell^3\qquad\textrm{and}\qquad \int \chi^2_{B(u,u')}(x)\dd u'=\chi^2_{B(u)}(x).\label{Chisquared_integrals}
\end{align}
If $B\!=\!B(u,u')$ is a small box with smallest side length $\lambda_1<d\ell$, we have the bound
\begin{align}
	\max \chi_B^2\leq C\off{\of{\frac{\lambda_1}{\ell}}^{M+1}\of{\frac{\abs{B}}{(d\ell)^3}}^{M+1}}^2\leq C\of{\frac{\lambda_1}{d\ell}}^{4(M+1)},\label{eq:maxchi_B^2 bound}
	\end{align}
which becomes useful in situations where $\lambda_1$ is small.
\subsection{Localization of the Potential}
Corresponding to the interaction potential $v_R$, we define two new potentials
\begin{eqnarray}
W_{\rm b}(x)&:=&\frac{v_R(x)}{(\chi*\chi)(x/\ell)}\label{eq:defw}
\\
\noalign{\hbox{and}}
W_{\rm s}(x)&:=&\frac{W_b(x)}{[(\chi*\chi)(x/(d\ell))]}=
\frac{v_R(x)}{[(\chi*\chi)(x/\ell)]\,[(\chi*\chi)(x/(d\ell))]}\ .\label{eq:defwtilde}
\end{eqnarray}
Here the subscripts b and s refer to the size of the box (big or
small). We will mostly omit this subscript as long as the context is clear. Note that $W_{\rm s,b}$ is well defined, since by
Condition~\ref{cond:1} the range $R$ of $v_R$ is smaller than the
scaled range of $\chi$, which is at least of order $d\ell$. Thus whenever the denominator
vanishes, then the numerator is already zero.
Since $\chi\ast\chi$ is a symmetric $C^{2M}$ function and because $(\chi\ast \chi)(0)=\int\chi^2=1$, we get the estimates
\begin{align}
v_R(x)\leq W_b(x)&\leq (1+C(\mfr{R}{\ell})^2)v_R(x)\label{eq: W_b estimate},\\
v_R(x)\leq W_s(x)&\leq (1+C(\mfr{R}{d\ell})^2)v_R(x)\label{eq: W_s estimate}.
\end{align}
We introduce localized potentials
\begin{equation}\label{eq:wBu}
  w_B(x,y):=\chi_B(x)W_{\rm b,s}(x-y)\chi_B(y)=\chi_B(x)W(x-y)\chi_B(y),
\end{equation}
where $b$ is used if the box $B$ is big, i.e., of the form $B(u)$ and $s$ is
used if $B$ is small, i.e., of the form $B(u,u')$. As indicated on the
right, we will often omit b and s. 
Recall that also the form of the localization functions depends on whether the box is
big or small. 
The potential $w_B$ is localized to the box $B$.\\

Because we do not want to have to consider boxes at the boundary of $\Lambda$ throughout this paper, we introduce $\Lambda':=\Lambda+[-\frac{\ell}{2},\frac{\ell}{2}]^3$\label{def: Lambda'} and $\Lambda'':=\Lambda+[-\ell,\ell]^3$. Note that $B(u)$ intersects $\Lambda$ exactly if $u\ell\in  \Lambda'$. Replacing the last term in \eqref{background Hamiltonian} by $\frac{1}{2}\rhozero^2\abs{\Lambda''}\int v_R(x)\dd x$ does not change the ground state energy of $H_\rhozero$ in the thermodynamic limit. We may therefore use the following localization for the potential energy.
\begin{prop}[Potential localization]\label{prop:vloc}
For all $x_1,\ldots,x_N\in\Lambda$ we have 
\begin{align*}
&-\sum_{j=1}^N \rhozero  \int\limits_{} v_R(x_j-y)\dy  + 
\sum_{1\leq i<j\leq N}v_R(x_i-x_j)+\frac12\rhozero ^2\abs{\Lambda''}\int v_R(x)\dx\\
=&\inttt{\R^3}{}\of{-\sum_{j=1}^N \rhozero  \int\limits_{} \omega_{B(u)}(x_j,y)\dy  + 
	\sum_{1\leq i<j\leq N}\omega_{B(u)}(x_i,x_j)+\frac12\rhozero ^2\iint\limits_{\R^3\times\Lambda''}\omega_{B(u)}(x,y)\dx\dy }\dd u\\
\geq&\inttt{\ell\inv \Lambda'}{}\left( -\sum_{j=1}^N \rhozero  \inttt{{}}{}
\omega_{B(u)}(x_j,y)\dy  + 
\sum_{1\leq i<j\leq N}\omega_{B(u)}(x_i,x_j)+\frac12\rhozero ^2\iint\limits_{}\omega_{B(u)}(x,y)\dx\dy\right)\dd u,
\end{align*}
where $\Lambda':=\Lambda+[-\mfr{\ell}{2},\mfr{\ell}{2}]^3$ and $\Lambda'':=\Lambda+[-\ell,\ell]^3$. Moreover, for all $u\in \R^3$,
\begin{align*}
&-\sum_{j=1}^N \rhozero  \intt{}{} \omega_{B(u)}(x_j,y)\dy  + 
\sum_{1\leq i<j\leq N}\omega_{B(u)}(x_i,x_j)+\frac12\rhozero ^2\iint{}{} \omega_{B(u)}(x,y)\dx\dy\\
=&\inttt{\R^3}{}\Bigg( -\sum_{j=1}^N \rhozero  \intt{}{} \omega_{B(u,u')}(x_j,y)\dy +
\sum_{1\leq i<j\leq N}\omega_{B(u,u')}(x_i,x_j)+\frac12\rhozero ^2\iint{}{} \omega_{B(u,u')}(x,y)\dx\dy\Bigg)\dd u'.
\end{align*}
\end{prop}
\begin{proof} This follows from the identity
	$
	(\chi*\chi)(x-y)=\int\chi(x-u)\chi(y-u)\dd u
	$.
\end{proof}

The background self-energy appears so frequently that we shall denote
it $\rhozero ^2|B|^2\,\cU_B$, i.e., we introduce the
symbol
\begin{equation}\label{eq:UB}
  \boxed{\ \ \cU_B=
  \frac12|B|^{-2}\iint\limits_{}^{\phantom{\Lambda}}w_B(x,y)\dx\dy.\ \ }
\end{equation}
On the large box $B(u)$ we use that $\chi$ and $\chi\ast \chi$ are even such that for all $u\in \R^3$
\begin{align}
\cU_B=&\frac{1}{2}\ell^{-6}\iint_{\R^6} w_{B(u)}(x,y)\dd x\dd y=\frac{1}{2}\ell^{-6}\iint_{\R^6} \chi(-\frac{x}{\ell})\chi(\frac{y}{\ell})W(x-y)\dd x\dd y\nn\\
=&\frac{1}{2}\iint_{\R^6} \chi(-x-y)\chi(y)W(\ell x)\dd x\dd y=\frac{1}{2}\int_{\R^3} v_R(\ell x)\dd x = \frac{4\pi a_1}{\ell^3}.
\end{align}
In a small (possibly rectangular) box $B$, $\cU_B$ may be significantly different. It will be important to know the following facts.
\begin{lm}\label{lm:wBUB} If $B$ is either a large or a small box (side lengths $\lambda_1\leq\lambda_2\leq\lambda_3$),
there is a  constant $C$ that depends only on $M$ used in the definition
of $\chi$ and on the potential $v_1$ in (\ref{eq:vscale}) such that  
\begin{equation}\label{eq:wBUB1}
\max_x\underset{}{\int} w_B(x,y)\dy\leq \frac12 C
|B|^{-1}\iint\limits_{} w_B(x,y)\dy\dx= C|B|\,\cU_B
\end{equation}
\begin{equation}\label{eq:wBUB2}
C^{-1}\frac{a}{|B|R^3}\max\chi_B^2
\prod_{j=1}^3\min\{\lambda_j,R\}\leq \cU_B\leq C\frac{a}{R^3}\max\chi_B^2
\end{equation}
\begin{equation}
\cU_B\leq C\frac{ a}{\abs{B}}\max \chi_B^2.\label{eq:wBUB3}
\end{equation}
(The scattering length $a$, which obviously is bounded above and below
by constants depending only on $v_1$, has been included in this
inequality for dimensional reasons.)
\end{lm}
\begin{proof}
The difficult case is a (possibly) rectangular box $B=B(u,u')$, 
for the large cubic boxes $B=B(u)$ the argument is the same
just simpler.
Recall that for a small box we have
$w_B(x,y)=\chi_B(x)W_{\rm s}(x-y)\chi_B(y)$
and
$\chi_B(x)=\chi_1(x_1)\chi_2(x_2)\chi_3(x_3)$, where
$$
\chi_i(x_i)=C_M^{2/3}\zeta(|(x_i/\ell)-u_i|)^{M+1}\zeta(|(x_i/(d\ell))-u'_i|)^{M+1}
$$
for $i=1,2,3$. The function $\chi_i$ is supported on an interval $I_i$
of length $\la_i$ corresponding to a side length of the box.  We have
$0<\la_i\leq d\ell $. Let $I_i'$ denote the middle third of this
interval. 
Since $\zeta$ is positive and concave on its support, it is a straightforward exercise to check that\\
\begin{equation}
\inf_{x_i\in I_i'}\chi_i(x_i)\geq c\,\underset{x_i\in I_i}{\max}\, \chi_i(x_i)\label{eq:chiiprop}
\end{equation}
for $i=1,2,3$, where $c>0$ depends only on $M$. 
It is important here that $\chi$ is not infinitely differentiable. 

By \eqref{eq:vscale}, \eqref{eq: W_s estimate} and the fact that $v_1$ is continuous, has 
compact support and $v_1(0)>0$ there exist constants $\Ca,\Cb>0$
(depending only on $v_1$) such that 
$$
\Ca\prod_{i=1}^3{\one}_{[-\Ca R,\Ca R]}(x_i)\leq a^{-1}R^3 W_{\rm s}(x)\leq \Cb \prod_{i=1}^3
{\one}_{[-\Cb R,\Cb R]}(x_i),
$$
where ${\one}_I$ is the characteristic function of the interval
$I$.
To prove the inequality (\ref{eq:wBUB1}),
it is therefore enough to prove the 1-dimensional versions:
$$
\max_{x_i\in \R}
\int \chi_i(x_i){\one}_{[-\Cb R,\Cb R]}(x_i-y_i)\chi_i(y_i)\dy_i\leq C
\la_i^{-1}\iint \chi_i(x_i) {\one}_{[-\Ca R,\Ca R]}(x_i-y_i)\chi_i(y_i)\dy_i\dx_i,
$$
for $i=1,2,3$, where $C$ is allowed to depend only on $v_1$ and $M$.
In view of (\ref{eq:chiiprop}) this follows from
$$
\max_{x_i{\in \R}}\int_{I_i} {\one}_{[-\Cb R,\Cb R]}(x_i-y_i)\dy_i\leq C
\la_i^{-1}\iint_{I_i'\times I_i'} {\one}_{[-\Ca R,\Ca R]}(x_i-y_i)\dy_i\dx_i.
$$
This is obvious since both sides can be estimated above and below by 
constants times $\min\{\la_i, R\}$. 

The lower bound in (\ref{eq:wBUB2}) is proved in a similar fashion.
The upper bound in \eqref{eq:wBUB2} follows from $v_R(x)\leq C\frac{a}{R^3}$.\\
\indent For the bound in \eqref{eq:wBUB3} we note that $\cU_B\leq C\abs{B}\inv \int_{\R^3}v_R(x)\max \chi_B^2\dx\leq C\frac{a}{\abs{B}}\max \chi_B^2$, since $\omega_B(x,y)\leq Cv_R(x-y)\max\chi_B^2$ and that by Condition \ref{cond:1} we have $\int v_R(x)\dx\leq C a$.
\end{proof}

\subsection{Localization of the Kinetic Energy}\label{sec:kinloc}
Let $\theta$ be the characteristic function of the cubic box $[-1/2,1/2]^3$.
For $u\in\R^3$ we denote the corresponding characteristic function of
the box $B(u)$ by $\theta_u(x)=\theta(x\ell^{-1}-u)$. We shall also
use the localization function
$\chi_u(x)=\chi(x\ell^{-1}-u)$ introduced on page~\ref{localization function for the big box}.

We define the operator $Q_u$ to be the orthogonal projection on $L^2(\R^3)$
defined by 
\begin{equation}\label{eq:Qu}
  Q_uf=\theta_u f-\ell^{-3}\langle\theta_u|f\rangle\theta_u.
\end{equation}
In other words $Q_uf$ is a function in 
$L^2(\R^3)$ that is zero outside the box 
$B(u)$ and is orthogonal to
the constant functions in the box.

\begin{lm}[Abstract kinetic energy localization]\label{lm:abskinloc}
Let $\mathcal{K}:\R^3\to[0,\infty)$ be a symmetric, continuous function, which is bounded by a polynomial of degree at most $2M$, where $M$ is the integer introduced in \eqref{zeta new and eq:defchi}.
 We use it to define an operator on $L^2(\R^3)$ by
\begin{equation}\label{eq:abskinloc}
T=\int_{\R^3} Q_u\chi_u \mathcal{K}(-i\ell\nabla)\chi_uQ_u \dd u,
\end{equation}
where $\chi_u$ is considered here as a multiplication operator in
configuration space.
This $T$ is translation invariant, i.e., a multiplication operator in
Fourier space $T=F(-i\ell\nabla)$, with 
\begin{equation}\label{eq:absF}
F(p)=(2\pi)^{-3}\mathcal{K}*|\widehat\chi|^2(p)
-2(2\pi)^{-3}\widehat\theta(p)\widehat\chi*(\mathcal{K}\widehat\chi)(p)+(2\pi)^{-3}\left(\int
  \mathcal{K}|\widehat\chi|^2\right)\widehat\theta(p)^2. 
\end{equation}
In particular, we have $F(0)=0$, $F\geq 0$ and $\nabla F(0)=0$. 
\end{lm}
\begin{proof} By a simple scaling it is enough to consider $\ell=1$.
This is a straightforward calculation. Note that $Q_u$ has the integral kernel $\theta_u(y)\off{\delta(y-x)-\one}\theta_u(x)$.
If we denote by $\check{\mathcal{K}}$ the inverse Fourier transform of $\mathcal{K}$ in the sense of a tempered distribution, then the integral kernel of the operator
$
Q_u\chi_u \mathcal{K}(-i\nabla)\chi_uQ_u
$
is given by 
\begin{eqnarray*}
&&\chi_u(x)\check{\mathcal{K}}(x-y)\chi_u(y)-\chi_u(x)[\check{\mathcal{K}}*\chi_u](x)\theta_u(y)\\
&&-\theta_u(x)[\check{\mathcal{K}}*\chi_u](y)\chi_u(y)+
\theta_u(x)\langle \chi_{u}|\mathcal{K}(-i\nabla)\chi_{u}\rangle\theta_u(y).
\end{eqnarray*}

Thus the integral kernel of $\int Q_u\chi_u \mathcal{K}(-i\nabla)\chi_uQ_u \dd u $ is
given by 
\begin{align*}
([\chi*\chi]\check{\mathcal{K}})(x-y)-2\left(\chi[\check{\mathcal{K}}*\chi\right])*\theta(x-y)
+(2\pi)^{-3}\left(\int \mathcal{K}(p)\widehat\chi(p)^2\dd p\right)\theta*\theta(x-y),
\end{align*}
where we used that $\int \mathcal{K}(p)\widehat\chi(p)^2\dd p$ is finite by the choice of $\mathcal{K}$.
We arrive at the expression for $F$ by calculating the inverse Fourier
transform. The fact that $F(0)=0$ follows since $\widehat\theta(0)=\int\theta=1$ and
\begin{align*}
(2\pi)^3F(0)=2\left(\int K\widehat \chi^2\right)(1-\widehat\theta(0))^2=0.
\end{align*}
That $F\geq 0$ is a direct consequence of \eqref{eq:abskinloc} since $\mathcal{K}$ is positive. Because $F$ is differentiable it follows that $\nabla F(0)=0$.
\end{proof}With $\ell=1$ this lemma is similar to the generalized IMS localization formula
$$
\int_{\R^3} \chi_u \mathcal{K}(-i\nabla)\chi_u\dd u=(2\pi)^{-3}\mathcal{K}*|\widehat\chi|^2,
$$
where $\mathcal{K}(p)=p^2$ gives the standard IMS formula since 
$(2\pi)^{-3}\mathcal{K}*|\widehat\chi|^2=p^2+\int|\nabla\chi|^2$.

\begin{cl}\label{eq:Quav} With the same notation as above we have
that
\begin{equation}\label{eq:averagedneumangap}
\int_{\R^3}Q_u \dd u=1-\widehat\theta(-i\ell\nabla)^2,
\end{equation}
i.e., the operator $\int_{\R^3}Q_u \dd u$ is the multiplication operator
in Fourier space given by $1-\widehat\theta(\ell p)^2$.
\end{cl}
\begin{proof} Simply take $\mathcal{K}=1$ and $\chi=\theta$ in the above lemma.
\end{proof}

We will use Lemma~\ref{lm:abskinloc} for the function
$\mathcal{K}(p)=((|p|-s^{-1})_+)^2$, where $s>0$ is the parameter introduced in
Condition~\ref{cond:1}. Here $u_+=\max\{u,0\}$ denotes the positive part of $u$ and we will henceforth write $u_+^2$ instead of $(u_+)^2$. Note that
$|-i\nabla|=\sqrt{-\Delta}$.

\begin{lm}\label{lm:kinloc}
There exist constants
$C>0$ and $s^\ast>0$ (depending on the integer $M$ in the definition \eqref{zeta new and eq:defchi} of $\chi$)
such that for $0<s\leq s^{\ast}$ and any $\ell>0$ we have the inequality for all $\varphi\in H^1(\R^3)$
\begin{equation}
\langle \varphi,F_s(\sqrt{-\Delta})\varphi\rangle\geq \int\langle \varphi, Q_u\chi_u\positivepart{\sqrt{-\Delta}-(s\ell)^{-1}}\chi_u Q_u \varphi\rangle \dd u
,
\end{equation}
where
\begin{equation}\label{eq:Fs}
F_s(|p|)=\left\{\begin{array}{lr}
(|p|-\frac12(s\ell )^{-1})^2,&\hbox{if }|p|\geq \frac56 (s\ell )^{-1}\\
Cs^{M-2}p^2,&\hbox{if }|p|<\frac56 (s \ell )^{-1}
\end{array}\right.,
\end{equation}
(assuming $M\geq 2$)
\end{lm}
\begin{proof} 
We may again by a simple scaling argument assume $\ell=1$. Since we defined $\chi$ in \eqref{zeta new and eq:defchi} 
as a $C_0^M$ function, we have for $n\leq 2M$ and $C$ only depending on $M$ that
\begin{align}
\int_{\abs{q}>s\inv} \abs{q}^{n}\widehat{\chi}^2(q)\dd q\leq s^{2M-n}\int_{\abs{q}>s\inv} \abs{q}^{2M}\widehat{\chi}^2(q)\dd q\leq Cs^{2M-n}.\label{eq:q^rchi^2}
\end{align}
We use (\ref{eq:abskinloc}) and (\ref{eq:absF}) with
$\mathcal{K}(p)=\positivepart{|p|-s^{-1}}$.  
For the first term in (\ref{eq:absF}) we find
\begin{align*}
	(2\pi)^{-3}\mathcal{K}*\widehat\chi^2(p)\leq{}&(2\pi)^{-3}\int (|p-q|-s^{-1})^2\widehat\chi^2(q)\dd q\\
	\leq{}&(2\pi)^{-3}\int 
	(p^2-2pq+q^2-2s^{-1}(|p|-|q|)+s^{-2})\widehat\chi^2(q)\dd q\\
	={}& (|p|-s^{-1})^2 +C+Cs^{-1},
\end{align*}
where we have used $(2\pi)^{-3}\int q^2\widehat\chi(q)^2\dd q=
\int |\nabla\chi|^2$ and that $\chi^2$ is even.\\
If $|p|\geq\frac56s\inv$ we thus find
\begin{eqnarray}
(2\pi)^{-3}\mathcal{K}*\widehat\chi^2(p)\leq
(|p|-\frac12s^{-1})^2-\frac1{12}s^{-2}+C+Cs^{-1}.\label{eq: Ffirst}
\end{eqnarray}
For the second term in (\ref{eq:absF}) we find since
$\widehat\theta\leq 1$ that 
\begin{eqnarray}
|\widehat\theta(p)\widehat\chi*(\mathcal{K}\widehat\chi)(p)|\leq
\norm{\widehat{\chi}}_2\norm{\mathcal{K}\widehat{\chi}}_2\leq
C\left(\int_{|q|\geq s^{-1}}|q|^4|\widehat\chi(q)|^2\dd q\right)^{1/2}\leq Cs^{M-2}.\label{eq:F2nd}
\end{eqnarray}
For the third term in (\ref{eq:absF}) we have similarly
\begin{eqnarray}
|\widehat\theta(p)|^2\int \mathcal{K}|\widehat\chi|^2\leq
\int_{|q|\geq s^{-1}}|q|^2\widehat\chi(q)^2\dd q\leq
Cs^{2M-2}.
\label{eq:F3rd}
\end{eqnarray}
For $|p|\geq \frac56 s^{-1}$ we therefore have that the
function $F$ in (\ref{eq:absF}) satisfies
$$
F(p)\leq (|p|-\frac12s^{-1})^2-\frac1{12}s^{-2} +Cs^{-1}+Cs^{M-2}+Cs^{2M-2}.
$$
With $s^\ast$ sufficiently small we arrive at 
the first line in (\ref{eq:Fs}).\\
\indent We turn to the proof of the second line in (\ref{eq:Fs}). We 
know that $F(0)=\nabla F(0)=0$. The lemma follows from Taylor's formula if we can show that 
for $|p|<\frac56s^{-1}$, we have
\begin{equation}\label{eq:F2der}
|\partial_i\partial_jF(p)|\leq Cs^{M-2}.
\end{equation} 
It is straightforward to see that all second derivatives of
$\mathcal{K}(p)=\positivepart{|p|-s^{-1}}$ are bounded independently of $s$.
For the first term in (\ref{eq:absF}) we thus find for 
$|p|<\frac56s^{-1}$
\begin{align*}
|\partial_i\partial_j (\mathcal{K}*\widehat\chi^2)(p)|={}&
|(\partial_i\partial_j \mathcal{K})*\widehat\chi^2(p)|\leq C\int_{|p-q|>s^{-1}}
\widehat\chi(q)^2\dd q\\
\leq{}&C\int_{|q|>(6s)^{-1}}\widehat\chi(q)^2\dd q\leq Cs^{2M}.
\end{align*}
For the second and third term in (\ref{eq:absF}) we use the fact that for all $i,j=1,2,3$ the numbers
$$
\|\widehat\theta\|_\infty,\
\|\partial_i\widehat\theta\|_\infty,\ 
\|\partial_i\partial_j\widehat\theta\|_\infty,\ 
\int|\widehat\chi|^2,\
\int|\partial_i\widehat\chi|^2,\
\int|\partial_i\partial_j\widehat\chi|^2
$$are bounded by a constant. The same estimates that led to (\ref{eq:F2nd}) and (\ref{eq:F3rd})
then imply \eqref{eq:F2der}.
\end{proof}

\begin{cl}\label{cl:Tusimple}
If $M\geq3$ there exist constants $b>0$ and $s'>0$ (depending on the integer $M$ in the definition \eqref{zeta new and eq:defchi} of $\chi$) such that for $0<s\leq s'$ and any $\ell>0$ we have for all $\varphi \in H_0^1(\Lambda)$
$$
\langle \varphi, -\Delta\varphi \rangle \geq \int_{\R^3}\langle \varphi, Q_u\left\{\chi_u\positivepart{\sqrt{-\Delta}-\mfr12 (s\ell)^{-1}}\chi_u
  +b\ell^{-2}\right\}Q_u\varphi \rangle \dd u
.
$$
\end{cl}
\begin{proof} We again consider $\ell=1$.
Note that by Corollary~\ref{eq:Quav} and a Taylor expansion at $p=0$ we have 
\begin{equation}\label{eq:QuAv}
  \int Q_u \dd u\leq \beta^{-1}\frac{-\Delta}{-\Delta+\beta}
\end{equation}for a universal constant $0<\beta<1$. We use the previous lemma with $s$ replaced by
$2s$. We then find that 
\begin{align*}
\int_{\R^3}Q_u\chi_u\positivepart{\sqrt{-\Delta}-\mfr12 s^{-1}}\chi_uQ_u \dd u
+b\int_{\R^3} Q_u\dd u\leq F_{2s}(\sqrt{-\Delta})+b\beta^{-1}\frac{-\Delta}{-\Delta+\beta}.
\end{align*}
For $|p|< (5/12)s^{-1}$ Lemma~\ref{lm:kinloc} gives
$$
F_{2s}(p)+b\beta^{-1}\frac{p^2}{p^2+\beta}
\leq C sp^2+b\beta^{-1}\frac{p^2}{p^2+\beta}\leq 
(Cs+b\beta^{-2})p^2\leq p^2
$$
for $s$ sufficiently small. For $|p|\geq(5/12)s^{-1}$ we find from Lemma~\ref{lm:kinloc} that 
\begin{align*}
  F_{2s}(p)+b\beta^{-1}\frac{p^2}{p^2+\beta}\leq
  (|p|-\frac14 s^{-1})^2+b\beta^{-1}\frac{p^2}{p^2+\beta}
  \leq p^2-\frac5{24}s^{-2}+\frac1{16}s^{-2}+b\beta^{-1}\leq p^2,
\end{align*}
for $s$ sufficiently small.
\end{proof}
This corollary will allow us to localize the kinetic energy to
boxes. What will be left as the kinetic energy in the box $B(u)$
centered at $\ell u$, is the operator
\begin{equation}\label{eq:Tusimple}
  \cT_u=Q_u\left\{\chi_u\positivepart{\sqrt{-\Delta}-\mfr12 (s\ell )^{-1}}\chi_u+
    b\ell^{-2}\right\}Q_u.
\end{equation}
Note that $\cT_u$ vanishes on constant functions. The last term in
$\cT_u$ will control the gap in the kinetic energy, i.e., on functions
orthogonal to constants in the box, $\cT_u$ is bounded below by at
least $b\ell^{-2}$.\\
In \cite{BBSFJPS} we modify the localization of the kinetic energy, thereby avoiding an error term originating from \eqref{W^correction}, which would not be compatible with the LHY-order in the unscaled setting of \cite{BBSFJPS}. We expect that a corresponding modification in the present paper would allow us to state Assumption~\ref{R assumption} with $\Rexponentialendpoint=\Rexponentialendpointvalue$ instead of $\Rexponentialendpoint<\Rexponentialendpointvalue$.\\
\indent As explained above, we need to localize even further to smaller boxes
whose size is a factor $d\ll1$ smaller than the larger box. In these
smaller boxes we also need a term in the kinetic energy that gives a
gap. Unfortunately, the above expression \eqref{eq:Tusimple} for the
kinetic energy does not immediately allow for such further localization. For this
reason we must introduce a more complicated kinetic energy in the
larger box.

We will use
\begin{eqnarray}
  \hcT_u&=&\varepsilon_{T} (d\ell
  )^{-2}\frac{-\Delta_u^\cN}{-\Delta_u^\cN+(d\ell
    )^{-2}}+b\ell^{-2}Q_u
  \label{eq:Tauhat}\\&&+
  Q_u\chi_u\left\{(1-\varepsilon_{T})\;\positivepart{\sqrt{-\Delta}-\mfr12
    (s\ell )^{-1}}\;
    +\varepsilon_{T}\;\positivepart{\sqrt{-\Delta}-\mfr12
    (ds\ell )^{-1}}\;\right\}\chi_uQ_u,\nonumber
\end{eqnarray}
where $0<\varepsilon_T<1$ is a parameter.
The operator \label{def: Neumann Laplacian}$\Delta_u^\cN$ is the Neumann Laplacian on the box 
$B(u)$.
As usual $\Delta$ is the Laplacian on $\R^3$.
Let us be clear about the action of
$\Delta_u^\cN$ as an operator on $L^2(\R^3)$. It is the operator
associated with the
quadratic form 
$$
(f,-\Delta_u^\cN f)=\int_{B(u)}|\nabla f(x)|^2\dx,
$$ 
which is defined for all functions $f\in L^2(\R^3)$ whose restriction to
the cube is an $H^1$ function, i.e., functions for which
the above integral is finite. Note that by the operator
$(-\Delta_u^\cN+(d\ell )^{-2})^{-1}$ we mean the inverse of
$-\Delta_u^\cN+(d\ell )^{-2}$ on the space $L^2(B(u))$ extended to
all of $L^2(\R^3)$ by letting it be 
0 on the orthogonal complement, i.e., on functions that live outside 
$B(u)$.
 
Note that if $\varepsilon_T=0$ then $\hcT_u$ equals $\cT_u$.
For the kinetic energy $\hcT_u$ we have a result similar to 
Corollary~\ref{cl:Tusimple}.
For the following theorem we note that on the domain $H_0^1\of{\Lambda}$ we have
\begin{align}\label{Neumann Laplacian upper bound}
\int -\Delta^\NN_u\dd u= -\Delta_D
\end{align}
in the sense of quadratic forms.
\begin{lm}[Large-box kinetic energy localization]
\label{lm:Tu} If $M\geq 5$, and $b,d,s,\varepsilon_T>0$ are smaller than
    some universal constant, then 
\begin{equation}
  \int_{\R^3}\hcT_u \dd u\leq -\Delta.
\end{equation}
\end{lm}
\begin{proof} As usual we set $\ell=1$. 
The first step in the proof is to show that for all $d>0$ 
\begin{equation}\label{eq:Neumannineq}
\int_{\R^3}\frac{-\Delta_u^\cN}{-\Delta_u^\cN+d^{-2}}\dd u\leq
\frac{-\Delta}{-\Delta+d^{-2}}.
\end{equation}
To show this, we recall that in the sense of quadratic forms
$\sum_{u\in\Z^3}-\Delta_u^\cN\leq-\Delta$. Thus 
\begin{align}
(-\Delta+d^{-2})^{-1}\leq
\left(\sum_{u\in\Z^3}-\Delta_u^\cN+d^{-2}\right)^{-1}=\sum_{u\in\Z^3}(-\Delta_u^\cN+d^{-2})^{-1}.\label{eq: Laplace inequality}
\end{align}
The last equality looks odd, but it is just the identity 
$\left(\bigoplus_u A_u\right)^{-1}=\bigoplus_u A_u^{-1}$ applied to the
operators $A_u=-\Delta_u^\cN+d^{-2}{\one}_{L^2(B(u))}$ acting on
$L^2(B(u))$
recalling that $L^2(\R^3)=\bigoplus_{u\in\Z^3}L^2(B(u))$.

Since $1-d^{-2}(-\Delta+d^{-2})^{-1}=-\Delta(-\Delta+d^{-2})^{-1}$, it follows from \eqref{eq: Laplace inequality} that 
\begin{align}
\sum_{u\in\Z^3}\frac{-\Delta_u^\cN}{-\Delta_u^\cN+d^{-2}}\leq\frac{-\Delta}{-\Delta+d^{-2}}.
\end{align}
This will also hold if we replace the sum over $\Z^3$ by a sum over $v+\Z^3$ for any
$v\in[0,1]^3$. An integration over $v\in [0,1]^3$ gives \eqref{eq:Neumannineq}.

The second step is to observe that from Lemma~\ref{lm:kinloc}, e.g.,
with $M\geq 5$ we find that
\begin{eqnarray*}
\lefteqn{\int Q_u\chi_u\left[(1-\varepsilon_T)\positivepart{\sqrt{-\Delta}-\mfr12s^{-1}}
+\varepsilon_T\positivepart{\sqrt{-\Delta}-\mfr12(ds)^{-1}}\right]\chi_uQ_u\dd u}&&\\
&\leq&(1-\varepsilon_T)\positivepart{\sqrt{-\Delta}-\mfr14s^{-1}}
+\varepsilon_T\positivepart{\sqrt{-\Delta}-\mfr14(ds)^{-1}}
+Cs\frac{-\Delta}{-\Delta+\beta}
\end{eqnarray*}
for some universal constant $C$ and where $\beta$ is the same constant
as in \eqref{eq:QuAv}. The proof is completed using \eqref{eq:QuAv} and observing that 
if $s,b,d,\varepsilon_T$ are all smaller than some universal
constant, then, for all $p\in\R^3$,
\begin{align*}
(b\beta^{-1}+Cs)
\frac{p^2}{p^2+\beta}+\varepsilon_Td^{-2}\frac{p^2}{p^2+d^{-2}}
+(1-\varepsilon_T)\positivepart{|p|-\mfr14s^{-1}}+\varepsilon_T\positivepart{|p|-\mfr14(ds)^{-1}}
\leq p^2.
\end{align*}
\end{proof}
We now discuss the further localization into smaller boxes of
relative size $d\ll 1$. As in the previous subsection 
we index these boxes by a parameter
$u'\in\R^3$. The small box is \label{second kind}$\hB(u')=d\ell\Gamma_{u'}
=d\ell{u'}+[-d\ell/2,d\ell/2]^3$, 
whose center is at $d\ell{u'}$.
We denote the corresponding characteristic function and localization
function
by
$$
\htheta_{u'}(x)=\theta((x/d\ell)-u'),\quad
\hchi_{u'}(x)=\chi((x/d\ell)-u').
$$
The corresponding orthogonal projection onto functions orthogonal to
constants in $L^2(\hB(u'))$ is $\hQ_{u'}$ given by 
$$
 \hQ_{u'}f=\htheta_{u'}f-(d\ell)^{-3}\langle\htheta_{u'}|f\rangle\htheta_{u'}.
$$
When localizing in to the smaller boxes, we are forced to consider 
the situation of overlap between the large boxes and small boxes,
i.e.,  $B(u,u')=B(u)\cap \hB(u')$. 
The corresponding characteristic function is $\theta_u\htheta_{u'}$, 
the corresponding localization function is $\chi_u\hchi_{u'}$, and the
corresponding orthogonal projection is 
$$
Q_{uu'}f=\theta_u\htheta_{u'}f-|B(u,u')|^{-1}
\langle\theta_u\htheta_{u'}|f\rangle\theta_u\htheta_{u'}.
$$
Our first result is that when we localize
the large box kinetic energy $\hcT_u$ into smaller boxes we
will get a gap in the localized energy spectrum.
This is a consequence of the next result.

\begin{lm}\label{lm:hQuuav}With $Q_{uu'}$ as defined above we have for all $d>0$ and $\ell>0$
\begin{equation}\label{eq:Quuslide}
\int Q_{uu'}\dd u'\leq (1+\pi^{-2})\frac{-\Delta_u^\cN}{-\Delta_u^\cN+(d\ell)^{-2}}.
\end{equation}
\end{lm} 
\begin{proof} It is enough to consider $\ell=1$. 
Let $-\Delta_{uu'}^\cN$ denote the Neumann Laplacian in
the box $B(u,u')$. Observe that since
$B(u,u')\subseteq\hB(u')$, the Neumann Laplacian $-\Delta_{uu'}^\cN$
has a gap of at least $\pi^2d^{-2}$, i.e., 
$-\Delta_{uu'}^\cN\geq \pi^2d^{-2}Q_{uu'}$. Thus
$$
Q_{uu'}\leq (1+\pi^{-2})\frac{-\Delta_{uu'}^\cN}{-\Delta_{uu'}^\cN+d^{-2}}.
$$
The same argument that led to \eqref{eq:Neumannineq} gives 
$$
\int\frac{-\Delta_{uu'}^\cN}{-\Delta_{uu'}^\cN+d^{-2}}\dd u'\leq 
\frac{-\Delta_{u}^\cN}{-\Delta_{u}^\cN+d^{-2}},
$$
which concludes the proof of the lemma.
\end{proof}
This lemma shows that the first term in (\ref{eq:Tauhat}), after
localization, leads to a gap in the small boxes.
The full kinetic energy localization is given in the next lemma. 
\begin{lm}[Small-box kinetic energy localization] \label{lm:Tuu'} Let
  $\hcT_u$ be the replacement for the kinetic energy given in (\ref{eq:Tauhat})
  in terms of the parameters $s,d,\varepsilon_T$ and the constant $b$. Let
\begin{equation}\label{eq:Tuu}
    \cT_{uu'}=\varepsilon_T(1+\pi^{-2})^{-1} (d\ell )^{-2}Q_{uu'}+
    Q_{uu'}\chi_u\hchi_{u'}
    \positivepart{\sqrt{-\Delta}
        -(ds\ell )^{-1}}\hchi_{u'}\chi_u
    Q_{uu'}.
  \end{equation}
  Our assertion is that if $(s^{-2}+d^{-3})(ds)^{-2}s^{M}\leq \delta$, then,
  for all $0<\varepsilon_T<1$,
  $$
  \hcT_u-\frac{b}{2}\ell^{-2}Q_{u}\geq\int \cT_{uu'}\dd u'.
  $$ 
  The first term in \eqref{eq:Tuu} yields a gap above zero of 
  size $\varepsilon_T(1+\pi^{-2})^{-1}
  (d\ell )^{-2}$ in the spectrum of
  $\cT_{uu'}$, which we will refer to as the Neumann gap. 
\end{lm}
\begin{proof} We again take $\ell=1$. The integral over the first term in
$\cT_{uu'}$ is bounded above by the first term in
$\hcT_u$ by Lemma~\ref{lm:hQuuav}. We concentrate on the second term in $\cT_{uu'}$.
  By a unitary transformation ($x\simpleto x/d$ and $p\simpleto pd$) of the
  result in Lemma~\ref{lm:kinloc} we obtain
  \begin{eqnarray*}
  \lefteqn{\int\hQ_{u'}\hchi_{u'}\positivepart{\sqrt{-\Delta}-(ds)^{-1}}
  \hchi_{u'}\hQ_{u'}\dd u'\leq d^{-2}F_{s}(d\sqrt{-\Delta})}&&\\
 &\leq&
  (1-\varepsilon_T)\positivepart{\sqrt{-\Delta}-\mfr12s^{-1}}
  +\varepsilon_T\positivepart{\sqrt{-\Delta}-\mfr12(ds)^{-1}}
  +C(ds)^{-2}s^{M-2},
\end{eqnarray*}
where the function $F_{s}$ is given in \eqref{eq:Fs} and we used that $ \positivepart{\sqrt{-\Delta}-\mfr12(ds)^{-1}}\leq\positivepart{\sqrt{-\Delta}-\mfr12s^{-1}}$.
Thus the proof would be complete if the operator appearing as the integrand in the second term in
$\cT_{uu'}$ would have been, instead,
$$
Q_u\chi_u\hQ_{u'}\hchi_{u'}\positivepart{\sqrt{-\Delta}-(ds)^{-1}}
  \hchi_{u'}\hQ_{u'}\chi_uQ_u.
$$
We will estimate the difference between these operators, which is
\begin{eqnarray*}
D&=&Q_u\chi_u\hQ_{u'}\hchi_{u'}\positivepart{\sqrt{-\Delta}-(ds)^{-1}}
  \hchi_{u'}\hQ_{u'}\chi_uQ_u\\
&&-
Q_{uu'}\chi_u\hchi_{u'}
    \positivepart{\sqrt{-\Delta}
        -(ds)^{-1}}\hchi_{u'}\chi_u
    Q_{uu'}\\
&=&Q_u\left(\chi_u\hQ_{u'}-Q_{uu'}\chi_u\right)\hchi_{u'}\positivepart{\sqrt{-\Delta}-(ds)^{-1}}
  \hchi_{u'}\hQ_{u'}\chi_uQ_u\\
&&+Q_{uu'}\chi_u\hchi_{u'}
    \positivepart{\sqrt{-\Delta}
        -(ds)^{-1}}\hchi_{u'}\left(\hQ_{u'}\chi_u-\chi_uQ_{uu'}\right)Q_u,
\end{eqnarray*}
where we have used that $Q_uQ_{uu'}=Q_{uu'}$.
We observe that (using Dirac notation)
\begin{eqnarray*}
\left(\chi_u\hQ_{u'}-Q_{uu'}\chi_u\right)\hchi_{u'}&=&\chi_u\htheta_{u'}\hchi_{u'}
-d^{-3}\chi_u|\htheta_{u'}\rangle\langle\htheta_{u'}|\hchi_{u'}
-\theta_u\htheta_{u'}\chi_u\hchi_{u'}\\&&{}+|B(u,u')|^{-1}|\theta_u\htheta_{u'}\rangle\langle\theta_u\htheta_{u'}|
\chi_u\hchi_{u'}\\
&=&|B(u,u')|^{-1}|\theta_u\htheta_{u'}\rangle\langle\chi_u\hchi_{u'}|
-d^{-3}|\chi_u\htheta_{u'}\rangle\langle\hchi_{u'}|,
\end{eqnarray*}
exploiting the facts that $\htheta_{u'}\hchi_{u'}=\hchi_{u'}$ and
$\theta_u\chi_u=\chi_u$. It is now simple to estimate the operator norm of $D$
\begin{eqnarray*}
  \|D\|&\leq& C\left\|\left(\chi_u\hQ_{u'}-Q_{uu'}\chi_u\right)\hchi_{u'}\positivepart{\sqrt{-\Delta}
        -(ds)^{-1}}\right\|\\&\leq&
  C|B(u,u')|^{-1/2}\left(\int_{|p|>(ds)^{-1}}|p|^4\,|\widehat{\chi_u\hchi_{u'}}(p)|^2\dd p\right)^{1/2}\\&&{}
+Cd^{-3/2}\left(\int_{|p|>(ds)^{-1}}|p|^4\,|\widehat{\hchi_{u'}}(p)|^2\dd p\right)^{1/2}\\
&\leq&C(ds)^{M-2}|B(u,u')|^{-1/2}\langle\chi_u\hchi_{u'}|(-\Delta)^M\chi_u\hchi_{u'}\rangle^{1/2}
+C(ds)^{M-2}d^{-3/2}\langle\hchi_{u'}|(-\Delta)^M\hchi_{u'}\rangle^{1/2}\\
 &\leq& C(ds)^{M-2}d^{-M}\leq C (ds)^{-2}s^M.
\end{eqnarray*}
Hence $D=Q_uDQ_u\geq -C (ds)^{-2}s^{M} Q_u$ and
\begin{eqnarray*}
  \lefteqn{\int Q_{uu'}\chi_u\hchi_{u'}
    \positivepart{\sqrt{-\Delta}
        -(ds)^{-1}}\hchi_{u'}\chi_u
    Q_{uu'}\dd u'}&&\\
  &=&\int Q_u\chi_u\hQ_{u'}\hchi_{u'}\positivepart{\sqrt{-\Delta}-(ds)^{-1}}
    \hchi_{u'}\hQ_{u'}\chi_uQ_u\dd u' -\int_{\{u' |B(u)\cap\hB(u')\not=\emptyset\}}D\, \dd u'\\
  &\leq&Q_u\chi_u\left((1-\varepsilon_T)\positivepart{\sqrt{-\Delta}-\mfr12s^{-1}}
  +\varepsilon_T\positivepart{\sqrt{-\Delta}-\mfr12(ds)^{-1}}\right)\chi_uQ_u
  \\&&+C(s^{-2}+d^{-3})(ds)^{-2}s^{M}Q_u.
\end{eqnarray*}
We have here used the fact that the volume of $\{u'
| B(u)\cap\hB(u')\not=\emptyset\}$ is bounded by $C d^{-3}$.
\end{proof}
We shall throughout the rest of the paper assume that the conditions in Lemmas~\ref{lm:Tu} and \ref{lm:Tuu'} are satisfied.
\begin{condition}\label{cond:2}
In terms of the integer $M$ appearing in the definition \eqref{zeta new and eq:defchi} of $\chi$ we have
\begin{equation}\label{eq:Condloc}
 0<\varepsilon_T<\delta,\quad 
\hbox{and}\quad
  (s^{-2}+d^{-3})(ds)^{-2}s^{M}\leq \delta ,
\end{equation}
where $\delta$ was introduced in 
Condition~\ref{cond:1} to ensure adequate smallness of relevant
parameters. This shows that we need $M\geq 8$.
\end{condition}
A stronger upper bound on $\varepsilon_T$ will be required in Condition~\ref{cond:epsilon3lower bound}.
\subsection{Localization of the Total Energy}
We define the localized Hamiltonian, $H_{u}$, in a ``large'' box $B(u)$ to be
\begin{align}
H_{u}=&\sum_{i=1}^N\hspace{-0.1 cm}\bigg (\hspace{-0.1 cm}(1-\varepsilon_0)\widehat{\cT}_{u,i}
-\rhozero \hspace{-0.1 cm}\int\limits_{}\hspace{-0.1 cm}w_{B(u)}(x_i,y)\dd y\bigg )
+\hspace{-0.3 cm}\sum_{1\leq i<j\leq N}\hspace{-0.4 cm}
w_{B(u)}(x_i,x_j)+\frac12\rhozero ^2\hspace{-0.15 cm}\iint\limits_{}\hspace{-0.1 cm}w_{B(u)}(x,y)\dx\dy, 
\label{eq:Hu}
\end{align}
where $0\leq\varepsilon_0\leq 1/2$ is a parameter to be determined later. The subscript $i$, as usual, refers to the $i^{\rm th}$ particle.  
Recall that $\hcT_u$ was defined in \eqref{eq:Tauhat} and $w_{B(u)}$
was defined in \eqref{eq:wBu}.
The Hamiltonian $H_u$ is defined as a quadratic form on 
the Bosonic Fock space $\FF_B (H_0^1(\Lambda))$.

The localized Hamiltonian in a ``small'' box $B(u,u')$ is 
\begin{eqnarray}
  H_{uu'}&=&\sum_{i=1}^N\left((1-\varepsilon_0)\cT_{uu',i}
    -\rhozero \int\limits_{}w_{B(u,u')}(x_i,y)\dy\right)
  +\sum_{1\leq i<j\leq N}
  w_{B(u,u')}(x_i,x_j)\nonumber\\&&\label{eq:Huu}
  +\frac12\rhozero ^2\iint\limits_{}w_{B(u,u')}(x,y)\dx\dy,
\end{eqnarray}
where $w_{B(u,u')}$ was defined in \eqref{eq:wBu} and 
$\cT_{uu'}$ was defined in \eqref{eq:Tuu}.
The results of Proposition~\ref{prop:vloc}, \eqref{Neumann Laplacian upper bound} and Lemmas~\ref{lm:Tu} and \ref{lm:Tuu'} can be combined to give our final localization estimate.
\begin{thm}[Main localization inequalities]\label{thm: Main localization}~\\If $M\geq 8$, Condition~\ref{cond:1} and Condition~\ref{cond:2} are satisfied, then we have for all 
  $0\leq\varepsilon_0\leq1/2$
  \begin{equation}\label{eq:HrhoHu}
    H_{\rhozero}\geq \int_{\ell\inv \Lambda '} \Big(-\varepsilon_0\Delta_{u}^{\cN}+H_u
    \Big)\dd u,\quad
    \mbox{and, for all $u\in \R^3$,}\quad
    H_u-\frac{b}{2}\ell^{-2}Q_u\geq \int_{\R^3} H_{uu'} \dd u',
  \end{equation}
where $\Lambda'=\Lambda+[-\mfr{\ell}{2},\mfr{\ell}{2}]^3$ and $H_\rho$ was introduced in \eqref{background Hamiltonian}.
\end{thm}
We introduce the notation
\begin{align}\label{eq:H_B and T_B}
H_B=\left\{\begin{array}{ll} H_u,&\hbox{if } B=B(u)\\
H_{uu'},&\hbox{if } B=B(u,u')
\end{array}\right.
\qquad \textrm{and} \qquad
\cT_B=\left\{\begin{array}{ll}(1-\varepsilon_0)\hcT_u,&\hbox{if
}B=B(u)\\
(1-\varepsilon_0)\cT_{uu'},&\hbox{if
}B=B(u,u').
\end{array}\right. 
\end{align}
The reason for the above $(1-\varepsilon_0)$-term is that we still need some kinetic energy in the end of this paper when we want to apply Lemma \ref{lm:Q3 bound}. \label{epsilon0mentioned:1}With the corresponding notations for $w_B$ in \eqref{eq:wBu} the box Hamiltonian can be written 
\begin{equation}
  H_B=\sum_{i=1}^N\left(\cT_{B,i}-\rhozero \int_{}w_B(x_i,y)\dy\right)+\sum_{1\leq
    i<j\leq N}
  w_B(x_i,x_j)+\frac12\rhozero ^2\iint\limits_{}w_B(x,y)\dx\dy.
\end{equation}

\section{Energy in a Single Box}\label{Energy in a Single Box}

In this section we will study the energy in a single box $B$, i.e.,
the ground state energy of the Hamiltonian $H_B$.  We denote by $P_B$
the orthogonal projection onto the characteristic function of $B$ and
by $Q_B$ the projection orthogonal to constant functions in $B$. In particular, $P_B+Q_B=\one_B$ is the projection onto the subspace of functions
supported on $B$. 
We define the operators 
$$
n=\sum_{i=1}^N{\one}_{B,i},\quad n_0=\sum_{i=1}^N P_{B,i},\quad n_+=\sum_{i=1}^NQ_{B,i}.
$$
Here $n$ represents the number of particles in the box $B$, $n_0$ the
number of particles in the constant function, which we will refer to as the \emph{condensate particles}, and $n_+$ the number of particles not
in the condensate, which we will refer to as the {\it excited particles}.  We have $n=n_++n_0$.

The particle number operator $n$ commutes with the box operator
$H_B$, but $n_+$ and $n_0$ do not commute with $H_B$.  In our
discussion below we may assume that $n$ is a parameter, i.e., we restrict
to eigenspaces for the operator $n$. We shall not distinguish
between the operator $n$ and its eigenvalues.

We give a simple a priori bound on $n_+$, which will be improved later.
\begin{lm}[Simple bound on the ground state energy of $H_B$ and on $n_+$]\label{lm: simmple bound on H_B and n_+}~\\
The ground state energy $E_B$ of $H_B$ satisfies
\begin{align}
    0>E_B\geq -Cn\rhozero |B|\cU_B\label{eq: simmple bound on H_B and n_+}
\end{align}
with $\cU_B$ given in (\ref{eq:UB}).
Moreover, if $B=B(u,u')$, for any $n$-particle state with expectation
value $\langle
H_B\rangle\leq
\frac12\rhozero ^2\iint_{}w_B(x,y)\dx\dy=\rhozero ^2|B|^2\cU_B$
we 
have
\begin{align}
\langle n_+\rangle\leq C(1-\varepsilon_0)^{-1}\varepsilon_T^{-1}\rhozero  a(d\ell)^2  n\max \chi_B^2.\label{n_+ is small}
\end{align}
\end{lm}
\begin{proof}
The upper bound on $E_B$ follows by using a trial state in which all
particles are in the condensate, i.e., an eigenstate of $n_+$ with eigenvalue $0$. 
If there are $n$ particles, we find that the expectation
value in this state is 
$$
\langle H_B\rangle =\begin{cases}  [-\abs{B}^{-1}\rho+\frac{1}{2}\rho^2] \iint w_B(x,y)\dx\dy&\text{for }n=1,\\
\frac12\left[(n-\rhozero |B|)^2-n\right]|B|^{-2} 
\iint\limits_{}w_B(x,y)\dx\dy & \text{for }n\geq 2.
\end{cases}
$$

We can choose $n=1$ if $\rho \abs{B}<2$ and otherwise choose $n\geq2$ such that $|n-\rhozero |B||<1$. Hence $\langle
H_B\rangle<0$ and thus the ground state energy of $H_B$ is negative.
The lower bound on $E_B$ follows immediately from \eqref{eq:wBUB1} since $\cT_B$ and $w_B$ are non-negative.

To obtain \eqref{n_+ is small}: If an $n$-particle state satisfies $\langle
H_B\rangle\leq
\frac12\rhozero ^2\iint\limits_{}w_B(x,y)\dx\dy$, we have
\begin{align*}
0\,\geq&\, \sum_{i=1}^N\big\langle \cT_{B,i}-\rhozero \int\limits_{}
w_B(x_i,y)\dy\big\rangle
\geq (1-\varepsilon_0)\varepsilon_T(1+\pi^{-2})^{-1}(d\ell)^{-2}\langle
n_+\rangle-C\rhozero  a n\max\chi_B^2,
\end{align*}
where we have used (\ref{eq:Tuu}) and that $\underset{x}{\max}\intt{}{}\omega_B(x,y)\dd y\leq Ca\max\chi_B^2$ by Lemma~\ref{lm:wBUB}.
\end{proof}
\subsection{The Negligible (Non-Quadratic) Parts of the Potential}
We treat the potential energy terms in $H_B$ according to how many
excited particles they involve. We write $\one_B=P_B+Q_B$ and we expand
and classify the terms according to the number of $Q$-factors, no-$Q$,
$1$-$Q$, \ldots, $4$-$Q$. In the following we will simply write
$P_B=P$ and $Q_B=Q$.\\
\bigskip
\vbox{\noindent{\bf no-$\cQ$ terms:} 
\begin{eqnarray}
\cQ_0&:=&-\sum_i\rhozero  P_i\int\limits_{} w_B(x_i,y)\dy P_i+
\sum_{i<j}P_iP_jw_B(x_i,x_j)P_iP_j+\frac12\rhozero ^2\iint\limits_{}w_B(x,y)\dx\dy
\nonumber\\
&\,=&\left[(n_0-\rhozero |B|)^2-n_0\right]\,\cU_B
=\left[(n-\rhozero |B|)^2-2(n-\rhozero |B|)n_++n_+^2-n_0\right]\,\cU_B,
\label{eq:noQ}
\end{eqnarray}
where we have used the notation (\ref{eq:UB}).}
\bigskip
\noindent{\bf 1-$\cQ$ terms:}
\begin{eqnarray}
\cQ_1&:=&\sum_{i,j}P_iP_jw_B(x_i,x_j)Q_iP_j-\sum_i\rhozero P_i\int w_B(x_i,y)\dy Q_i+
\hbox{h.c.}\nonumber\\
&\,=&\sum_iP_i\int w_B(x_i,y)\dy Q_i(n_0|B|^{-1}-\rhozero )+
\hbox{h.c.}\nonumber\\
&\,=&\cQ'_1+\cQ''_1, \label{eq:Q1}
\end{eqnarray}
where
\begin{equation}\label{eq:Q1'}
\cQ'_1:=(n-\rhozero |B|)|B|^{-1}\left(\sum_iP_i\int w_B(x_i,y)\dy Q_i+
\sum_iQ_i\int w_B(x_i,y)\dy P_i\right)
\end{equation}
and
\begin{equation}\label{eq:Q1''}
\cQ''_1:=-|B|^{-1}\sum_iP_i\int
w_B(x_i,y)\dy Q_i n_+-
|B|^{-1}n_+\sum_iQ_i\int
w_B(x_i,y)\dy P_i.
\end{equation}
\begin{lm}[Estimates on $\cQ_1$ ]\label{lm:Q1}
For all $\varepsilon_1',\varepsilon_1''>0$
\begin{align*}
\cQ'_1\geq -|n-\rhozero |B||(\varepsilon_1'n_0+\varepsilon'^{-1}_1Cn_+)\,\cU_B\qquad \textrm{and} \qquad
\cQ''_1\geq -(\varepsilon_1''(n_++1)n_0+C\varepsilon_1''^{-1}n_+^2)\,\cU_B.
\end{align*}
\end{lm}
\begin{proof}
We prove first the bound on $\cQ''_1$. We have 
$$
\cQ''_1=-\sqrt{n_++1}\sum_iP_i\int
w_B(x_i,y)\dy Q_i \sqrt{n_+}\,|B|^{-1}+\hbox{h.c.}
$$
since for any self-adjoint operator $A$ we get $ \sum_iP_iA_i Q_i n_+=(n_++1)\sum_iP_iA_iQ_i$ and hence
$\sum_iP_iA_i Q_i \sqrt{n_+}=\sqrt{n_++1}\sum_iP_iA_iQ_i$.

Since $w_B\geq0$, we obtain from a Schwarz inequality\label{application of Young's inequality} and
Lemma~\ref{lm:wBUB} that 
$$
\cQ''_1\geq
-\left(\varepsilon_1'' (n_++1)\summ{i}{}P_i+C\varepsilon_1''^{-1}n_+\summ{i}{}Q_i\right)\cU_B .
$$
The estimate on $\cQ'_1$ follows by applying a similar Schwarz inequality. 
\end{proof}
\medskip

\noindent{\bf 2-$\cQ$ terms:}
There are two kinds of 2-$\cQ$ terms. There are terms which contribute to the energy to the order of interest. They will primarily be treated together with the quadratic Hamiltonian later on. The remaining 2-$\cQ$ terms are negligible error terms that we will estimate here.\\
\indent The 2-$\cQ$  terms that will later appear in the quadratic Hamiltonian are
\begin{eqnarray}\label{def:Q2'}
\hspace{-0.4 cm}\cQ'_2&:=&\sum_{i,j}P_iQ_jw_B(x_i,x_j)P_jQ_i+\sum_{i<j}\left(
Q_iQ_jw_B(x_i,x_j)P_jP_i+P_jP_iw_B(x_i,x_j)Q_iQ_j\right)\!.
\end{eqnarray}
We shall however give an a priori estimate on these terms already
now. This estimate will be used in Case II in the proof of Lemma~\ref{lm: napriori}. Using the estimate $0\leq w_B\leq
CaR^{-3}\max\chi_B^2$ and the Schwarz inequality twice gives that  
\begin{eqnarray}
  \cQ'_2&\geq&-\sum_{i,j}P_iQ_jw_B(x_i,x_j)P_iQ_j-\sum_{i<j}\left(
    2Q_iQ_jw_B(x_i,x_j)Q_jQ_i+\frac12P_jP_iw_B(x_i,x_j)P_iP_j\right)
  \nonumber\\
  &\geq& -n_0|B|^{-1}\sum_iQ_i\int w_B(x_i,y)\dy Q_i
  -Cn_+^2aR^{-3}\max\chi_B^2
  -\frac12\sum_{i<j}P_jP_iw_B(x_i,x_j)P_iP_j\nonumber\\
  &\geq&-Cnn_+\,\cU_B-Cn_+^2aR^{-3}\max\chi_B^2-\frac12n^2\,\cU_B.
  \label{eq:Q2apriori}
\end{eqnarray}
In the last inequality we have used that $\summ{i}{}Q_i\int \omega_B(x_i,y)\dd y Q_i\leq C\abs{B}\cU_Bn_+$ by \eqref{eq:wBUB1} since $\summ{i}{}Q_i\int \omega_B(x_i,y)\dd y Q_i$ commutes with $n_+$.
The negligible 2-$\cQ$  terms are estimated in the same way
\begin{align}
\hspace{-0.2 cm}\cQ''_2&:=-\sum_i\rhozero Q_i\int
w_B(x_i,y)\dy Q_i+\sum_{i,j}Q_iP_jw_B(x_i,x_j)P_jQ_i\nonumber\\
&\,=(n_0-\rhozero |B|)|B|^{-1}\sum_iQ_i\int w_B(x_i,y)\dy Q_i\nonumber\\
&\,=(n-\rhozero |B|-n_+)|B|^{-1}\sum_iQ_i\int w_B(x_i,y)\dy Q_i
\geq -C\left([\,\rhozero |B|-n\,]_+n_++n_+^2\right)\cU_B.\label{eq:Q2second}
\end{align}

\bigskip

\noindent{\bf 3-$\cQ$ terms:} For all $\varepsilon_3>0$
\begin{eqnarray}
\cQ_3&:=&\sum_{i,j}P_jQ_iw_B(x_i,x_j)Q_iQ_j+\hbox{ h.c. }\nonumber\\
&\,\geq& -\sum_{i\ne
	j}\left(2\varepsilon_3^{-1}P_jQ_iw_B(x_i,x_j)Q_iP_j
+\frac{\varepsilon_3}{2}Q_jQ_iw_B(x_i,x_j)Q_iQ_j\right)\nonumber\\
&\,\geq&-C\varepsilon_3^{-1}nn_+\,\cU_B\,-\varepsilon_3\sum_{i<
	j}Q_jQ_iw_B(x_i,x_j)Q_iQ_j.\label{eq:Q3}
\end{eqnarray}
The first inequality uses a Schwarz inequality, while the second
inequality uses Lemma~\ref{lm:wBUB} and the fact that $n_0\leq n$. 
Note that the above estimates are given as lower bounds, but that we of course also could have stated them as two-sided bounds.
The last term above can be absorbed into the positive 4-$\cQ$ term if $\varepsilon_3\leq 1$.
We will do this initially, but at a later stage in the proof we have to control the first term in \eqref{eq:Q3} by choosing $\varepsilon_3\gg 1$. At that time we will have good control over the number of excited particles, $n_+$, in the large box. The second term in \eqref{eq:Q3} will then be controlled by applying Lemma \ref{lm:Q3 bound}\label{explanation: Q4 bound}, where we apply the kinetic energy term $-\varepsilon_0\Delta_{u}^{\cN}$ that we have saved in Theorem~\ref{thm: Main localization} for exactly this purpose.\\\\
\bigskip
\noindent{\bf 4-$\cQ$ term:} This is the positive term 
\begin{equation}
\cQ_4:=\sum_{i<j}Q_jQ_iw_B(x_i,x_j)Q_iQ_j\label{eq:Q4}
\end{equation}
and for a lower bound
it can be ignored or used to control other errors.
We will only need an estimate on the 4-$\cQ$ term in a large box
$B=B(u)$ as explained above. 
\subsection{The Quadratic Hamiltonian} 
We can write the box Hamiltonian as
\begin{equation}\label{eq: box decomposition}
H_B=\sum_{i=1}^N\cT_{B,i}+\cQ_0+\cQ'_1+\cQ''_1+\cQ'_2+\cQ''_2+\cQ_3+\cQ_4.
\end{equation}
We have estimated all terms except the {\it quadratic part}
$\sum_{i=1}^N\cT_{B,i}+\cQ'_2$.

We first consider the localized kinetic energy.\\
\indent For $B=B(u)$, i.e., a large box we have from \eqref{eq:Tauhat} and \eqref{eq:H_B and T_B} that 
\begin{align}\label{eq:TB1}
\cT_B={}&(1-\varepsilon_0)\varepsilon_{T} (d\ell
)^{-2}\frac{-\Delta_u^\cN}{-\Delta_u^\cN+(d\ell
	)^{-2}}+(1-\varepsilon_0)b\ell^{-2}Q+
Q\chi_B\tau_B(-\Delta)\chi_B Q
\intertext{with}
\label{eq:tauB1}
\tau_B(k^2)={}&(1-\varepsilon_0)(1-\varepsilon_{T})\positivepart{|k|-\mfr12
(s\ell )^{-1}}\;
+(1-\varepsilon_0)\varepsilon_{T}\positivepart{|k|-\mfr12
(ds\ell )^{-1}}.
\end{align}
For $B=B(u,u')$, i.e., a small box we have from \eqref{eq:Tuu} and \eqref{eq:H_B and T_B} that
\begin{equation}\label{eq:TB2}
\cT_B= (1-\varepsilon_0)\varepsilon_T(1+\pi^{-2})^{-1}(d\ell  )^{-2}Q+
Q\chi_B\tau_B(-\Delta)\chi_B Q
\end{equation}
with \label{epsilon0mentioned:2}
\begin{equation}\label{eq:tauB2}
\tau_B(k^2)=(1-\varepsilon_0)\positivepart{|k|-
(ds\ell )^{-1}}.
\end{equation}
The interesting part of $\cT_B$ is the term of form
$Q\chi_B\tau_B(-\Delta)\chi_BQ$. The other terms, which are positive,
are only used to control errors. We put them aside for the moment and define the
{\it quadratic Hamiltonian}
\begin{equation}\label{eq:Hquad}
  H_{\rm Quad}=\sum_{i=1}^{N}(Q\chi_B\tau_B(-\Delta)\chi_BQ)_i+\cQ'_2.
\end{equation}
\noindent At the end of this paper it will be useful to treat the term $\cQ'_1$
	together with $H_{\rm Quad}$.\\
To handle $H_{\rm Quad}$, we use the formalism of second quantization. 
For all $k\in\R^3$ we define the operator 
\begin{equation}\label{eq:bannihilation}
b_k=a^*_0a(Q(e^{ikx}\chi_B)),
\end{equation}
where $a(Q(e^{ikx}\chi_B))$ is the operator that annihilates an exited particle
in the state given by the function
$Q(e^{ikx}\chi_B)\in L^2(B)$, and $a_0$ is the
operator that creates a particle in the condensate. 
These two operators commute. Note that $b_k$ is a bounded operator when restricted to a subspace of finite $n$.
Its adjoint is 
\begin{equation}
\label{eq:bcreation}
b_k^*=a(Q(e^{ikx}\chi_B))^*a_0.
\end{equation}
Since $a_0^*$ commutes with $a(Q(e^{ikx}\chi_B))$, we have the commutation relations
\begin{equation}\label{eq:bkbk'com}
[b_k,b_{k'}]=0,\quad [b_k,b^*_{k'}]=a_0^*a_0
\big\langle Q(e^{ikx}\chi_B)\big|Q(e^{ik'x}\chi_B)\big\rangle-a(Q(e^{ik'x}\chi_B))^*a(Q(e^{ikx}\chi_B)),
\end{equation}
for all $k,k'\in\R^3$.
In particular, 
\begin{align}\label{eq:bkbkcom}
[b_k,b^*_{k}]&\leq a_0^*a_0\int \chi_B^2=n_0\int\chi_B^2.
\end{align}
The term $Q\chi_B\tau_B(-\Delta)\chi_BQ$ and its
second quantization can be written
\begin{eqnarray}
  Q\chi_B\tau_B(-\Delta)\chi_BQ&=&(2\pi)^{-3}\int_{\R^3}\tau_B(k^2)\left|Q(\chi_Be^{ikx})\right\rangle
  \left\langle Q(\chi_Be^{ikx})\right| \dd k\nonumber\\
  \stackrel{\rm 2^{nd}\,quant}
  {\longrightarrow}&&(2\pi)^{-3}\int_{\R^3}\tau_B(k^2)a(Q(\chi_Be^{ikx}))^*
  a(Q(\chi_Be^{ikx}))\dd k\nonumber\\
  &\geq&  (2\pi)^{-3}\int_{\R^3}\tau_B(k^2)a(Q(\chi_Be^{ikx}))^*\frac{a_0a_0^*}{n}
  a(Q(\chi_Be^{ikx}))\dd k\nonumber\\&=&(2\pi)^{-3}n^{-1}\int_{\R^3}\tau_B(k^2)b_k^*b_k\dd k.\label{eq:qchitauchiq}
\end{eqnarray}
Here we used that $b_k\psi=0$ if $\psi$ is in the condensate allowing us to assume \label{nplus assumption}that $n_+\geq 1$ such that in fact $a_0a_0\ad\leq n$. Likewise we may write
\begin{eqnarray}
\cQ'_2&=&\frac12(2\pi)^{-3}|B|^{-1}\int
\widehat{W}(k)\left(b_k^*b_k+b_{-k}^*b_{-k}+b_k^*b^*_{-k}+b_kb_{-k}\right)\dd k
\nonumber\\&&
-(2\pi)^{-3}|B|^{-1}\int\widehat{W}(k)
a(Q(\chi_Be^{ikx}))^*a(Q(\chi_Be^{ikx}))\dd k\label{eq:Q2'2nd}
\end{eqnarray}
and
\begin{eqnarray}\label{eq:Q1'2nd}
  \cQ'_1&=&(n-\rhozero |B|)(2\pi)^{-3}|B|^{-3/2}\int\widehat{W}(k)\left(\overline{\widehat{\chi}_B(k)}b_k
  +\widehat{\chi}_B(k)b_k^*\right)\dd k.
\end{eqnarray}
The last term in $\cQ'_2$ may be written
\begin{equation}\label{eq:QZQ}
(2\pi)^{-3}|B|^{-1}\int\widehat{W}(k)
a(Q(\chi_Be^{ikx}))^*a(Q(\chi_Be^{ikx}))\dd k=
\sum_{i=1}^NQ_iZ_iQ_i,
\end{equation}
where $Z$ is the operator  
with integral kernel
\begin{equation}\label{eq:Zintegralkernel}
  k_Z(x,y)=|B|^{-1}\chi_B(x)W(x-y)\chi_B(y).
\end{equation}
\begin{lm}\label{lm:Z estimate}
The operator $Z$ on $L^2(\R^3)$ 
with integral kernel \eqref{eq:Zintegralkernel}
satisfies the bound
$$
\|Z\|\leq Ca\min\{R^{-3},|B|^{-1}\}\max\chi_B^2.
$$
In particular, if $B=B(u)$, we have
\begin{align}\label{large box Z bound}
	\norm{Z}\leq Ca\abs{B}\inv.
\end{align}
\end{lm}
\begin{proof} It is clear that
$$\|Z\|\leq
|B|^{-1}\max\chi_B^2\int W\leq C|B|^{-1}\max\chi_B^2\int v_1\leq
Ca|B|^{-1}\max\chi_B^2.
$$
If we use that the Hilbert-Schmidt norm is greater than the operator norm,
$\|Z\|\leq\|Z\|_{\rm HS}$,
 we find
$$
\|Z\|\leq|B|^{-1}\hspace{-0.1 cm}\left(\iint
  \chi_B(x)^2W(x-y)^2\chi_B(y)^2\dx\dy\right)^{1/2}\hspace{-0.1 cm}
\leq |B|^{-1}\hspace{-0.05 cm}\max W\hspace{-0.15 cm}\int\chi_B^2\leq C aR^{-3}\hspace{-0.05 cm}\max\chi_B^2.
$$
\end{proof}

Combining the above lemma with \eqref{eq:Hquad},
\eqref{eq:qchitauchiq},
\eqref{eq:Q2'2nd},
and \eqref{eq:Q1'2nd} gives the following result.
\begin{lm}\label{lm:quadratic}
For all $\quadpar\in\R$ we have the following estimate
\begin{equation}
H_{\rm
    Quad}+\quadpar\cQ'_1\geq\frac12(2\pi)^{-3}\int_{\R^3} h_{\quadpar}(k) \dd k
-Cn_+a\min\{R^{-3},|B|^{-1}\}\max\chi_B^2\label{eq:HQuad inequlaity}
\end{equation}
on each $n$-particle sector (if $n=0$ then $h_\quadpar(k)=0$), where 
\begin{eqnarray}\label{def:h_sigma}
  h_{\quadpar}(k)&=&
  n^{-1}\tau_B(k^2)(b_k^*b_k+b_{-k}^*b_{-k})
  +\widehat{W}(k)|B|^{-1}(b_k^*b_k+b_{-k}^*b_{-k}+b_k^*b^*_{-k}+b_kb_{-k})\nn
  \\&&{}+
  \quadpar(n-\rhozero |B|)\widehat{W}(k)|B|^{-3/2}
  \left(\overline{\widehat{\chi}_B(k)}(b_k+b_{-k}^*)
    +\widehat{\chi}_B(k)(b_k^*+b_{-k})
  \right).
\end{eqnarray}
\end{lm}

Note that if $B$ is a large box, then the last term in \eqref{eq:HQuad inequlaity} is small compared to the Neumann gap. The same is true if $B$ is a small box with smallest side length $\lambda_1\geq \rhozero ^{-\frac{1}{3}}$, since then $Cn_+a\min\{R^{-3},|B|^{-1}\}\max\chi_B^2\leq Cn_+\rhozero a$, which, in view of Condition~\ref{cond:epsilon3lower bound} below, is smaller than the Neumann gap on the small box.\\
\indent We shall now give an estimate on $h_{\quadpar}(k)$, which is based on a simple version of Bogolubov's treatment of quadratic Hamiltonians.
This estimate requires, however, assumptions which will not be fulfilled 
in all our situations.
The following result is Theorem 6.3 in \cite{LS1} except that we state
it here a bit more generally. In the original \cite{LS1} it was required $\cA\geq\cB>0$, but this is not
needed. 
The operators $b_\pm$ can, for example, be any commuting pair of bounded operators (the case we will use here) or they can be
annihilation operators in Fock space (the original Bogolubov case).
\begin{thm}[Simple case of Bogolubov's method]\label{thm:bogolubov}\hfill\\
For arbitrary $\cA,\cB\in \R$ satisfying $-\cA<\cB\leq\cA$ and $\kappa\in\C$ we have 
the operator inequality
\begin{eqnarray*}
\cA(b^*_+\bn_++b^*_{-}\bn_{-})+\cB(b^*_+b^*_{-}+\bn_+\bn_{-})+
\kappa(b^*_++\bn_{-})+\overline{\kappa}(\bn_++b^*_{-})\\ 
\geq-\mfr{1}{2}(\cA-\sqrt{\cA^2-\cB^2})
([\bn_{+},b^*_{+}]+[\bn_{-},b^*_{-}])-\frac{2|\kappa|^2}{\cA+\cB},\hspace{1cm}
\end{eqnarray*}
where $b_\pm$ are operators (possibly unbounded) on a Hilbert space satisfying $[b_+,b_-]=0$.
\end{thm}
\begin{proof}
The proof is essentially the same as in the original \cite{LS1}. See also \cite{BB}.
\end{proof}
When applying Theorem~\ref{thm:bogolubov} with $\kappa=0$, we can replace $\cB$ by $-\cB$, even if $\cB=\cA$, without changing the lower bound. This is easily seen by replacing $b_\pm$ by $ib_\pm$. Hence
		\begin{align}\label{eq:quadratic term bound}
		\abs{ \cB(b_{+}\ad b_{-}\ad +b_+b_{-})} \leq  \cA(b^*_+\bn_++b^*_{-}b_{-})+ \mfr{1}{2}(\cA-\sqrt{\cA^2-\abs{\cB}^2})
		([ b_{+},b^*_{+}]+[b_{-},b^*_{-}]),
		\end{align}
		which we will use on page~\pageref{apply:quadratic term bound}.
When applying this theorem to estimate $h_\quadpar(k)$, we will take
$b_+=b_k$, $b_-=b_{-k}$, restricted to the appropriate $n$-particle sector,
\begin{align}\label{def: A, B and kappa}
\cA=n^{-1}\tau_B(k^2)+\widehat{W}(k)|B|^{-1},\quad
\cB=\widehat{W}(k)|B|^{-1},\quad \kappa=\quadpar(n-\rhozero |B|)\widehat{W}(k)|B|^{-3/2}
\widehat{\chi}_B(k).
\end{align}
\indent This choice of $\cA$ and $\cB$ does not necessarily satisfy the conditions in Theorem~\ref{thm:bogolubov}.
We will now give conditions for when these are satisfied. 
We first observe that $\widehat W(0)=\int W(x)\dx>0$ and thus
\begin{align}
\widehat{W}(k)=\int \cos(kx) W(x)\dx\geq \int(1-\frac12(kx)^2)W(x)\dx>0\label{W>0}
\end{align}
if $|k|<R^{-1}$ (using that $W$ has the same range as $v_R$, i.e.,
$R$).  
Hence $\cB>0$
for these values of $k$, and the conditions in the theorem are
satisfied since $\tau_B\geq0$.

To ensure the condition for $|k|\geq R^{-1}$, 
we will use that $R<\delta ds\ell$ by
Condition~\ref{cond:1}.
We may then from the definitions \eqref{eq:tauB1} and
\eqref{eq:tauB2} of $\tau_B$ assume that $\tau_B(k^2)\geq \frac12 k^2$ 
for $|k|\geq R^{-1}\geq \delta^{-1}(ds\ell)^{-1}$.
Given $C'>0$ we have for these $k$, since $|\cB|=|B|^{-1}|\widehat W(k)|\leq Ca/|B|$, that
$$
\cA\geq \frac12 n^{-1}R^{-2}-Ca|B|^{-1}>C' a|B|^{-1}
$$ if $n|B|^{-1}\leq c (a R^2)^{-1}$ for $c$ sufficiently small (depending only on $v_1$ and $C'$). In particular, for $\abs{k}\geq R^{-1}$ we then have $-\mathcal{A}<\abs{\mathcal{B}}\leq \mathcal{A}$ and are therefore allowed to use Theorem~\ref{thm:bogolubov} to bound $h_\quadpar(k)$.

This implies that $\cA$ is positive and that $\cA\geq|\cB|$. 
In this case we are therefore allowed to use Theorem~\ref{thm:bogolubov}
to bound $h_\quadpar(k)$ and in fact we may assume that $\cA+\cB\geq 2\abs{\cB}$ if $c$ is sufficiently small. When the condition $n\abs{B}\inv\leq c(aR^2)\inv$ can not be satisfied, which only happens on page \pageref{no Bogolubov bound} in Case 2 of the proof of Lemma~\ref{lm: napriori}, we use \eqref{eq:Q2apriori} instead.
\begin{lm}\label{lm:boglowerbound} There exists $c>0$ (depending only on $v_1$) such that if
	$n|B|^{-1}\leq c (a R^2)^{-1}$ then we have on each $n$-particle sector (if $n=0$ then $h_\quadpar(k)=0$) for all $k\in\R^3$
	\begin{align}
	h_\quadpar(k)\geq& -\left(n^{-1}\tau_B(k^2)+|B|^{-1}\widehat{W}
	(k)
	-\sqrt{n^{-2}\tau_B(k^2)^2
		+2n^{-1}|B|^{-1}\tau_B(k^2)\widehat{W}(k)}\right)n_0
	\int\chi_B^2\nn\\ & -\quadpar^2(n-\rhozero |B|)^2|B|^{-2}|\widehat\chi_B(k)|^2\abs{\widehat{W}(k)}\label{eq:h_lambda estimate}.
	\end{align}
\end{lm}
\begin{proof}
	As we just saw, we can choose $c>0$ such that we may use Bogolubov's
	method from Theorem~\ref{thm:bogolubov} with $\cA+\cB\geq 2 \abs{\cB}$. 
	This gives the estimate\marginpar{\tiny }
	$$2|\kappa|^2(\cA+\cB)^{-1}\leq \quadpar^2(n-\rhozero |B|)^2
	|B|^{-2} |\widehat{\chi}_B(k)|^2\abs{\widehat{W}(k)}.
	$$
	Now \eqref{eq:h_lambda estimate} follows, since we have already seen in \eqref{eq:bkbkcom} that $[b_k,b_k\ad]\leq n_0\int\chi_B^2$ and in \eqref{W>0} that $\widehat{W}(k)>0$ for $\abs{k}<R\inv$.
\end{proof}
We shall primarily use the above lemma with $\quadpar=0$. On page \pageref{applying lambda=1} we will also use it with $\quadpar=1$\label{lambda=1} on the large box. If we could replace $\abs{\widehat{W}(k)}$ by $\widehat{W}(k)$ in the last term in \eqref{eq:h_lambda estimate} an integration over $k$ gives 
\begin{align}
-\frac12(2\pi)^{-3}\quadpar^2(n-\rhozero |B|)^2|B|^{-2}\int
|\widehat\chi_B(k)|^2\widehat{W}(k)\dd k=  
-\quadpar^2(n-\rhozero |B|)^2\cU_B,\label{eq: sigmsa=1 term leading order}
\end{align}
using the notation \eqref{eq:UB}, which then exactly cancels the first positive term in \eqref{eq:noQ}.\\
Now we estimate the error in approximating the integral over the last term in \eqref{eq:h_lambda estimate} by \eqref{eq: sigmsa=1 term leading order}. We recall that $\widehat{W}(k)>0$ for $\abs{k}< R^{-1}$ by \eqref{W>0}, that $|\widehat{W}(k)|\leq C a$ and that ${\chi_u(x)=\chi(x\ell\inv-u)}$. Then we use \eqref{eq:q^rchi^2} to obtain the bound 
\begin{align}
C\quadpar^2\abs{B}^{-1}a(n-\rhozero \abs{B})^2
(R/\ell)^{2M}\int_{\abs{k}\geq R^{-1}} |k|^{2M}|\widehat{\chi}(k)|^2\dd k
\leq C\quadpar^2\abs{B}^{-1}a(n-\rhozero \abs{B})^2
(R/\ell)^{2M},\label{sigma=1 small error term}
\end{align}
where $M$ is the integer in the definition \eqref{zeta new and eq:defchi} of $\chi$.
We have by Condition~\ref{cond:1} that $R\ll \ell$ and this will be enough to
control the last term in \eqref{eq:h_lambda estimate} if $M$ is sufficiently large.
\section{A Priori Bounds on the Non-Quadratic Part of the Hamiltonian and on $n$}
\indent We shall eventually prove that the lowest energy of the box Hamiltonian $H_B$ will be achieved when the particle number $n$ is
close in an appropriate sense to $\rhozero |B|$. In this subsection we will give a much weaker a priori bound on
$n$. The main difficulty lies in treating the (possibly) rectangular small boxes $B(u,u')$ of side lengths 
$\lambda_1\leq\lambda_2\leq\lambda_3\leq d\ell$.
\begin{lm}[Estimates on the non-quadratic part of $H_B$]~\\\label{lm:HNonQuad}On each $n$-particle sector, if $B$ is a small box, we have
	\begin{align}
	H_B-H_{\rm Quad}\geq (1-\varepsilon_0)\varepsilon_T\widetilde{b}(d\ell)^{-2}n_++\off{\frac{7}{8}\abs{n-\rhozero \abs{B}}^2-C\abs{n-\rhozero \abs{B}}n_+-Cn-Cnn_+}\cU_B\label{eq:HNonQuad}
	\end{align}
and, if $B$ is a large box,
	\begin{align}
	H_B-H_{\rm Quad}\geq (1-\varepsilon_0)b\ell^{-2}n_++\off{\frac{7}{8}\abs{n-\rhozero \abs{B}}^2-C\abs{n-\rhozero \abs{B}}n_+-Cn-Cnn_+}\cU_B.\label{eq:HNonQuadbig}
	\end{align}
	Here $\widetilde{b}=(1+\pi^{-2})^{-1}$ and $b$ is the universal constants appearing in \eqref{eq:Tuu}.
\end{lm}
\begin{proof}
We use Equation \eqref{eq:noQ}, Lemma~\ref{lm:Q1} (with $\varepsilon_1'=\frac{\abs{n-\rhozero B}}{8n}$ if $n\neq 0$),
\eqref{eq:Q2second}, \eqref{eq:Q3}, \eqref{eq:Q4}, the respective Neumann gaps and the fact that $n_+\leq n$ to obtain \eqref{eq:HNonQuad} and \eqref{eq:HNonQuadbig}.
\end{proof}
The constant $\frac{7}{8}$ is of course not optimal and has been chosen for notational purposes only.
To prove the next lemma, we would like to use that $n_+$ is much smaller than $n$. This follows from Lemma~\ref{lm: simmple bound on H_B and n_+} in view of the following condition, which we henceforth assume to hold.
\begin{condition}\label{cond:epsilon3lower bound}
	We require
	\begin{align}
		(\sqrt{\rhozero a}d\ell)^2\leq\epsilonTconst{\theepsilonTconst}{\delta} \varepsilon_T,\qquad \varepsilon_T\ln\bigg (\frac{ds\ell}{R}\bigg )<\delta ds\ell \sqrt{\rho a}\qquad\text{and}\qquad\varepsilon_0\frac{a}{R}<\delta \sqrt{\rho a^3}.
	\end{align}
for $0<\delta <\delta_\Rexponentialendpoint<1$, where $\Rexponentialendpoint$ is as in Assumtion~\ref{R assumption}. This shows that we need $M\geq 13$.
\end{condition}

\begin{lm}[A priori bound on $n$]~\\\label{lm: napriori}There is a constant $C_0>0$ such that for any $n$-particle state on a small box satisfying
		$\langle \psi \vert H_B\vert \psi\rangle\leq0$, we have 
		\begin{align}\label{eq:napriori}
		n\leq C_0|B|\max\left\{\prod_{j=1}^3(\min\{\lambda_j,R\})^{-1},\rhozero \right\}=:K_B.
		\end{align}
\end{lm}
\begin{proof}
Assume $C_0\geq 2$. Then $K_B\geq C_0|B|\prod_{j=1}^3\lambda_j^{-1}=C_0\geq2$. If $n\geq K_B$, let $m$ be the integer part of $nK_B^{-1}$. 
We can then divide the particles into $m$ groups of particles consisting of $n_1,\ldots,n_m$ particles where $\sum_{j=1}^mn_j=n$ and
\begin{align}
\frac12 K_B\leq n_j\leq 2K_B,\quad j=1,\ldots,m.\label{eq: grouped particles}
\end{align}

We now use that the interaction $\omega_B$ between the particles is
non-negative and thereby get the following lower bound if we ignore the
interactions between the groups and correct for the background self-energy term
\begin{eqnarray*}
\langle \psi \vert H_B\vert \psi \rangle-\rhozero ^2|B|^2\cU_B
\geq
m\inf\Big\{\langle\psi'|H_B|\psi'\rangle-\rhozero ^2|B|^2\cU_B\Big|
\psi' \hbox{ has $n'$ particles in }B,\ \frac{1}{2}K_B\leq n'\leq 2K_B\Big\}.
\end{eqnarray*}

Our aim is to prove that if $C_0$ is large enough, then
$$
\langle\psi'|H_B|\psi'\rangle-\rhozero ^2\abs{B}^2\cU_B\geq 0
$$
if $\psi'$ has particle number $n'$ satisfying $K_B/2\leq n'\leq
2K_B$. 
We have
\begin{equation}\label{eq:n'app}
  \rhozero ^2|B|^2\leq C_0^{-2}K_B^2\leq 4C_0^{-2} n'^2.
\end{equation}
Thus we have $\abs{n'-\rhozero \abs{B}}\geq (1-\frac{2}{C_0})n'$ and, using that $n'\geq \frac{C_0}{2}$, that $n'\leq \frac{2}{C_0}n'^2$. With $n_+'=\langle \psi'\mid n_+\mid \psi'\rangle $ we obtain from Lemma~\ref{lm: simmple bound on H_B and n_+}, Condition~\ref{cond:epsilon3lower bound}, Lemma~\ref{lm:HNonQuad} and \eqref{eq:n'app} that
\begin{align}
\langle \psi'\vert H_B-H_\textrm{Quad}\vert \psi'\rangle-\rhozero ^2\abs{B}^2\cU_B
\geq{}&(1-\varepsilon_0)\varepsilon_T(1+\pi^{-2})^{-1}(d\ell )^{-2}
n_+'+\frac{3}{4}n'^2\cU_B,\label{eq:H_B-H_Quad}
\end{align}
if $C_0$ is sufficiently large and $\delta$ sufficiently small.
It remains to bound $H_{\textrm{Quad}}$. To do this, we differentiate between whether or not we are allowed to apply Bogolubov's method.
With $c>0$ being the constant in Lemma~\ref{lm:boglowerbound} which will allow us to use the Bogolubov bound, we first treat boxes where the side length and the parameter $R$ are not too small in the following sense.\\
\indent \underline{{\bf Case I:} $K_B< (c/2) |B|(a R^{2})^{-1}$}\\
Here we are allowed to use Bogolubov's method, i.e., Lemma~\ref{lm:quadratic} and Lemma~\ref{lm:boglowerbound}. Combining \eqref{h_0estimate1}, \eqref{W^3error}, \eqref{W^3estimate} and \eqref{W^2estimate}, we get
\begin{align}
\langle \psi'\vert H_{\rm{Quad}}\vert \psi '\rangle&\geq \frac{1}{2}(2\pi)^{-3}\int_{\R^3}h_0(k)\dk
-Cn_+'a\min\offf{R^{-3},\abs{B}\inv}\max \chi_B^2\nn\\
&\geq  -Cn'a(ds\ell)^{-3}\max \chi_B^2 -Cn'\frac{n'}{\abs{B}}a\frac{a}{R}\max \chi_B^2 -Cn'\frac{n'^2}{\abs{B}^2}a^3R\max\chi_B^2-Cn'\frac{a}{R^3}\max\chi_B^2\nn\\
&\geq -Cn'\frac{a}{R^3}\max\chi_B^2-Cn'\frac{K_B}{\abs{B}}a\frac{a}{R}\max \chi_B^2-Cn'\frac{K_B^2}{\abs{B}^2}a^3R\max \chi_B^2\nn\\
&\geq -Cn'\frac{a}{R^3}\max \chi_B^2.\label{eq:HQUAD'}
\end{align}
Combining \eqref{eq: grouped particles} and \eqref{eq:H_B-H_Quad} and noting that \eqref{eq:wBUB2} implies $\frac{a}{R^3}\max \chi_B^2\,\cU_B^{-1}\leq \frac{C}{C_0}K_B$ gives
\begin{align*}
\langle\psi'|H_B|\psi'\rangle-\rhozero ^2\abs{B}^2\cU_B&\geq Cn'^2\cU_B-Cn'\frac{a}{R^3}\max \chi_B^2\\
&=Cn'\cU_B(n'-C\cU_B\inv \frac{a}{R^3}\max\chi_B^2)\\
&\geq Cn'\cU_B(n'-CC_0\inv n'),
\end{align*}
which is positive if $C_0$ is sufficiently large.\\
\indent \underline{{\bf Case II:} $K_B\geq (c/2) |B|(a R^{2})^{-1}$}\\
By Condition~\ref{cond:1} and the Case II assumption we have 
\begin{align}
	\frac{K_B}{\abs{B}}=C_0\prod_{j=1}^3(\min\{\lambda_j,R\})^{-1}\geq \frac{c}{2}(aR^2)\inv.
\end{align}
Since $R,\lambda_j\leq d\ell$ we have $\abs{B}\leq Ca(d\ell)^2$ such that $\max\chi_B^2\leq C\of{\frac{a}{d\ell}}^{4(M+1)}$ by \eqref{eq:maxchi_B^2 bound}. By Lemma~\ref{lm: simmple bound on H_B and n_+} and Condition~\ref{cond:epsilon3lower bound}, we have 
$n_+\leq C\theepsilonTconst n\max\chi_B^2$. From the lower bound in \eqref{eq:wBUB2} we get 
\begin{align*}
Cn_+'^2\frac{a}{R^3}\max\chi_B^2&\leq C\epsilonTconst{\theepsilonTconst}{\delta} n'^2\frac{a}{R^3}\max\chi_B^4\\
&\leq C\epsilonTconst{\theepsilonTconst}{\delta} n'^2\abs{B}\prod_{j=1}^3(\min\{\lambda_j,R\})^{-1}\max\chi_B^2\,\cU_B\\
& \leq C\epsilonTconst{\theepsilonTconst}{\delta} n'^2\of{\frac{d\ell}{R}}^3\of{\frac{a}{d\ell}}^{4(M+1)}\cU_B.
\end{align*}
Together with equation \eqref{eq:H_B-H_Quad} and the estimate 
\begin{align*}
\langle \psi' |H_{\textrm{Quad}}|\psi'\rangle\geq \langle \psi' |\cQ_2'|\psi'\rangle\geq -\frac{5}{8}n'^2\,\cU_B-Cn_+^2 \frac{a}{R^{3}}\max \chi_B^2,
\end{align*}
which follows from \eqref{eq:Q2apriori}\label{no Bogolubov bound}, this yields
\begin{align}
\langle \psi'\vert H_B\vert \psi'\rangle-\rhozero ^2\abs{B}^2\,\cU_B\geq\frac{1}{8}n'^2\,\cU_B-Cn_+'^2 aR^{-3}\max \chi_B^2\geq 0,\nn
\end{align}
provided $\delta$ is sufficiently small.
\end{proof}
When applying the above lemma, we will assume that the box $B=B(u,u')$ has either smallest side length $\lambda_1\leq \rhozero ^{-\frac{1}{3}}$ or $\lambda_1>\rhozero ^{-\frac{1}{3}}$. Note that if $\lambda_1>\rho^{-\frac{1}{3}}$, we get $n\leq C_0\abs{B}\max \offf{R^{-3},\rhozero }$ and may apply Lemma~\ref{lm:boglowerbound}.
\subsection{A Priori Bounds on the Energy in the Small Box}\label{sec: small box}
\indent Small boxes at the boundary of the large box may be arbitrarily small. We first consider the case of boxes which are so small that Bogolubov's method can not be applied. By the lemma below these boxes only contribute to $e_0(\rhozero )$ by an amount, which is of higher order than the LHY-term.
\begin{lm}[Lower bound on the energy on small boxes with $\lambda_1\leq\rhozero^{-\frac{1}{3}}$]~\label{lm:energy for tiny boxes on the boundary}\\
For any $n$-particle state $\psi$ on a small box $B$ with smallest side length $\lambda_1\leq \rhozero ^{-\frac{1}{3}}$, we have
	\begin{align}
	\langle \psi \vert H_B\vert \psi\rangle \geq 
	-C\abs{B}\max\offf{\rhozero ,R^{-3}}\frac{a}{R^{3}}\of{\frac{\rhozero ^{-\frac{1}{3}}}{d\ell}}^{4M+2},\label{eq: HB estimate lambda_1 small}
	\end{align}
	where $M$ is the integer in the definition \eqref{zeta new and eq:defchi} of $\chi$. For all $u\in \R^3$
	\begin{align}
	\intt{\lambda_1(B(u,u'))\leq \rhozero ^{-\frac{1}{3}}}{}H_{uu'}\dd u'\geq
	-C\tinyyboxcontribution,\label{eq:smallboxenergy}
	\end{align}
	with $\tinyyboxcontribution=\abs{B}\max\offf{\rhozero ,R^{-3}}\frac{a}{R^{3}}\of{\frac{\rhozero ^{-\frac{1}{3}}}{\ell}}^{2M}$.
\end{lm}
\begin{proof}
	We use Lemma \ref{lm: simmple bound on H_B and n_+} to get the bound $\langle H_B\rangle\geq -Cn\rhozero \abs{B}\cU_B$. Since we may assume that $\langle H_B\rangle \leq 0$, we use Lemma \ref{lm: napriori}, which together with $\lambda_1\leq \rhozero ^{-\frac{1}{3}}$ gives \mbox{$n\leq C\abs{B}(\min\offf{\lambda_1,R})^{-3}$}. Using the upper bound in \eqref{eq:wBUB2} followed by \eqref{eq:maxchi_B^2 bound}, we arrive at
	\begin{align}
	\langle H_B\rangle\geq &-C\frac{\abs{B}^2}{\min\offf{\lambda_1^3,R^3}}\rhozero \frac{a}{R^3}\max\chi_B^2\nn\\
	\geq& -C\frac{\lambda_1(d\ell)^2\abs{B}}{\min\offf{\lambda_1^3,R^3}}\rhozero  \frac{a}{R^3}\of{\frac{\lambda_1}{\ell}}^{2(M+1)}.
	\end{align}
	The estimate in \eqref{eq:smallboxenergy} is obtained by integrating over $u'$ such that $B(u,u')$ has $\lambda_1\leq \rhozero ^{-\frac{1}{3}}$, which gives a volume smaller than $Cd^{-2}\frac{\rhozero ^{-\frac{1}{3}}}{d\ell}$, and using that $\lambda_1\leq \rhozero^{-\frac{1}{3}}<d\ell$.
\end{proof}

Now we turn to the case of small boxes which have smallest side length larger than $\rhozero^{-\frac{1}{3}}$ and where Bogolubov's method, i.e., Lemma~\ref{lm:boglowerbound}, therefore is applicable.
\begin{lm}[$2^{nd}$ a priori bound on $n$ for small boxes with $\lambda_1\geq\rhozero ^{-\frac{1}{3}}$]\label{L:n bound}~\\
If $B$ is a small box with $\lambda_1\geq \rhozero ^{-\frac{1}{3}}$, then there exists a constant $C_1>1$ such that for any $n$-particle state satisfying $\langle \psi \vert  H_B\vert \psi \rangle \leq 0 $ we have
	\begin{align*}
	n\leq \Cone\rhozero  \abs{B}.
	\end{align*}
\end{lm}
\begin{proof}
	We may assume that $R\leq \rhozero ^{-\frac{1}{3}}$, since otherwise the lemma follows from the first a priori bound on $n$ in Lemma~\ref{lm: napriori}.
	Hence the estimates \eqref{eq:wBUB2} and \eqref{eq:wBUB3} give 
	\begin{align}
	C\inv\frac{a}{\abs{B}}\max \chi_B^2\leq \cU_B\leq C\frac{a}{\abs{B}}\max \chi_B^2.\qquad \label{U_Border}
	\end{align}
	Assume $n\geq \Cone\rhozero  \abs{B}$ with $\Cone>1$. Note here that $\rhozero  \abs{B}\geq \rhozero  \rhozero \inv =1$ since $\lambda_1\geq \rhozero ^{-\frac{1}{3}}$ and we therefore by assumption have that $n\geq \Cone$. 
	Now Lemma~\ref{lm:quadratic}, Lemma~\ref{lm:HNonQuad} and \eqref{U_Border} imply
	\begin{align}
	\langle H_{B}\rangle &\geq Cn^2\cU_B+\frac{1}{2}(2\pi)^{-3}\int h_0(k)\dd k-Cn_+\frac{a}{\abs{B}}\max\chi_B^2\nn\\
	&\geq Cn^2\frac{a}{\abs{B}}\max\chi_B^2+\frac{1}{2}(2\pi)^{-3}\int h_0(k)\dd k,
	\end{align}
	provided $\Cone$ is sufficiently large and $\delta$ is sufficiently small.
	Combining \eqref{h_0estimate1}, \eqref{W^3error}, \eqref{W^3estimate} and \eqref{W^2estimate} in the appendix, we get
	\begin{align}
	\inttt{}{}h_0(k)\dd k
	\geq & -C\frac{n}{\abs{B}}a(ds\ell)^{-3}\int \chi_B^2 -C\frac{n^2}{\abs{B}^2}\frac{a^2}{R}\int \chi_B^2 -C\frac{n^3a^3}{\abs{B}^3}R\int\chi_B^2.
	\end{align}
	Since $n\leq C_0\abs{B}R^{-3}$ by Lemma~\ref{lm: napriori} the assumption $n\geq \Cone\rhozero  \abs{B}$ leads to
	\begin{align*}
	0\geq\langle  H_B\rangle \geq \of{Cn^2-Cn^2\rhozero \inv (ds\ell)^{-3}-Cn^2\frac{a}{R}-Cn^2\frac{a^2}{R^2}}\frac{a}{\abs{B}}\max \chi_B^2,
	\end{align*}
	which contradicts that 
	$C-C(\frac{\rhozero ^{-\frac{1}{3}}}{ds\ell})^{3}-C\frac{a}{R}-C\frac{a^2}{R^2}>0$ by Condition~\ref{cond:1} if $\delta$ is sufficiently small.\\
\end{proof}
\begin{lm}[Lower bound on the energy on small boxes with $\lambda_1\geq\rhozero^{-\frac{1}{3}}$]\label{lm: energy on not too small boxes}~\\
	If $B$ is a small box with $\lambda_1\geq \rhozero ^{-\frac{1}{3}}$ and $\psi$ is an $n$-particle state satisfying $\langle \psi \vert  H_B\vert \psi \rangle \leq 0 $, then
	\begin{align}
	\langle H_B\rangle\geq{}& C\varepsilon_T(d\ell)^{-2}n_++\frac{3}{4}\abs{n-\rhozero \abs{B}}^2\cU_B-C\rhozero a\nn\\
	&-\frac{1}{4}(2\pi)^{-3}\rhozero ^2\frac{1}{R}\inttt{}{}\frac{\widehat{v}_{1}(k)^2}{\abs{k}^2}\dd k \int \chi_B^2-C\rhozero ^2a(\rhozero a^3)^{\frac{1}{2}} \int \chi_B^2(\sqrt{\rhozero  a}ds\ell)^{-3}.\label{eq:HB not too small box}
	\end{align}
\end{lm}
\begin{proof}
	We start by estimating $H_B-H_{\textrm{Quad}}$. Recall that $n\leq C\rhozero \abs{B}$ by Lemma~\ref{L:n bound} and $\cU_B\leq C\frac{a}{\abs{B}}$ by \eqref{eq:wBUB3}. Now Lemma~\ref{lm:HNonQuad} and that that $(d\ell\sqrt{\rhozero a})^2<\theepsilonTconst \varepsilon_T $ by Condition~\ref{cond:epsilon3lower bound} yields
\begin{align}
H_B-H_{\textrm{Quad}}
\geq{}& C\varepsilon_T(d\ell)^{-2}n_++\frac{7}{8}\abs{n-\rhozero |B|}^2\cU_B-C\rhozero a-C\rhozero an_+\nn\\
\geq {}& C\varepsilon_T(d\ell)^{-2}n_++\frac{7}{8}\abs{n-\rhozero |B|}^2\cU_B-C\rhozero a.\label{eq:NonQuad bound improved2}
\end{align}
We use Lemma~\ref{lm:quadratic} with $\quadpar=0$ to estimate the quadratic part
\begin{align}
H_{\rm
	Quad}\geq\frac12(2\pi)^{-3}\int_{\R^3} h_{0}(k) \dd k
-Cn_+a\min\{R^{-3},|B|^{-1}\}\max\chi_B^2.\label{eq:HQuad inequlaity2}
\end{align}
Using the bounds \eqref{h_0estimate1}, \eqref{W^3error}, \eqref{W^3estimate} and \eqref{W^2estimate} in the appendix, we obtain
\begin{align}
\frac{1}{2}(2\pi)^{-3}\int h_0(k)\dd k\geq & -\frac{1}{4}(2\pi)^{-3}(1+C\varepsilon_0)\frac{1}{R}\int\of{\frac{n}{\abs{B}}}^2\frac{\widehat{v}_{1}(k)^2}{\abs{k}^2}\dd k\int \chi_B^2\nn\\
&-C\rhozero ^2a\frac{a}{ds\ell}\ln\labs{\frac{ds\ell}{R}}\int \chi_B^2-C(\rhozero a)^3R\int \chi_B^2\nn\\
&-C\rhozero a(ds\ell)^{-3}\int \chi_B^2.
\end{align}
The second term in \eqref{eq:HQuad inequlaity2} is smaller than $Cn_+\rhozero a$, since we assumed that $\lambda_1\geq \rhozero^{-\frac{1}{3}}$, and may therefore also be absorbed into the Neumann gap.
Instead of estimating the term $\abs{n-\rhozero \abs{B}}$, we can use that if $C_1>0$, then
\begin{align}
&-\frac{1}{4}(2\pi)^{-3}\of{\frac{n}{\abs{B}}}^2\frac{1}{R}\inttt{}{} \frac{\widehat{v}_1(k)^2}{\abs{k}^2}\dd k\int \chi_B^2+C_1\abs{n-\rhozero \abs{B}}^2\cU_B&\label{eq:n optimization2firstline}\\
\geq& -\frac{1}{4}(2\pi)^{-3}\rhozero ^2\frac{1}{R}\inttt{}{} \frac{\widehat{v}_1(k)^2}{\abs{k}^2}\dd k\int \chi_B^2-C\rhozero ^2a\frac{a^2}{R^2}\max \{\rhozero R^3,1\}\int \chi_B^2.\label{eq:n optimization2}
\end{align}
To see this we insert the following estimate, which follows from \eqref{eq:wBUB2},
\begin{align*}
\cU_B\geq C\min\offf{\frac{\rhozero \inv}{R^2},R}\frac{1}{R}\frac{a}{\abs{B}}\max\chi_B^2.
\end{align*}
The $n^2$-term in \eqref{eq:n optimization2firstline} is then positive if $\delta$ is sufficiently small and we can verify \eqref{eq:n optimization2} by optimizing over $n$ and noticing that the optimal particle number satisfies $n\leq \rhozero \abs{B}(1+C\frac{a}{R}\max \{\rhozero R^3,1\})$.
After noting that by Condition~\ref{cond:1} and Condition~\ref{cond:epsilon3lower bound}
\begin{align*}
(\sqrt{\rhozero  a}ds\ell)^{-3}\geq &\,\varepsilon_0(\rhozero a^3)^{-\frac{1}{2}}\frac{a}{R}
+(\sqrt{\rhozero  a}ds\ell)^{-1}\ln\of{\frac{ds\ell}{R}}
+\sqrt{\rhozero a}R
+(\rhozero a^3)^{-\frac{1}{2}}\frac{a^2}{R^2}\max \{\rhozero R^3,1\}
\end{align*}
the lemma follows.
\end{proof}
\section{Estimates on the Large Box}\label{sec: large box}
From now on we will only focus on the large box, where we have $\int\chi_B^2=\abs{B}=\ell^3$. 
The following lemma gives a lower bound for the operator $H_B$ on the large box and is the starting point for the bounds on $n_+$ and $\abs{n-\rhozero \abs{B}}$ on the large box.
\begin{lm}On a large box we have\label{lm: energy bound on large box}
\begin{align}
H_B\geq & -\frac{1}{4}(2\pi)^{-3}\rhozero ^2\frac{1}{R}\intt{}{} \frac{\widehat{v}_{1}(k)^2}{\abs{k}^2}\dd k\abs{B}\,+\,Cb\ell^{-2}n_+-C\rhozero ^2a\abs{B}\sqrt{\rhozero a^3}(\sqrt{\rhozero a}ds\ell)^{-3}.
\end{align}
\end{lm}
\begin{proof}
We use Lemma~\ref{thm: Main localization}, together with \eqref{Chisquared_integrals}, and sum the contribution of the small boxes. We use Lemma~\ref{lm: energy on not too small boxes} for boxes with $\lambda_1>\rhozero^{-\frac{1}{3}}$, of which we have less then $Cd^{-3}$. For the small boxes boxes with $\lambda_1\leq \rhozero^{-\frac{1}{3}}$ we use Lemma~\ref{lm:energy for tiny boxes on the boundary} and note that $\frac{\mathcal{L}}{\abs{B}}\leq (\rho a)^{\frac{5}{2}}$ by Condition~\ref{cond:1} if $M\geq 3$.
\end{proof}
States satisfying the condition below will play an important role in our following estimates.
\begin{condition}\label{energy condition}
	We require that that $B$ is a large box and that $\psi$ is a state with fixed particle number, which satisfies 
	\begin{align}\label{energy assumption}
	\state{\psi}{H_B} \leq -\frac{1}{4}(2\pi)^{-3}\rhozero ^2\frac{1}{R}\int \frac{\widehat{v}_1(k)^2}{\abs{k}^2}\dd k 
	\bigbox+C\rhozero ^2a\abs{B}\sqrt{\rhozero  a^3}\SFACTORED.
	\end{align}
\end{condition}
\begin{lm}[Control of $\langle n_+ \rangle$ in the large box]\label{lm: n+ bound Lemma}
	Assume that $\psi$ satisfies Condition~\ref{energy condition}, then
	 we have, with \label{def: SFINAL}$\SFINAL:=\rhozero a\ell^2 \SFACTORED$,
	\begin{align}
	\state{\psi}{n_+} \leq C\rhozero  \abs{B} \sqrt{\rhozero  a^3}\SFINAL.\vspace{-0.3 cm}\label{lm: n_+ bound}
	\end{align}
\end{lm}

\begin{proof}
	Simply apply Lemma~\ref{lm: energy bound on large box}.
\end{proof}
Note that Condition~\ref{conditions start large box} below in particular ensures that the second term in \eqref{energy assumption}, while being larger than the LHY-order, is of higher order than the leading order term. We therefore have that $\state{\psi}{H_B}\leq 0$ if the requirement in Lemma~\ref{lm: n+ bound Lemma} is satisfied and $\delta$ is sufficiently small.
In this section we will introduce new error terms and these will be smaller than the LHY-order.
\begin{condition}\label{conditions start large box}
	We require, with $\SFINAL$ as in Lemma~\ref{lm: n+ bound Lemma},
	\begin{align*}
	\sqrt{\rhozero a^3}\SFINAL<\delta \frac{a}{R}.
	\end{align*}
\end{condition}
\begin{lm}\label{lm: HQuadbig}%
Assume that $\psi$ satisfies Condition~\ref{energy condition}. Then $n\leq C\rhozero \abs{B}$ and\vspace{-0.2 cm}
	\begin{align}
	\state{\psi}{H_{\mathrm{Quad}}}\geq&-\frac{1}{4}(2\pi)^{-3}\of{\frac{n}{\abs{B}}}^2\frac{1}{R}\int\frac{\widehat{v}_1(k)^2}{\abs{k}^2}\dd k\bigbox-C \rhozero ^2a\sqrt{\rhozero a^3} \bigbox \Quaderror,\vspace{-0.9 cm}\label{eq:quadratic bound on the bog box2}
	\end{align}
	\vspace{-0.6 cm}where
	\begin{align*}\vspace{-0.2 cm}
	\Quaderror={}&(\sqrt{\rhozero a}s\ell)^{-3}+(\rhozero a^3)^{\frac{1}{2}}\frac{R}{a}+ \varepsilon_0(\rhozero a^3)^{-\frac{1}{2}}\frac{a}{R}\nn\\
	& +(\sqrt{\rhozero a}s\ell)\inv \ln\of{\frac{s\ell}{R}}+\varepsilon_T(\sqrt{\rhozero a}ds\ell)\inv \ln\of{\frac{ds\ell}{R}}.
	\end{align*}
\end{lm}
\begin{proof}
	We start with the bound on $n$ and assume that $n\geq C\rhozero \abs{B}$. From Lemma~\ref{lm:quadratic}, Lemma~\ref{lm:boglowerbound}, \eqref{sqrtbound1}, \eqref{h_0estimateBIG} and \eqref{W^2estimateBIG} we obtain (using $n_+\leq Cn(s\ell)^{-3}\abs{B}$ by Lemma~\ref{lm: n+ bound Lemma} and Condition~\ref{cond:1})
	\begin{align*}
	\state{\psi}{ H_{\textrm{Quad}}} \geq -Cn^2\frac{a}{R}\frac{a}{\abs{B}}-Cn(s\ell)^{-3}\abs{B}\frac{a}{\abs{B}}.
	\end{align*}
	For $\delta$ sufficiently small we have $\state{\psi}{H_B}\leq 0$ by Condition~\ref{conditions start large box} and that $n_+\leq C\sqrt{\rhozero a^3}\SFINAL n<C\delta \frac{a}{R}n$ by Lemma~\ref{lm: n+ bound Lemma} . 
	It then follows from Lemma~\ref{lm:HNonQuad} that $\state{\psi}{(H_B-H_{\textrm{Quad}})}\geq Cn^2\frac{a}{\abs{B}}$ if $\delta$ is sufficiently small and we obtain the contradiction
	\begin{align*}
	0\geq \state{\psi}{H_{B}}\geq \of{Cn^2-Cn^2\frac{a}{R}-Cn^2(s\ell)^{-3}\rhozero \inv}\frac{a}{\abs{B}}.
	\end{align*}
	Hence $n\leq C\rhozero \abs{B}$. Applying 
	\eqref{h_0estimateBIG} to \eqref{W^3estimateBIG} and \eqref{W^2estimateBIG} we obtain the lower bound on $\state{\psi}{H_{\textrm{Quad}}}$.
\end{proof}
Clearly $\Quaderror<\SFINAL$ and, in contrast to the estimate on the small box, once we choose our parameters as in Lemma~\ref{lm: Chosen parameters are allowed} we will have $\Quaderror\ll 1$.
\begin{lm}[Improved bound on $n$ on the large box]~\label{lm: n- rho B}\\
	Assume that $\psi$ is an $n$-particle state which satisfies Condition~\ref{energy condition}. Then
	\begin{align}
	\abs{n-\rhozero  \abs{B}}\leq C\rhozero  \abs{B}\of{\rhozero  a^3}^{\frac{1}{4}}{\SFINAL}^{\frac{1}{2}}.\label{eq: n- rho B}
	\end{align}
\end{lm}
\begin{proof} We use the bound on the non-quadratic part in Lemma~\ref{lm:HNonQuad} and the bound on $H_{\mathrm{Quad}}$ in Lemma~\ref{lm: HQuadbig}, the bound on $n_+$ in Lemma~\ref{lm: n+ bound Lemma}. Then we use $\Quaderror<\SFINAL$ and, analogously to \eqref{eq:n optimization2}, a part of the positive term $\abs{n-\rhozero\abs{B}}$ to control the integral corresponding to the second Born term and obtain
	\begin{align*}
	\state{\psi}{H_{B}}\geq {}&C\ell^{-2}n_++\off{\frac{7}{8}\abs{n-\rhozero \abs{B}}^2-C\abs{n-\rhozero \abs{B}}n_+-Cn-Cnn_+}\cU_B\\
	&-\frac{1}{4}(2\pi)^{-3}\of{\frac{n}{\abs{B}}}^2 \frac{1}{R}\int \frac{\widehat{v}_{1}(k)^2}{\abs{k}^2}\dd k \bigbox- C\rhozero ^2a\sqrt{\rhozero a^3}\abs{B}\Quaderror\\
	\geq {}&\off{C\abs{n-\rhozero \abs{B}}^2-C\rhozero ^2\abs{B}^2\sqrt{\rhozero a^3}\SFINAL}\frac{a}{\abs{B}}\\
	&-\frac{1}{4}(2\pi)^{-3}\rhozero ^2 \frac{1}{R}\int \frac{\widehat{v}_{1}(k)^2}{\abs{k}^2}\dd k \bigbox
	-C\rhozero ^2a\frac{a^2}{R^2}\abs{B}.
	\end{align*}
	Hence $\abs{n-\rhozero  \abs{B}}\leq C\rhozero  \abs{B}\of{\rhozero  a^3}^{\frac{1}{4}}{\SFINAL}^{\frac{1}{2}}$, by Condition~\ref{cond:1}.
\end{proof}
\noindent We will apply the following theorem, whose proof can be found in \cite{LS1}.
	\begin{thm}[Localization of large matrices]~\\\label{thm:Localizing large matrices}Suppose that $\cA$ is an $(N+1)\times (N+1)$ Hermitean matrix and let $\cA^k$, with $k=0,1,\dots ,N$, denote the matrix consisting of the $k^{\mathrm{th}}$ supra- and infra-diagonal of $\cA$. Let $\psi\in \C^{N+1}$ be a normalized vector and set $d_k=\innerp{\psi}{\cA^k\psi}$ and $\lambda=\innerp{\psi}{\cA \psi}=\summ{k=0}{N}d_k$ ($\psi$ need not be an eigenvector of $\cA$). Choose some positive integer $\Mloc\leq N+1$. Then, with $\Mloc$ fixed, there is some $n\in [0,N+1-\Mloc]$ and some normalized vector $\phi\in \C^{N+1}$ with the property that $\phi_j=0$ unless $n+1\leq j\leq n+\MM$ (i.e., $\phi$ has length $\Mloc$) and such that 
		\begin{align}
		\innerp{\phi}{\cA\phi}\leq \lambda +\frac{C}{\Mloc^2}\summ{k=1}{\Mloc-1}k^2\abs{d_k}+C\summ{k=\Mloc}{N}\abs{d_k},\label{thm:Localization of large matrices}
		\end{align}
		where $C>0$ is a universal constant. (Note that the first sum starts at $k=1$.)
	\end{thm}
	We apply the theorem in the following way. Let $\aprioripsi$ be a (normalized) $n$-particle wave function. Since the $n$-particle sector of Fock space is spanned by $n_+$-eigenfunctions, we can write $\aprioripsi=\summ{m=0}{n}c_m\aprioripsi_m$, with $\aprioripsi_m$ normalized and $n_+\aprioripsi_m=m\aprioripsi_m$ for $m\in \{0,1,\dots ,n\}$.
	Since we aim to show a lower bound on the ground state energy density it is natural, in view of Theorem~\ref{thm:main1} and Theorem~\ref{thm: Main localization}, to assume that
	\begin{align}
	\innerp{\aprioripsi}{\Big (-\varepsilon_0\Delta_{u}^{\cN}+H_B\Big )\aprioripsi}\leq & -\frac{1}{4}(2\pi)^{-3}\rhozero ^2 \frac{1}{R}\int\frac{\widehat{v}_{1}(k)^2}{\abs{k}^2}\dd k\bigbox +4\pi\rhozero ^2a\abs{B_\noell}\sqrt{\rhozero  a^3}\LHYConstant\nn\\
	= &\, 4\pi \rho^2\Big (a_2+a\sqrt{\rhozero  a^3}\LHYConstant\Big )\abs{B}.\label{energyassumption}
	\end{align}
	This lets us consider the $(n+1)\times (n+1)$ Hermitean matrix $\cA_{m,m'}=\innerp{\aprioripsi_m}{(-\varepsilon_0\Delta_{u}^{\cN}+H_B)\aprioripsi_{m'}}$ and the vector $\psi=(c_0,c_1,\dots,c_n)$.
	From the form of $H_B$ we obtain that $d_k=0$ for $k\geq 3$ and thus Theorem~\ref{thm:Localizing large matrices} yields a (normalized) state, $\locpsi$, which, for given $\MMloc$, has $n_+$-eigenvalues localized to an interval of length $\MMloc$ and an energy that satisfies the bound
\begin{align}\label{eq:psi loc bound}
\innerp{\aprioripsi}{(-\varepsilon_0\Delta_{u}^{\cN}+H_{B_\noell})\aprioripsi}\geq \innerp{\locpsi}{(-\varepsilon_0\Delta_{u}^{\cN}+H_{B_\noell})\locpsi}-C\MM^{-2}\of{\abs{d_1}+\abs{d_2}},
\end{align}	
where
\begin{align}
d_1:=\innerp{\aprioripsi}{\of{\cQ_1'+\cQ_1''+\cQ_3}\aprioripsi }\qquad\textrm{and}\qquad d_2:=\innerp{\aprioripsi}{\sum_{i<j}\left( Q_iQ_jw_B(x_i,x_j)P_jP_i+\textrm{h.c.}\right)\aprioripsi }.\label{def: d1 and d2}
\end{align}
Under the assumption in \eqref{energyassumption} we will estimate $\abs{d_1}$ and $\abs{d_2}$ in Lemma~\ref{lm:d_1 and d_2}. Adequate bounds on $\MMloc$ are collected in Condition~\ref{cond:4} below and properties of $\locpsi$ are listed in Theorem~\ref{Prop:locpsi}.
We will spend the rest of this section on establishing tools for obtaining a lower bound for $\langle \locpsi , ( -\varepsilon_0\Delta^\cN_{u}+ H_{B} )\locpsi \rangle$, which will be applied in the next section in our proof of Theorem~\ref{thm:main1}.
	\begin{lm}[Control of $d_1$ and $d_2$]~\label{lm:d_1 and d_2}\\
	Let $d_1,\,d_2$ as in \eqref{def: d1 and d2}. If $B$ is a large box and $\aprioripsi$ an $n$-particle state satisfying \eqref{energyassumption}, then
	\begin{align}
	\abs{d_1}+\abs{d_2}\leq C\rhozero ^2a\abs{B}\frac{a}{R}.
	\end{align}
\end{lm}	
	\begin{proof}
	First we estimate $\abs{d_1}$. From Lemma~\ref{lm:Q1}, equation~\eqref{eq:Q3}, which could have been stated as two-sided bounds, and that $n\leq C\rhozero \abs{B}$, we obtain by setting $\varepsilon_1'=c\frac{\abs{n-\rhozero  \abs{B}}}{\rhozero  \abs{B}}$ with $c$ sufficiently small and $\varepsilon_1''=1$
	\begin{align*}
	\abs{d_1}\leq {}& \of{n-\rhozero  \abs{B}}^2\frac{a}{\abs{B}}+C\innerpp{\aprioripsi}{\rhozero  a((1+\varepsilon_3\inv)n_++1)}{\aprioripsi}
	+\varepsilon_3\langle \aprioripsi \vert\summ{i<j}{}Q_jQ_iw_B(x_i,x_j)Q_iQ_j\vert \aprioripsi \rangle.
	\end{align*}
	We can use Lemmas~\ref{lm: n+ bound Lemma} and \ref{lm: n- rho B} together with Conditions~\ref{cond:1} and \ref{conditions start large box} to see that
	\begin{align*}
	\of{n-\rhozero  \abs{B}}^2\frac{a}{\abs{B}}+C\innerpp{\aprioripsi}{\rhozero  a((1+\varepsilon_3\inv)n_++1)}{\aprioripsi}\leq C\rhozero ^2\abs{B}a\frac{a}{R},
	\end{align*}
	as long as we choose $\varepsilon_3$ to be a constant. Further we have
	\begin{align*}
	\langle \aprioripsi \vert\summ{i<j}{}Q_jQ_iw_B(x_i,x_j)Q_iQ_j\vert \aprioripsi \rangle\leq C\rhozero ^2a\abs{B}\frac{a}{R},
	\end{align*}
	since otherwise we get by combining Lemma~\ref{lm: n+ bound Lemma} and Condition~\ref{conditions start large box} with the estimates in the proof of Lemma~\ref{lm:HNonQuad}
	\begin{align}
	\innerpp{\aprioripsi}{H_B-H_{\mathrm{Quad}}}{\aprioripsi}\geq C\abs{n-\rhozero \abs{B}}^2\frac{a}{\abs{B}}+C\rhozero ^2a\abs{B}\frac{a}{R},
	\end{align}
	which would contradict \eqref{energyassumption} in view of Lemma~\ref{lm: HQuadbig}.
	It follows that $\abs{d_1}\leq C\rhozero ^2a\abs{B}\frac{a}{R}.$\\
	\indent Now we estimate $\abs{d_2}$. Note that $d_2$ corresponds to the terms containing $b\ad b\ad$ and $bb$ appearing in the second quantization of $\cQ_2'$ in \eqref{eq:Q2'2nd}. Since $\cQ_2'$ is part of $H_{\mathrm{Quad}}$ we utilize equation \eqref{eq:quadratic term bound}\label{apply:quadratic term bound} with $\cA$ and $\cB$ as in \eqref{def: A, B and kappa} to bound $\abs{d_2}=\abs{\langle \Psi \vert \cB(b_k^\ast b_{-k}^\ast +b_kb_{-k})\vert \Psi \rangle}$. The contribution from the the commutator terms in \eqref{eq:quadratic term bound} is controlled using \eqref{eq:bkbkcom}. We have already seen in the proof of Theorem~\ref{lm: HQuadbig} that
	\begin{align*}
	&\langle \aprioripsi \vert \frac{1}{2}(2\pi)^{-3}\int \of{n\inv \tau_B(k^2)+\abs{B}\inv\widehat{W}(k)-\sqrt{n^{-2}\tau_B(k^2)+2n\inv \abs{B}\inv \tau_B(k^2)\widehat{W}(k)}}n_0\int\chi_B^2\vert \aprioripsi \rangle\\
	&\leq C\rhozero ^2a\abs{B}\frac{a}{R}.
	\end{align*}
	It is therefore left to show that
	\begin{align}\label{eq:first term of A}
	\innerpp{\aprioripsi}{\frac{1}{2}(2\pi)^{-3}n^{-1}\int \tau_B(k^2)b_k\ad b_k\dd k}{\aprioripsi}&\leq C\rhozero ^2a\abs{B}\frac{a}{R}
	\intertext{and that}
	\abs{B}\inv \innerpp{\aprioripsi}{\frac{1}{2}(2\pi)^{-3}\int \widehat{W}_B(k)b_k\ad b_k\dd k}{\aprioripsi}&\leq C\rhozero ^2a\abs{B}\frac{a}{R}.\label{eq:second term of A}
	\end{align}
	We show \eqref{eq:second term of A} first.
	As on page \pageref{nplus assumption} we may assume that $n_+\geq 1$.
	From \eqref{eq:QZQ} and Lemma~\ref{lm:Z estimate} we then see that
	\begin{align*}
	\abs{B}\inv \innerpp{\aprioripsi}{\int \widehat{W}_B(k)b_k\ad b_k\dd k}{\aprioripsi}\leq{}& \abs{B}\inv \innerpp{\aprioripsi}{\int \widehat{W}_B(k)a(Q(\chi_Be^{ikx}))^*na(Q(\chi_Be^{ikx}))\dd k}{\aprioripsi}\\
	\leq{}&C\abs{B}\inv a\innerpp{\aprioripsi}{nn_+}{\aprioripsi}\\
	\leq{} &C\rhozero^2a\abs{B}\sqrt{\rhozero a^3}\SFINAL\\
	\leq {}& C\rhozero^2 a\abs{B}\frac{a}{R},
	\end{align*}
	where we also used Lemma~\ref{lm: n+ bound Lemma}, that $n\leq C\rhozero\abs{B}$ and Condition~\ref{conditions start large box}.\\
	\indent We now show \eqref{eq:first term of A}.
	Repeating the estimate for the lower bound on $H_{\mathrm{Quad}}$ with only half 
	of the term in \eqref{eq:first term of A} included would again give a lower bound of order $\rhozero ^2\abs{B}a\frac{a}{R}$ because the second Born approximation to $a$ would be calculated wrong to this order. If \eqref{eq:first term of A} would not hold this would give that $\innerp{\aprioripsi}{H_{\mathrm{Quad}}\aprioripsi}\geq 0$, which contradicts the assumption in \eqref{energyassumption}.
\end{proof}
\begin{condition}\label{cond:4}
		We require
		\begin{align}
		\mathrm{(i)}\hspace{0.3 cm} \frac{a}{R}\MM^{-2}<C\sqrt{\rhozero a^3}\SFACTORED, \hspace{0.6 cm}
		\mathrm{(ii)}\hspace{0.3 cm} \frac{a}{R}\MM^{-2}<\delta \sqrt{\rho a^3},\hspace{0.6 cm}\mathrm{(iii)}\hspace{0.3 cm}\rhozero\abs{B}\sqrt{\rhozero a^3}\SFINAL \leq \MMloc<\delta \rho \abs{B}.\nn
		\end{align}
	\end{condition}
We note that (ii) actually is a stronger requirement than (i).\\
For convenience we collect properties of the $n_+$-localized state $\locpsi$ provided by Theorem~\ref{thm:Localizing large matrices}.
	\begin{thm}[Properties of a $n_+$-localized state $\locpsi$]~\label{Prop:locpsi}\\
If $\Psi$ is an $n$-particle state satisfying \eqref{energyassumption} and $\MMloc$ satisfies Condition~\ref{cond:4}, we obtain from Theorem~\ref{thm:Localizing large matrices} an $n$-particle state $\locpsi$ which satisfies:
		\begin{enumerate}
			\itemsep-0.5em
			\item[(i)] Condition~\ref{energy condition} is satisfied by $\locpsi$.
			\item[(ii)] $\innerp{\locpsi}{n_+\locpsi}\leq C\rhozero  \abs{B} \sqrt{\rhozero  a^3}\SFINAL$\qquad and\qquad$\innerp{\locpsi}{\abs{n-\rho\abs{B}}\locpsi}\leq C\rhozero  \abs{B}\of{\rhozero  a^3}^{\frac{1}{4}}{\SFINAL}^{\frac{1}{2}}$.
			\item[(iii)] The requirement $n\abs{B}^{-1}\leq c(aR^2)^{-1}$ in Lemma~\ref{lm:boglowerbound} holds.
			\item [(iv)] The $n_+$ eigenvalues of $\locpsi$ are localized to an interval of length $\MM$.
			\item[(v)] $\innerp{\locpsi}{n_+^2\locpsi}\leq C \MM\innerp{\locpsi}{n_+\locpsi}$.
			\item[(vi)] $\innerp{\aprioripsi}{\big ( -\varepsilon_0\Delta^\cN_{u}+ H_{B}\big )\aprioripsi}\geq \innerp{\locpsi}{\big ( -\varepsilon_0\Delta^\cN_{u}+ H_{B}\big )\locpsi}-C\delta\rho^2 a\sqrt{\rho a^3}\abs{B}$.			
		\end{enumerate}
	Here $\Delta^\cN_{u}$ denotes the Neumann Laplacian described on page~\pageref{def: Neumann Laplacian} and $\SFINAL$ is as defined in Lemma~\ref{lm: n+ bound Lemma}.
	\end{thm}
	\begin{proof}
(i): This follows from equations \eqref{energyassumption}, \eqref{eq:psi loc bound}, Lemma~\ref{lm:d_1 and d_2} and Conditions~\ref{cond:1} and \ref{cond:4}~(i).
(ii): This follows from Lemma~\ref{lm: n+ bound Lemma} and Lemma~\ref{lm: n- rho B}, which by (i) apply to the state $\locpsi$.\\
(iii): This follows from (ii) and the upper bound on $R$ in \eqref{R conditions}.\\
(iv): This property of $\locpsi$ follows from Lemma~\ref{thm:Localizing large matrices}.\\
(v): By (ii) and Condition~\ref{cond:4}~(iii) the largest $n_+$-eigenvalue of $\locpsi$ is bounded by $C\MM$.
This yields the stated estimate.\\	
(vi): Apply Lemma~\ref{lm:d_1 and d_2}, with $H_B$ replaced by $-\varepsilon_0\Delta^\cN_{u}+ H_{B}$, and Condition~\ref{cond:4} (ii) to \eqref{eq:psi loc bound}.
\end{proof}

The following lemma will be used to control the expectation of the second term in $\cQ_3$, see \eqref{eq:Q3}, in the state $\locpsi$ when $\varepsilon_3\gg 1$.
\begin{lm}\label{lm:Q3 bound} With
	\begin{align}
	h:=\sum_{i=1}-\varepsilon_0\Delta^\cN_{u,i}-\varepsilon_3\sum_{i<j}Q_jQ_iw_B(x_i,x_j)Q_iQ_j
	\end{align}
	we have
	\begin{align}
	\innerppa{\psi}{h}{\psi}\geq 0,
	\end{align}
	provided \label{epsilon0 used:3}all $n_+$-eigenvalues of $\psi$ satisfy $\varepsilon_0n_+\geq C_0\varepsilon_3\frac{a}{R}n_+^2$ with $C_0$ sufficiently large and depending only on $v_1$.
\end{lm}
\begin{proof}
	The operator $h$ acts in Fock space and commutes with
	$n_+$. In fact, $h$ only depends on the excited particles in the
	following sense. We can identify the Fock space as
	$\cF_B(L^2(B))=\cF_B(QL^2(B))\otimes\cF_B(PL^2(B))$.
	In this representation $h$ is an operator acting only on the first factor. In a 
	fixed $n_+$ subspace of this factor the operator has the form
	\begin{align*}
	h=\sum_{i=1}^{n_+}\big(-\varepsilon_0\Delta^\cN_{i}-\varepsilon_3\sum_{j=i+1}^{n_+}Q_iQ_jw_B(x_i,x_j)Q_iQ_j\big ).
	\end{align*}
	If $\psi$ is in this subspace, we have
	\begin{align*}
	\innerp{\psi}{h\psi}=\sum_{i=1}^{n_+}\Big(\innerp{\psi}{\--\varepsilon_0\Delta^\cN_{i}\psi}
	- \varepsilon_3\sum_{j=i+1}^{n_+} \innerp{\psi}{w_B(x_i,x_j)\psi}\Big).
	\end{align*}
	
	Note that a function $\phi$, orthogonal to constants, i.e.,  $\phi\in
	QL^2(B)$ satisfies the Sobolev inequality $\innerp{\phi}{-\Delta^\cN\phi}\geq C\|\phi\|_6^2$. This implies
	that if $\psi$ is normalized, then
	\begin{align*}
	&\summ{i=1}{n_+} \int \varepsilon_0 \abs{\nabla_i\psi(x_1,\dots,x_{n_+})}^2-\varepsilon_3\summ{j=i+1}{n_+}\omega_B(x_i,x_j)\abs{\psi(x_1,\dots, x_{n_+})}^2\dd x_1\cdots\dd  x_{n_+}\\
	={}&n_+\int \varepsilon_0\abs{\nabla_1\psi(x_1,\dots, x_{n_+})}^2-\varepsilon_3\mfr{(n_+-1)}{2}\omega_B(x_1,x_2)\abs{\psi(x_1,\dots, x_{n_+})}^2\dd x_1\cdots \dd x_{n_+}\\
	\geq {}&\int \of{C\varepsilon_0n_+ -\varepsilon_3\mfr{n_+^2}{2}\norm{\omega_B(\;\cdot\;,x_2)}_{\frac{3}{2}}}\of{\int\abs{\psi(x_1,\dots,x_{n_+})}^6\dd x_1}^{\frac{1}{3}}\dd x_2\cdots \dd x_{n_+}\\
	\geq {}&\of{C\varepsilon_0n_+ -C\varepsilon_3n_+^2\frac{a}{R}}\int\of{\int\abs{\psi(x_1,\dots,x_{n_+})}^6\dd x_1}^{\frac{1}{3}}\dd x_2\cdots \dd x_{n_+}\\
	\geq {}&0,
	\end{align*}
	where we have used that $\int\omega_B(x,y)^{\frac{3}{2}}\dd x\leq C\int W(x-y)^{\frac{3}{2}}\dd x\leq C\int v_R(x)^{\frac{3}{2}}\dd x\leq C\of{\frac{a}{R}}^{\frac{3}{2}}$. 
\end{proof}
\noindent By Condition~\ref{cond:4} (iii) Lemma~\ref{lm:Q3 bound} is applicable to $\locpsi$ if we introduce the below condition.
\begin{condition}\label{cond:epsilon 0 and 3}
	We require
	\begin{align*}
	\varepsilon_0\geq C\varepsilon_3\frac{a}{R}\MM.
	\end{align*}
\end{condition}
\section{Proof of Theorem~\ref{thm:main1}}\label{sec: proof of Theorem 2.1}
Our goal has been to prove Theorem~\ref{thm:main1} as we have already concluded on page~\pageref{proof of main result} that it implies the main result, Theorem~\ref{thm:main}. We now put the ingredients together to prove Theorem~\ref{thm:main1}. We start with the lower bound on $H_\rho$ in
Theorem~\ref{thm: Main localization} (here $H_B$ is $H_u$). We then focus on $n_+$-localized states, $\locpsi$, provided by Theorem~\ref{Prop:locpsi} and use the decomposition
\begin{align}
\langle \locpsi\vert -\varepsilon_0\Delta_{u}^{\cN}+H_u\vert \locpsi \rangle=\langle \locpsi\vert\Big (H_{\rm Quad}+\cQ_1'\Big )+\Big (-\varepsilon_0\Delta_{u}^{\cN}+ H_u-H_{\rm Quad}-\cQ_1'\Big )\vert \locpsi \rangle.\label{eq: final decomposition}
\end{align}
We use Lemma~\ref{lm:quadratic} and Lemma~\ref{lm:boglowerbound} to show that the first parenthesis of the r.h.s. of \eqref{eq: final decomposition} is consistent with \eqref{eq:energy} (up to a term that cancels with a corresponding positive term in the second parenthesis of \eqref{eq: final decomposition}). Using Lemma~\ref{lm:Q3 bound} we show that the remaining terms in \eqref{eq: final decomposition} are of higher order.\\ 
As a prelude to the proof of Theorem~\ref{thm:main1} we show that our parameters $d,s,\ell,\MMloc,\varepsilon_0,\varepsilon_T,\varepsilon_3$ and the integer $M$ can be chosen such that all imposed conditions are satisfied. 
We choose our parameters in such a way that the error terms appearing in the proof of Theorem~\ref{thm:main1} are easily seen to be of size $o((\rho a)^{\frac{5}{2}})$. However, our choice is certainly not optimal in that it does not give an optimal bound on our error terms. In particular, we decided not to determine an optimal value for $N$ below. Our choice for $\MM$ in Lemma~\ref{lm: Chosen parameters are allowed} and why we have $\Rexponentialendpoint<\Rexponentialendpointvalue$ in Assumption~\ref{R assumption} will be discussed on page~\pageref{discuss eta and MM}.

\begin{lm}\label{lm: Chosen parameters are allowed}For any $0<\delta<\delta_\Rexponentialendpoint<1$ Conditions~\ref{cond:1}, \ref{cond:2}, \ref{cond:epsilon3lower bound}, \ref{conditions start large box}, \ref{cond:4} and \ref{cond:epsilon 0 and 3} are satisfied if we take $N$ to be sufficiently large, $\rho$ sufficiently small, denote $X:=\max\{(\rho a^3)^{-\frac{3}{10}}\frac{a}{R},R\sqrt{\rho a}\}$ and choose	
\begin{align}
d=&\,X^\frac{4}{N},\qquad s= X^{\frac{3}{2N}},\qquad \ell= (\rho a)^{-\frac{1}{2}}X^{-\frac{2}{N}},\qquad \theepsilonTconst\varepsilon_T=(\sqrt{\rho a}d\ell)^{2}=X^{\frac{4}{N}},\\
\varepsilon_0=&\, C\varepsilon_3\frac{a}{R}\MM, \qquad \varepsilon_3=\SFINAL X^{-\frac{1}{N}}=X^{-\frac{31}{2N}},\qquad \MM=\of{\frac{R}{a}}^{\frac{1}{3}},\qquad M=20.\label{second line choices for parameters}
\end{align}
Here $M$ was introduced in \eqref{zeta new and eq:defchi} and $\SFINAL=\rhozero a\ell^2 \SFACTORED=X^{-\frac{29}{2N}}$ was defined in Lemma~\ref{lm: n+ bound Lemma}.
\end{lm}
\begin{proof}
Note that $\limit{\rho \to 0}X=0$ by Assumption~\ref{R assumption}. Since $4>\frac{3}{2}$ we have $d<\delta s$, as required in Condition~\ref{cond:1}, provided $\rho$ is sufficiently small.\\
The requirement $\delta<\delta_\Rexponentialendpoint $ is only relevant for the second item in Condition~\ref{cond:epsilon3lower bound}, which we check now.
We may assume that $\frac{R}{a}\leq (\rho a^3)^{-\frac{1}{5}}$ such that $X=(\rho a^3)^{-\frac{3}{10}}\frac{a}{R}$. By Assumption~\ref{R assumption} we have $(\rho a^3)^{-\frac{1}{5}}\leq X^{-\alpha_\Rexponentialendpoint}$ for some $\alpha_\Rexponentialendpoint>0$. Thus
\begin{align}
(\sqrt{\rho a}d\ell)^2\ln\of{\frac{ds\ell}{R}}(ds\ell \sqrt{\rho a})^{-1}=&
\sqrt{\rho a}ds^{-1}\ell\ln\of{\frac{ds\ell}{R}}\nn\\
=&X^{\frac{-2+4-\frac{3}{2}}{N}}\ln\of{X^{1+\frac{7}{2N}}(\rho a^3)^{\frac{3}{10}-\frac{1}{2}}}\nn\\
\leq&X^{\frac{1}{2N}}\ln\of{X^{1+\frac{7}{2N}-\alpha_\Rexponentialendpoint}}.\nn
\end{align}
The other conditions are checked in a similar but simpler manner.
\end{proof}
\begin{proof}[Proof of Theorem~\ref{thm:main1}]We choose our parameters as prescribed in Lemma~\ref{lm: Chosen parameters are allowed}.
From Theorem~\ref{thm: Main localization} we obtain with  $\Lambda':=\Lambda+[-\mfr{\ell}{2},\mfr{\ell}{2}]^3$ the following lower bound on the background Hamiltonian, $H_\rho$, which was defined via \eqref{background Hamiltonian}, in terms of the localized Hamiltonian, $H_u$, which was  defined in \eqref{eq:Hu},
\begin{align}H_{\rhozero}\geq \int_{\ell\inv \Lambda '} \Big(-\varepsilon_0\Delta_{u}^{\cN}+H_u
\Big)\dd u.
\end{align}
This inequality is to be understood in the sense of quadratic forms and since all operators commute with the particle number, $n$, it suffices to consider states with fixed particle number.
Since we in Theorem~\ref{thm:main1} are interested in the ground state energy density after passing to the thermodynamic limit it remains to show that
\begin{align}\label{eq: lower bound on PSI}
\langle \Psi , \Big(-\varepsilon_0\Delta_{u}^{\cN}+H_u
\Big) \Psi \rangle \ell^{-3}\geq 4\pi\rhozero ^2
\left(a_2+\frac{128}{15\sqrt{\pi}}a(\sqrt{\rho a^3}+o(\sqrt{\rho a^3}))\right),
\end{align}
for any state $\Psi$ on the box $B(u)$ with fixed particle number $n$. Note that in case $\Psi$ does not satisfy \eqref{eq: lower bound on PSI}, then $\Psi$ satisfies \eqref{energyassumption}.\\
Hence, provided $\MM$ satisfied Condition~\ref{cond:4}, we obtain from Theorem~\ref{Prop:locpsi} a $n_+$-localized state $\locpsi$ which, in particular, satisfies 
\begin{align}
\innerp{\locpsi}{n_+^2\locpsi}\leq C \MM\innerp{\locpsi}{n_+\locpsi}
\end{align}
and
\begin{align}
\innerp{\aprioripsi}{(-\varepsilon_0\Delta_{u}^{\cN}+H_{u})\aprioripsi}\geq \innerp{\locpsi}{(-\varepsilon_0\Delta_{u}^{\cN}+H_{u})\locpsi}+o(\rho^2 a\sqrt{\rho a^3}\abs{B}).
\end{align}
Thus it suffices to show that 
\begin{align}\label{eq: locpsi total bound}
\langle \locpsi , (-\varepsilon_0\Delta_{u}^{\cN}+H_u
) \locpsi \rangle \ell^{-3}\geq 4\pi\rhozero ^2
\left(a_2+\frac{128}{15\sqrt{\pi}}a(\sqrt{\rho a^3}+o(\sqrt{\rho a^3}))\right).
\end{align}
With the notation from Section~\ref{Energy in a Single Box} we decompose $H_u$ (see \eqref{eq: box decomposition}) as
\begin{align}
H_u=\Big (\sum_{i=1}^N\cT_{u,i}+\cQ'_1+\cQ'_2\Big )+\Big (\cQ_0+\cQ''_1+\cQ''_2+\cQ_3+\cQ_4\Big ).
\end{align}
Using this decomposition the estimate in \eqref{eq: locpsi total bound} is a consequence of Lemma~\ref{lm: locpsi main part}, Lemma~\ref{lm: h0 part} and Lemma~\ref{lm: locpsi other part} below.
\end{proof}
\begin{lm}\label{lm: locpsi main part}
Let $\locpsi$ be some $n$-particle state on $B(u)$ with properties as in Theorem~\ref{Prop:locpsi}, let $h_0$ be defined as in \eqref{def:h_sigma} and choose parameters as in Lemma~\ref{lm: Chosen parameters are allowed}. Then 
\begin{align}
\langle \locpsi , \Big (\sum_{i=1}^N\cT_{u,i}+\cQ'_1+\cQ'_2\Big ) \locpsi \rangle \ell^{-3}\geq \frac12(2\pi)^{-3}\int_{\R^3} h_{0}(k) \dd k
-\abs{n-\rho \ell^3}^2\frac{a}{\ell^6}+o((\rho a)^{\frac{5}{2}}).
\end{align}
\end{lm}
\begin{lm}\label{lm: h0 part}
Let $\locpsi$ be some $n$-particle state on $B(u)$ with properties as in Theorem~\ref{Prop:locpsi}, let $h_0$ be defined as in \eqref{def:h_sigma} and choose parameters as in Lemma~\ref{lm: Chosen parameters are allowed}. Then
\begin{align}
\langle \locpsi ,\frac12(2\pi)^{-3}\int_{\R^3} h_{0}(k) \dd k\, \locpsi \rangle\ell^{-3} \geq 4\pi\rhozero ^2
\left(a_2+\frac{128}{15\sqrt{\pi}}a(\sqrt{\rho a^3}+o(\sqrt{\rho a^3}))\right).\label{eq: h_0 integral}
\end{align}
\end{lm}
\begin{lm}\label{lm: locpsi other part}
Let $\locpsi$ be some $n$-particle state on $B(u)$ with properties as in Theorem~\ref{Prop:locpsi} and choose parameters as in Lemma~\ref{lm: Chosen parameters are allowed}. Then
\begin{align}
\langle \locpsi , \Big (-\varepsilon_0\Delta_{u}^{\cN}+\cQ_0+\cQ''_1+\cQ''_2+\cQ_3+\cQ_4\Big ) \locpsi \rangle \ell^{-3}\geq \abs{n-\rho \ell^3}^2\frac{a}{\ell^6}+o((\rho a)^{\frac{5}{2}}).\label{eq: estimate on the small terms large box}
\end{align}
\end{lm}
\begin{proof}[Proof of Lemma~\ref{lm: locpsi main part}]
Recall that $\mathcal{T}_B$ is defined via \eqref{eq:H_B and T_B} and \eqref{eq:Tauhat} as $(1-\varepsilon_0)\widehat{\mathcal{T}}$. We note that the second term in $\widehat{\mathcal{T}}$ is positive and that $(1-\varepsilon_0)$ times the last term is part of what we in \eqref{eq:Hquad} denoted as the quadratic Hamiltonian, $H_{\textrm{Quad}}$, using the notation in \eqref{eq:tauB1}. 
Thus we have
\begin{align}
\langle \locpsi , \Big (\sum_{i=1}^N\cT_{B,i}+\cQ'_1+\cQ'_2\Big ) \locpsi \rangle \,\ell^{-3}\geq \langle \locpsi , \Big (H_{\rm Quad}+\cQ'_1\Big ) \locpsi \rangle\ell^{-3}+C\ell^{-5}n_+.
\end{align}
We proceed and apply Lemma~\ref{lm:quadratic} with $\sigma=1$\label{applying lambda=1} to obtain
\begin{equation}
\langle \locpsi , \Big (H_{\rm Quad}+\cQ'_1\Big ) \locpsi \rangle\, \ell^{-3}\geq\frac12(2\pi)^{-3}\int_{\R^3} h_{1}(k) \dd k\,\ell^{-3}
-Cn_+\frac{a}{\abs{B}^2},\label{eq: L4.4 applied}
\end{equation}
with $h_1$ as defined in \eqref{def:h_sigma}.
By Condition~\ref{cond:1} we have $C\ell^{-5}n_+-Cn_+\frac{a}{\abs{B}^2}\geq 0.$ 
We use Lemma~\ref{lm:boglowerbound}, which is applicable by Theorem~\ref{Prop:locpsi} (iii), to estimate the integral over $h_1$ as explained on page~\pageref{lambda=1} and get
\begin{align}
\frac12(2\pi)^{-3}\int_{\R^3} h_{1}(k) \dd k\, \ell^{-3}\geq - \abs{n-\rho\abs{B}}^2\frac{a}{\abs{B}^2}+\frac12(2\pi)^{-3}\int_{\R^3} h_{0}(k) \dd k\, \ell^{-3} + o((\rho a)^{\frac{5}{2}}).
\end{align}
\end{proof}
\begin{proof}[Proof of Lemma~\ref{lm: h0 part}]
The value $4\pi\frac{128}{15\sqrt{\pi}}$ for the LHY-term is obtained by integrating over values of $k$ close to $\sqrt{\rhozero a}$ after subtracting a part related to the (negative) second Born term. This term will again be added in \eqref{add second Born term again} and the complete Born term will be obtained in \eqref{error term epsilon_0}.\\
	In other words, we prove Lemma~\ref{lm: h0 part} by showing the two estimates
	\begin{align}
	&\frac{1}{2}(2\pi)^{-3}\intt{(s\ell)\inv<\abs{k}<(d\ell)\inv}{}h_0(k)+\frac{\abs{k}^2}{2}\frac{n^2\widehat{W}(k)^2}{\abs{B}
		\tau_B(k^2)^2}\dd k\abs{B}\inv\nn\\
	&\qquad\geq {}4\pi\rho^2 \frac{128}{15\sqrt{\pi}}a\sqrt{\rho a^3}+o((\rho a)^{\frac{5}{2}})\label{error term R^2}
	\end{align}
	and
	\begin{align}
	&\frac{1}{2}(2\pi)^{-3}\bigg (-\hspace{-0.8cm}\inttt{(s\ell)\inv<\abs{k}<(d\ell)\inv}{}\hspace{-0.8 cm}\frac{\abs{k}^2}{2}\frac{n^2\widehat{W}(k)^2}{\abs{B}
		\tau_B(k^2)^2}\dd k+\hspace{-0.5 cm}\inttt{\abs{k}<(s\ell)^{-1}}{}\hspace{-0.3 cm} h_0 \dk+\hspace{-0.5 cm}\inttt{\abs{k}>(d\ell)^{-1}}{}\hspace{-0.4 cm} h_0 \dk\bigg )\abs{B}\inv \label{add second Born term again}\\
	&\qquad\geq -\frac{1}{4}(2\pi)^{-3}\hspace{-0.4 cm}\inttt{\abs{k}>(s\ell)^{-1}}{}\hspace{-0.3 cm}\abs{k}^2\frac{n^2\widehat{W}(k)^2}{\abs{B}^2\tau_B(k^2)^2}+o((\rho a)^{\frac{5}{2}})\label{eq: outer part}\\
 &\qquad\geq-\frac{1}{4}(2\pi)^{-3}\rhozero ^2\frac{1}{R}\int \frac{\widehat{v}_1(k)^2}{\abs{k}^2}\dd k
	+o((\rho  a)^{\frac{5}{2}}).\label{error term epsilon_0}
	\end{align}
{\it Proof of \eqref{error term R^2}}. By Theorem~\ref{Prop:locpsi} (iii) we may lower bound the integral over $h_0$ using Lemma~\ref{lm:boglowerbound}. We get
	\begin{align}
	&\frac{1}{2}(2\pi)^{-3}\intt{(s\ell)\inv<\abs{k}<(d\ell)\inv}{}h_0(k)+\frac{\abs{k}^2}{2}\frac{n^2\widehat{W}(k)^2}{\abs{B}
		\tau_B(k^2)^2}\dd k\abs{B}\inv\label{comp:LHY 0b}\\
	\geq{}&\frac{1}{2}(2\pi)^{-3}\intt{(s\ell)\inv<\abs{k}<(d\ell)\inv}{}\abs{k}^2\of{\sqrt{1+2\frac{n\widehat{W}(k)}{\abs{B}\tau_B(k^2)}}-1-\frac{n\widehat{W}(k)}{\abs{B}\tau_B(k^2)}+\frac{1}{2}\frac{n^2\widehat{W}(k)^2}{\abs{B}^2\tau_B(k^2)^2}}\dd k\nn\\
	\geq{}&\frac{1}{2}(2\pi)^{-3}\intt{(s\ell)\inv<\abs{k}<(d\ell)\inv}{}\abs{k}^2\of{\sqrt{1+2\frac{n\widehat{W}(k)}{\abs{B}\abs{k}^2}}-1-\frac{n\widehat{W}(k)}{\abs{B}\abs{k}^2}+\frac{1}{2}\frac{n^2\widehat{W}(k)^2}{\abs{B}^2\abs{k}^4}}\dd k\nn\\
	\geq{}&\frac{1}{2}(2\pi)^{-3}\hspace{-0.8 cm}\inttt{(s\ell)\inv<\abs{k}<(d\ell)\inv}{}\hspace{-0.8 cm}\abs{k}^2\of{\sqrt{1+2 \frac{n}{\abs{B}}\widehat{v_R}(0)\abs{k}^{-2}}-1-\frac{n}{\abs{B}} \widehat{v_R}(0)\abs{k}^{-2}+\frac{1}{2}\frac{n^2}{\abs{B}^2}\widehat{v_R}(0)^2\abs{k}^{-4}}\dd k \nn\\
	&-C\rho^2 a\sqrt{\rho a^3}X^{2+\frac{2}{N}}.
	\label{error: Born term subtraction}
	\end{align}
	In the first inequality we used \eqref{sqrtbound3} to conclude that the lower bound on $h_0(k)$ in \eqref{appendix:h_0estimate} is negative and then replace $\tau_B(k^2)$ by the larger value $\abs{k}^2$.\\
	Then we used that if $f(x)=\sqrt{1+2x}-1-x+\frac{1}{2}x^2$, then $f'(x)=(1+2x)^{-\frac{1}{2}}-1+x$ and $0\leq\frac{1}{1+x}-1+x \leq f'(x)\leq \min\offf{x,\frac{3}{2}x^2}$ for $x\geq 0$.
	To obtain the error term in \eqref{error: Born term subtraction} we recall that if $\abs{k}<R\inv$, then $\widehat{v_R}(0)-\widehat{W}(k)\leq \mfr12(kR)^2\int v_R(x)\dd x$ by \eqref{eq: W_b estimate} and \eqref{W>0}. Then we estimate
	\begin{align*}
	&\intt{(s\ell)^{-1}<\abs{k}<(\rho a)^{\frac{1}{2}}}{}\abs{k}^2\frac{\rho a}{\abs{k}^2}\rho aR^2\dd k+\intt{(\rho a)^{\frac{1}{2}}<\abs{k}<(d\ell)^{-1}}{}\abs{k}^2\big (\frac{\rho a}{\abs{k}^2}\big )^2\rho aR^2\dd k\\
	&\leq C(\rho a)^{\frac{7}{2}}R^2+C(\rho a)^{3}R^2(d\ell)^{-1} \leq C(\rho a)^{3}R^2(d\ell)^{-1}\leq C\rho^2 a\sqrt{\rho a^3}X^{2+\frac{2}{N}}.
	\end{align*}
We recall that $\abs{\frac{n}{\abs{B}}-\rhozero  }\leq C\rhozero  \of{\rhozero  a^3}^{\frac{1}{4}}{\SFINAL}^{\frac{1}{2}}$ by Theorem~\ref{Prop:locpsi} (ii) and lower bound the integral appearing in \eqref{error: Born term subtraction} by
	\begin{align}
	&\frac{1}{2}(2\pi)^{-3}\hspace{-0.5 cm}\inttt{(s\ell)\inv<\abs{k}<(d\ell)\inv}{}\hspace{-0.5 cm}\abs{k}^2\of{\sqrt{1+2\rhozero \widehat{v_R}(0)\abs{k}^{-2}}-1-\rhozero \widehat{v_R}(0)\abs{k}^{-2}+\frac{1}{2}\rhozero ^2\widehat{v_R}(0)^2\abs{k}^{-4}}\dd k\nn\\
	&\qquad-C(\rho a)^{\frac{5}{2}}(\rho a^3)^{\frac{1}{4}}\SFINAL^{\frac{1}{2}}\nn\\
	&=\frac{1}{2}(2\pi)^{-3}(\rhozero \widehat{v_R}(0))^{\frac{5}{2}}\hspace{-2.2 cm}\inttt{(s\ell)\inv(\rhozero \widehat{v_R}(0))^{-\frac{1}{2}}<\abs{k}<(d\ell)\inv(\rhozero \widehat{v_R}(0))^{-\frac{1}{2}}}{}\hspace{-2.2 cm}-\abs{k}^2-1+\abs{k}^2\sqrt{1+2\abs{k}^{-2}}+\frac{1}{2}\abs{k}^{-2}\dd k-C(\rho a)^{\frac{5}{2}}(\rho a^3)^{\frac{1}{4}}\SFINAL^{\frac{1}{2}}\nn\\
	&\geq  4\pi \rhozero ^2a\frac{128}{15\sqrt{\pi}}\sqrt{\rhozero a^3}-C(\rhozero a)^{\frac{5}{2}}[(s\ell)\inv(\rhozero a)^{-\frac{1}{2}}+d\ell (\rhozero a)^{\frac{1}{2}}+(\rho a^3)^{\frac{1}{4}}\SFINAL^{\frac{1}{2}}]\nn\\
	&=4\pi \rhozero ^2a\frac{128}{15\sqrt{\pi}}\sqrt{\rhozero a^3}-C(\rhozero a)^{\frac{5}{2}}[X^{\frac{1}{2N}}+X^{\frac{2}{N}}+(\rho a^3)^{\frac{1}{4}}X^{-\frac{29}{4N}}]\label{comp2:LHY 1a}.
	\end{align}
	Here we used the identity $\int_{\R^3}-\abs{k}^2-1+\abs{k}^2\sqrt{1+2\abs{k}^{-2}}+\frac{1}{2}\abs{k}^{-2}\dd k=\frac{32}{15}\pi\sqrt{2}$, for which we note that the integrand is bounded by $\frac{1}{2}\abs{k}^{-2}$ for small $\abs{k}$ and $C\abs{k}^{-4}$ for large $\abs{k}$. \\\\
	{\it Proof of \eqref{error term epsilon_0}}. First we show how to arrive at \eqref{eq: outer part}.
	We use \eqref{sqrtbound2} for $(d\ell)^{-1}<\abs{k}<R^{-1}$, where $\widehat{W}(k)\geq 0$, and \eqref{sqrtbound3} for $\abs{k}\geq R^{-1}$ to obtain the estimate
	\begin{align}\label{sjdhgf}
	\frac{1}{2}(2\pi)^{-3}\hspace{-0.4 cm}\inttt{\abs{k}>(d\ell)^{-1}}{}\hspace{-0.4 cm} h_{0}(k) \dd k\abs{B}^{-1} &\geq-\frac14(2\pi)^{-3}\hspace{-0.4 cm}\inttt{\abs{k}>(d\ell)^{-1}}{}\hspace{-0.4 cm}\abs{k}^2\frac{n^2\widehat{W}^2(k)}{\abs{B}^2\tau_B(k^2)^2} \dd k-C\hspace{-0.3 cm}\inttt{\abs{k}>R^{-1}}{}\hspace{-0.2 cm}\frac{n^3\abs{\widehat{W}(k)}^3}{\abs{B}^3\abs{k}^{4}} \dd k.
	\end{align}
	By \eqref{W^3estimateBIG} and Theorem~\ref{Prop:locpsi} (ii) the second term in \eqref{sjdhgf} is bounded by $C(\rho a)^{3}R$.\\
	For $\abs{k}<(s\ell)^{-1}$ Theorem~\ref{Prop:locpsi} (ii) and \eqref{h_0estimateBIG} yield
	\begin{align*}
	\inttt{\abs{k}<(s\ell)\inv}{} h_0(k)\dd k\abs{B}^{-1}\geq -C\frac{n}{\abs{B}}a(s\ell)^{-3}\geq -C\rho a(s\ell)^{-3}=-C\rho^2 a\sqrt{\rho a^3}X^{\frac{3}{2N}}.
	\end{align*}
	To arrive at \eqref{error term epsilon_0} from \eqref{eq: outer part} we use \eqref{tau estimate}, \eqref{W^2estimateBIG} and Theorem~\ref{Prop:locpsi} (ii) and obtain
	\begin{align}
	&-\frac14(2\pi)^{-3}\intt{\abs{k}>(s\ell)^{-1}}{}\abs{k}^2\frac{n^2\widehat{W}^2(k)}{\abs{B}^2\tau_B(k^2)^2} \dd k\nn\\
	\geq& -\frac14(2\pi)^{-3}\rho^2\int \frac{\widehat{v_R}(k)^2}{\abs{k}^2}\dd k
	-C\rho^2a\varepsilon_0\frac{a}{R}
	-C\rho^2a \frac{a}{s\ell}\ln\of{\frac{s\ell}{R}}\nn\\
	&-C\rho^2a\varepsilon_T\frac{a}{ds\ell}\ln\of{\frac{ds\ell}{R}}-C\rho^2a(\rho a^3)^{\frac{1}{4}}\SFINAL^{\frac{1}{2}}\frac{a}{R}\label{eq:Born term and error terms}\\
	\geq & -\frac14(2\pi)^{-3}\rho^2\int \frac{\widehat{v_R}(k)^2}{\abs{k}^2}\dd k-C\rhozero^2a \sqrt{\rhozero a^3}[X^{\frac{5}{3}-\frac{31}{2N}}+X^{\frac{1}{2N}}\ln\of{X^{1-\frac{1}{2N}}(\rho a^3)^{-\frac{1}{5}}}]\nn\\
	&-C\rhozero^2a \sqrt{\rhozero a^3}[X^{\frac{1}{2N}}\ln\of{X^{1+\frac{7}{2N}}(\rho a^3)^{-\frac{1}{5}}}+(\rho a^3)^{\frac{1}{20}}X^{-\frac{29}{4N}}].\label{eq: log error terms}
	\end{align}
The error terms in \eqref{eq: log error terms} are small since we required $\Rexponentialendpoint<\Rexponentialendpointvalue$ in Assumption~\ref{R assumption}.
\end{proof}
\begin{proof}[Proof of Lemma~\ref{lm: locpsi other part}]
Recall that Theorem~\ref{Prop:locpsi} (ii) provides the bounds
\begin{align}
\langle \locpsi, n_+\locpsi \rangle \leq C\rhozero  \abs{B} \sqrt{\rhozero  a^3}\SFINAL\qquad \textrm{and}\qquad \langle \locpsi,\abs{n-\rhozero  \abs{B}}\locpsi \rangle\leq C\rhozero  \abs{B}\of{\rhozero  a^3}^{\frac{1}{4}}{\SFINAL}^{\frac{1}{2}}.
\end{align}
With $\varepsilon_1''=(\rhozero \abs{B})^{-\frac{1}{2}}\MMloc^{\frac{1}{2}}$ and Condition~\ref{cond:4} (iii) we obtain from \eqref{eq:noQ}, Lemma~\ref{lm:Q1} and \eqref{eq:Q2second}
\begin{align}
&\innerpp{\locpsi}{\cQ_0+\cQ_1''+\cQ_2''}{\locpsi}\abs{B}\inv\nn\\
&\geq \innerpp{\locpsi}{\off{\abs{n-\rho\abs{B}}^2-C\abs{n-\rhozero \abs{B}}n_++n_+^2-n_0-\varepsilon_1''(n_++1)n_0-C{\varepsilon_1''}\inv n_+^2-Cn_+^2}\cU_B}{\locpsi}\abs{B}\inv\nn\\
&\geq  \abs{n-\rho\abs{B}}^2\frac{a}{\abs{B}^2}-C\rhozero ^2a\sqrt{\rhozero a^3}(\rhozero a^3)^{\frac{1}{4}}\SFINAL^{\frac{3}{2}}-C\rhozero^2a \sqrt{\rhozero a^3}\SFINAL(\rhozero \abs{B})^{-\frac{1}{2}}\MM^{\frac{1}{2}}-C\rho \frac{a}{\abs{B}}\nn\\
&\geq \abs{n-\rho\abs{B}}^2\frac{a}{\abs{B}^2}-C\rhozero^2a \sqrt{\rhozero a^3}\SFINAL(\rhozero \abs{B})^{-\frac{1}{2}}\MM^{\frac{1}{2}}-C\rho \frac{a}{\abs{B}}\nn\\
&\geq\abs{n-\rho\abs{B}}^2\frac{a}{\abs{B}^2}-C\rhozero^2a \sqrt{\rhozero a^3}[X^{-\frac{23}{2N}}(\rhozero a^3)^{\frac{1}{6}}+X^{\frac{6}{N}}].\label{final non-quadratic estimate}
\end{align}
By Theorem~\ref{Prop:locpsi} we have $\innerp{\locpsi}{n_+^2\locpsi}\leq C \MM\innerp{\locpsi}{n_+\locpsi}$. Thus Condition~\ref{cond:epsilon 0 and 3} ensures that Lemma~\ref{lm:Q3 bound} may be applied to the state $\locpsi$ and we obtain
	\begin{align}\label{eq:Q3 bound for large epsilon3}
\innerp{\locpsi}{(\cQ_3-\varepsilon_0\Delta_{u}^{\cN})\locpsi}\abs{B}\inv&\geq
\innerp{\locpsi}{\of{-C\varepsilon_3^{-1}nn_+\,\cU_B-\varepsilon_0\Delta_{u}^{\cN}-\varepsilon_3QQw_B(x,y)QQ}\locpsi}\abs{B}\inv\nn\\
&\geq -C\rhozero ^2a(\rhozero a^3)^{\frac{1}{2}}\SFINAL\varepsilon_3\inv=-C\rhozero ^2a(\rhozero a^3)^{\frac{1}{2}}X^{\frac{1}{N}}.
\end{align}
Finally we note that $\langle \locpsi , \cQ_4 \locpsi \rangle$ is defined to be positive in \eqref{eq:Q4}.\\
\end{proof}
Now we discuss why we have $\Rexponentialendpoint=\Rexponentialendpointvalue$ \label{discuss eta and MM}in Assumption~\ref{R assumption} and why we chose $\MM=\of{\frac{R}{a}}^{\frac{1}{3}}$ in Lemma~\ref{lm: Chosen parameters are allowed}. Combining \eqref{eq:psi loc bound} and Lemma~\ref{lm:d_1 and d_2}, respectively \eqref{eq:Born term and error terms} and \eqref{second line choices for parameters}, we obtain the error terms
\begin{align}
C\rho^2a\frac{a}{R}\MM^{-2}\qquad \textrm{and}\qquad C\rho^2a\varepsilon_0\frac{a}{R}\geq C\rho^2a\of{\frac{a}{R}}^2\MM.
\end{align}
Optimizing $C\rho^2 a\frac{a}{R}\MM^{-2}+C\rho^2 a\of{\frac{a}{R}}^2\MM$ gives $\MM=\of{\frac{R}{a}}^{\frac{1}{3}}$ and the estimate
\begin{align}
\of{C(\rho a^3)^{-\frac{1}{2}}\frac{a}{R}\MM^{-2}+C(\rho a^3)^{-\frac{1}{2}}\of{\frac{a}{R}}^2\MM}
\geq{}&C\off{\frac{a}{R}(\rho a^3)^{-\frac{3}{10}}}^{\frac{5}{3}}.\label{R lower scaling}
\end{align}
\appendix
\section{Appendix}
\renewcommand{\theequation}{\thesection.\arabic{equation}}
\setcounter{equation}{0}
Throughout this appendix we will assume that $n\geq 1$ and that the condition on $n\abs{B}\inv$ in Lemma~\ref{lm:boglowerbound} is satisfied, so that by Lemma~\ref{lm:quadratic} we have the lower bounds
\begin{align*}
H_{\rm	Quad}\geq\frac12(2\pi)^{-3}\int_{\R^3} h_{0}(k) \dd k
-Cn_+a\min\{R^{-3},|B|^{-1}\}\max\chi_B^2
\end{align*}
and
\begin{align}\label{appendix:h_0estimate}
h_0(k)\geq& -\left(n^{-1}\tau_B(k^2)+|B|^{-1}\widehat{W}
(k)
-\sqrt{n^{-2}\tau_B(k^2)^2
	+2n^{-1}|B|^{-1}\tau_B(k^2)\widehat{W}(k)}\right)n_0
\int\chi_B^2.
\end{align}
\subsection{Bounds for the Quadratic Part of $H_B$}
We estimate the operator valued integral $\frac{1}{2}(2\pi)^{-3}\int h_0(k)\dd k$ using the following facts.\\
Around $x=0$ we can write $\sqrt{1+x}=1+\frac{1}{2}x-\frac{1}{8}x^2+\frac{1}{16}x^3+O(x^4)$ yielding the bounds
\begin{align}
1+\frac{1}{2}x-Cx^2\leq \sqrt{1+x}&\leq 1+\frac{1}{2}x-\frac{1}{8}x^2+C\abs{x}^3\label{sqrtbound1}\\
\sqrt{1+x}&\geq 1+\frac{1}{2}x-\frac{1}{8}x^2\qquad (\textrm{if and only if }x\geq 0)\label{sqrtbound2}\\
1+\frac{1}{2}x-\frac{1}{8}x^2-C\abs{x}^3\leq \sqrt{1+x}&\leq 1+\frac{1}{2}x.\label{sqrtbound3}
\end{align}
If $B$ is either a small or a large box, then, for all $k\in \R^3$, the parentheses in \eqref{appendix:h_0estimate} is positive.\\
This is easy to see. If $B$ is a small box and $\abs{k}\leq(ds\ell)\inv$, then $\tau_B(k^2)=0$ and $\widehat{W}(k)>0$ by \eqref{W>0} since $(ds\ell)\inv<R\inv$. For $\abs{k}>(ds\ell)\inv$ we have $\tau_B(k^2)>0$ and apply \eqref{sqrtbound3}. For the large box the claim is proven analogously. We may therefore replace $n_0$ by $n$ in \eqref{appendix:h_0estimate} when bounding $h_0$ from below such that
\begin{align}
h_0(k)\geq& -\left(n^{-1}\tau_B(k^2)+|B|^{-1}\widehat{W}
(k)
-\sqrt{n^{-2}\tau_B(k^2)^2
	+2n^{-1}|B|^{-1}\tau_B(k^2)\widehat{W}(k)}\right)n
\int\chi_B^2.
\end{align}Provided $\tau_B(k^2)>0$ we write
\begin{align}
h_0(k)\geq& -\left(\tau_B(k^2)+n|B|^{-1}\widehat{W}
(k)
-\tau_B(k^2)\sqrt{1
	+\frac{2n\widehat{W}(k)}{|B|\tau_B(k^2)}}\right)
\int\chi_B^2.
\end{align}
\subsubsection{Estimates on the Small Box}$\tau_B(k^2)=0$ if $\abs{k}<(ds\ell)\inv$ while $\widehat{W}(k)>0$ if $\abs{k}<R\inv$. Since $\sqrt{1+x}\geq 1$ if $x\geq 0$, we have
\begin{align}
\inttt{\abs{k}<2(ds\ell)\inv}{} h_0(k)\dd k
\geq& -\inttt{\abs{k}<2(ds\ell)\inv}{}\frac{n_0}{\abs{B}}\widehat{W}(k)\dd k\int\chi_B^2
\geq -C\frac{n}{\abs{B}}a(ds\ell)^{-3}\int\chi_B^2.\label{h_0estimate1}
\end{align}
Using \eqref{sqrtbound2} for $2(ds\ell)\inv<\abs{k}<R\inv$ and \eqref{sqrtbound3} for $\abs{k}>R\inv$ gives
\begin{align}
\frac{1}{2}(2\pi)^{-3}\hspace{-0.6 cm}\inttt{\abs{k}>2(ds\ell)\inv}{} \hspace{-0.5 cm}h_0(k)\dd k
\geq\frac{1}{2}(2\pi)^{-3}\hspace{-0.6 cm}\inttt{\abs{k}>2(ds\ell)\inv}{}\hspace{-0.5 cm}-\frac{1}{2}\frac{n^2}{\abs{B}^2}\frac{\widehat{W}(k)^2}{\tau_B(k^2)}\dd k \int \chi_B^2-C\hspace{-0.3 cm}\inttt{\abs{k}>R\inv}{}\hspace{-0.2 cm}\frac{n^3}{\abs{B}^3}\frac{\abs{\widehat{W}(k)}^3}{\tau_B(k^2)^2}\dd k \int \chi_B^2.\label{W^3error}
\end{align}
Since $\tau_B=(1-\varepsilon_0)[\abs{k}-(ds\ell)\inv]_+^2\geq C\abs{k}^2$ if $\abs{k}\geq 2(ds\ell)\inv$ and $\abs{\widehat{W}(k)}\leq\widehat{W}(0)\leq  Ca$, it is easy to estimate the last term in \eqref{W^3error}
\begin{align}
&\inttt{\abs{k}>R\inv}{}\frac{n^3}{\abs{B}^3}\frac{\abs{\widehat{W}(k)}^3}{\tau_B(k^2)^2}\dd k\int \chi_B^2\leq C\inttt{\abs{k}>R\inv}{}\frac{n^3}{\abs{B}^3}\frac{a^3}{\abs{k}^4}\dd k \int \chi_B^2
\leq  C\frac{n^3}{\abs{B}^3}a^3R \int \chi_B^2.\label{W^3estimate}
\end{align}
The integral $\inttt{\abs{k}>2(ds\ell)\inv}{}\frac{\widehat{W}(k)^2}{\tau_B(k^2)}\dd k$ is related to the second Born term and we have to estimate it carefully. Recall that on the small box we $v_R(x)\leq W(x)\leq v_R(x)(1+C(\frac{R}{d\ell})^2)$ by \eqref{eq: W_s estimate}. Thus we have
\begin{align}
\norm{v_R-W}_{\frac{6}{5}}\leq C(\mfr{R}{d\ell})^2\norm{v_R}_{\frac{6}{5}}\qquad \textrm{and}\qquad \norm{v_R}_{\frac{6}{5}}= R^{-\frac{1}{2}}\norm{v_{1}}_{\frac{6}{5}}.\label{VW}
\end{align}
Since $v_R(x)=\frac{1}{R^3} v_{1}(\frac{x}{R})$, we have $\widehat{v_R}(k)=\widehat{v_{1}}(Rk)$ and that 
\begin{align}
\intt{\abs{k}\geq 2(ds\ell)\inv}{}\frac{\widehat{v_R}(k)^2}{\abs{k}^2}\dd k =\frac{1}{R}\intt{\abs{k}\geq 2(ds\ell)\inv R}{}\frac{\widehat{v}_{1}(k)^2}{\abs{k}^2}\dd k. \label{v_r k^{-1}}
\end{align}
For $f\in L^{\frac{6}{5}}(\R^{3})$ a real space representation (see \cite{lieb2001analysis} Cor 5.10) followed by an application of the Hardy-Littlewood-Sobolev inequality gives
\begin{align}
\int \abs{\widehat{f}(k)}^2\abs{k}^{-2}\dd k\leq C \norm{f}_{\frac{6}{5}}^2.\label{eq:6/5 norm}
\end{align}
On the small box we have for $\abs{k}>2(ds\ell)\inv$
\begin{align}
\tau_B(k^2)\inv \leq (1+C\varepsilon_0)\abs{k}^{-2}+C(ds\ell)\inv\abs{k}^{-3}.\label{tausmallinverse}
\end{align}
We use the Cauchy-Schwarz inequality \eqref{VW} and \eqref{eq:6/5 norm}
to obtain the estimate
\begin{align}
\inttt{\abs{k}>2(ds\ell)\inv }{}\hspace{-0.3 cm}\frac{\widehat{W}(k)^2}{\abs{k}^2}\dd k
\leq {}& \intt{\abs{k}>2(ds\ell)\inv}{}\hspace{-0.3 cm}\frac{\widehat{v}_R(k)^2}{\abs{k}^2}\dd k\nn\\
&+2 \of{\intt{\abs{k}>2(ds\ell)\inv }{}\hspace{-0.3 cm}\frac{\abs{\widehat{W}(k)-\widehat{v}_R(k)}^2}{\abs{k}^2}\dd k}^{\frac{1}{2}}\of{\intt{\abs{k}>2(ds\ell)\inv }{}\hspace{-0.3 cm}\frac{\widehat{v}_R(k)^2}{\abs{k}^2}\dd k}^{\frac{1}{2}}\nn\\
&+ \intt{\abs{k}>2(ds\ell)\inv }{}\frac{\abs{\widehat{W}(k)-\widehat{v}_R(k)}^2}{\abs{k}^2}\dd k\nn\\
\leq {}& \frac{1}{R}\intt{\abs{k}>2(ds\ell)\inv R}{}\frac{\widehat{v}_{1}(k)^2}{\abs{k}^2}\dd k+C\norm{W-v_R}_{\frac{6}{5}}\norm{v_R}_\frac{6}{5}+ C\norm{W-v_R}_{\frac{6}{5}}^2\nn\\
\leq {}& \frac{1}{R}\intt{\abs{k}>2(ds\ell)\inv R}{}\hspace{-0.3 cm}\frac{\widehat{v}_{1}(k)^2}{\abs{k}^2}\dd k+ C\of{\frac{R}{d\ell}}^2\frac{a^2}{R}.\label{apr: sec Born term2}
\end{align}Using $\abs{\widehat{W}(k)}\leq Ca$, gives
\begin{align}
\hspace{-1.5 cm}\intt{\abs{k}>2(ds\ell)\inv }{}\widehat{W}\of{k}^2\abs{k}^{-3} \dd k ={}&\intt{2(ds\ell)\inv < \abs{k} < R\inv }{}\widehat{W}(k)^2\abs{k}^{-3}\dd k+ \intt{ \abs{k}>R\inv }{}\widehat{W}\of{k}^2\abs{k}^{-3} \dd k \nn\\
\leq {}& C\intt{2(ds\ell)\inv R< \abs{k} < 1 }{}\widehat{W}(0)^2\abs{k}^{-3}\dd k+ R\intt{\abs{k}>2(ds\ell)\inv }{}\widehat{W}\of{k}^2\abs{k}^{-2} \dd k \nn\\
\leq {}&Ca^2\ln \labs{ \frac{ds\ell}{R}}.\label{W^correction}
\end{align}

Combining \eqref{tausmallinverse}, \eqref{apr: sec Born term2} and \eqref{W^correction}, we arrive at
\begin{align}
\hspace{-2 cm}\intt{\abs{k}>2(ds\ell)\inv}{}\frac{\widehat{W}(k)^2}{\tau_B(k^2)}\dd k
\leq {}& (1+C\varepsilon_0)\intt{\abs{k}>2(ds\ell)\inv }{}\frac{\widehat{W}(k)^2}{\abs{k}^2}\dd k+C(ds\ell)\inv \intt{\abs{k}>2(ds\ell)\inv }{}\frac{\widehat{W}\of{k}^2}{\abs{k}^{3}} \dd k \nn\\
\leq {}&(1+C\varepsilon_0)\frac{1}{R}\intt{}{}\frac{\widehat{v}_{1}(k)^2}{\abs{k}^2}\dd k
+Ca\frac{a}{ds\ell}\ln\labs{\frac{ds\ell}{R}}.\label{W^2estimate}
\end{align}
In the last inequality we used that the second term in \eqref{apr: sec Born term2} and the positive term\\
\mbox{$(1+C\varepsilon_0)\frac{1}{R}\intt{\abs{k}<2(ds\ell)\inv R}{}\frac{\widehat{v}_{1}(k)^2}{\abs{k}^2}\dd k$} by Condition~\ref{cond:1} are bounded by the last term in \eqref{W^2estimate}.
\subsubsection{Estimates on the Large Box}
Recall that ${\tau_B(k^2)=(1-\varepsilon_0)(1-\varepsilon_{T})\positivepart{|k|-\mfr12
		(s\ell )^{-1}}+(1-\varepsilon_0)\varepsilon_{T}\hspace{-0.05 cm}\positivepart{|k|-\mfr12
		(ds\ell )^{-1}}}$ and\\
${\int\chi_B^2=\abs{B}}$ on the large box. Hence, as in \eqref{h_0estimate1},
\begin{align}\label{h_0estimateBIG}
\inttt{\abs{k}<(s\ell)\inv}{} h_0(k)\dd k
\geq& -\inttt{\abs{k}<(s\ell)\inv}{}\frac{n_0}{\abs{B}}\widehat{W}(k)\dd k\int\chi_B^2
\geq -C\frac{n}{\abs{B}}a(s\ell)^{-3}\bigbox.
\end{align}
Analogously to \eqref{W^3error} we have
\begin{align}\label{W^3errorBIG}
\frac{1}{2}(2\pi)^{-3}\hspace{-0.3 cm}\inttt{\abs{k}>(s\ell)\inv}{}\hspace{-0.3 cm} h_0(k)\dd k
\geq\frac{1}{2}(2\pi)^{-3}\hspace{-0.3 cm}\inttt{\abs{k}>(s\ell)\inv}{}\hspace{-0.3 cm}-\frac{1}{2}\frac{n^2}{\abs{B}^2}\frac{\widehat{W}(k)^2}{\tau_B(k^2)}\dd k \bigbox-C\hspace{-0.2 cm}\inttt{\abs{k}>R\inv}{}\hspace{-0.1 cm}\frac{n^3}{\abs{B}^3}\frac{\abs{\widehat{W}(k)}^3}{\tau_B(k^2)^2}\dd k \bigbox.
\end{align}
The last term in \eqref{W^3errorBIG} is estimated in the same way as the last term in \eqref{W^3error}
\begin{align}
&\inttt{\abs{k}>R\inv}{}\frac{n^3}{\abs{B}^3}\frac{\abs{\widehat{W}(k)}^3}{\tau_B(k^2)^2}\dd k\int \chi_B^2\leq C\inttt{\abs{k}>R\inv}{}\frac{n^3}{\abs{B}^3}\frac{a^3}{\abs{k}^4}\dd k \bigbox
\leq  C\frac{n^3}{\abs{B}^3}a^3R \bigbox.\label{W^3estimateBIG}
\end{align}
By \eqref{eq: W_b estimate} we have $v_R(x)\leq W(x)\leq v_R(x)(1+C(\frac{R}{\ell})^2)$ and hence, similarly to \eqref{apr: sec Born term2},
\begin{align}
\intt{\abs{k}>(s\ell)\inv }{}\frac{\widehat{W}(k)^2}{\abs{k}^2}\dd k
\leq & \intt{\abs{k}>(s\ell)\inv}{}\frac{\widehat{v}_{R}(k)^2}{\abs{k}^2}\dd k+
C\frac{a^2R}{\ell^2}.\label{apr: sec Born term2BIG}
\end{align}
Similarly to \eqref{W^correction} we have
\begin{align}
\intt{\abs{k}>(s\ell)\inv }{}\widehat{W}\of{k}^2\abs{k}^{-3} \dd k \leq Ca^2\ln \labs{ \frac{s\ell}{R}}.\label{W^correctionBIG}
\end{align}
Since 
\begin{align}\label{tau estimate}
\tau_B(k^2)\inv\leq\begin{cases}
\of{1+C\varepsilon_0+C\varepsilon_T}\abs{k}^{-2}+C(s\ell)\inv\abs{k}^{-3} & \textrm{for }(s\ell)\inv<\abs{k}<(ds\ell)\inv\\
\of{1+C\varepsilon_0}\abs{k}^{-2}+C\of{(s\ell)\inv+\varepsilon_T(ds\ell)\inv}\abs{k}^{-3} & \textrm{for }\abs{k}\geq(ds\ell)\inv,
\end{cases}
\end{align}
we have, similar to \eqref{W^2estimate},
\begin{align}
\inttt{\abs{k}>(s\ell)\inv}{}\frac{\widehat{W}(k)^2}{\tau_B(k^2)}\dd k
\leq{}&(1+C\varepsilon_0+C\varepsilon_T)\intt{(s\ell)\inv<\abs{k}<(ds\ell)\inv }{}\hspace{-0.3cm }\frac{\widehat{W}(k)^2}{\abs{k}^2}+C(s\ell)\inv\frac{\widehat{W}(k)^2}{\abs{k}^3}\dd k\nn\\
&+\of{1+C\varepsilon_0}\intt{\abs{k}>(ds\ell)\inv}{}\hspace{-0.3 cm}\frac{\widehat{W}(k)^2}{\abs{k}^2}+C\of{(s\ell)\inv+\varepsilon_T(ds\ell)\inv} \frac{\widehat{W}(k)^2}{\abs{k}^3}\dd k\nn\\
\leq {}& \intt{\abs{k}>(s\ell)\inv}{}\frac{\widehat{v_R}(k)^2}{\abs{k}^2}\dd k +C\varepsilon_0a\frac{a}{R}+ C\frac{a^2R}{\ell^2}+C\epsilon_Ta^2(ds\ell)\inv\nn\\
& +C(s\ell)\inv a^2\ln\of{\frac{s\ell}{R}}
+C\varepsilon_T(ds\ell)\inv a^2\ln\of{\frac{ds\ell}{R}}\nn\\
\leq {}& \intt{}{}\frac{\widehat{v_R}(k)^2}{\abs{k}^2}\dd k +C\varepsilon_0a\frac{a}{R}
+C(s\ell)\inv a^2\ln\of{\frac{s\ell}{R}}\nn\\
&+C\varepsilon_T(ds\ell)\inv a^2\ln\of{\frac{ds\ell}{R}}.\label{W^2estimateBIG}
\end{align}
In the last inequality we used that $C\frac{aR}{\ell^2}a$ and the positive term $\intt{\abs{k}<(s\ell)\inv }{}\frac{\widehat{v}_{R}(k)^2}{\abs{k}^2}\dd k$ by Condition~\ref{cond:1} are bounded by the second to last term in \eqref{W^2estimateBIG}.
\subsection*{Acknowledgments}%
This project started out many years ago from a fruitful discussion with E. H. Lieb, who the authors are very grateful to. The authors acknowledge support by ERC Advanced grant 321029 and by VILLUM
FONDEN via the QMATH Centre of Excellence (Grant No. 10059). 

\end{document}